\documentclass[sn-mathphys-num]{sn-jnl}

\usepackage{graphicx}%
\usepackage{multirow}%
\usepackage{amsmath,amssymb,amsfonts}%
\usepackage{amsthm}%
\usepackage{mathrsfs}%
\usepackage[title]{appendix}%
\usepackage{xcolor}%
\usepackage{textcomp}%
\usepackage{manyfoot}%
\usepackage{booktabs}%
\usepackage{algorithm}%
\usepackage{algorithmicx}%
\usepackage{algpseudocode}%
\usepackage{listings}%
\usepackage{mathtools}

\usepackage{indentfirst} 

\usepackage{tikz}
\usetikzlibrary{calc} 
\usepackage{ytableau} 
\usepackage[bottom]{footmisc} 
\usepackage{bbm} 
\usetikzlibrary{arrows.meta, decorations.pathreplacing} 
\usepackage{subcaption}
\usepackage{empheq} 
\usepackage{multirow} 
\usepackage{thm-restate} 

\newtheoremstyle{mydefinition}
  {0.75em} {0.75em} {} {} {\bfseries} {.} {1em} {}
\theoremstyle{mydefinition}
\newtheorem{definition}{Definition}

\theoremstyle{mydefinition}
\newtheorem{example}{Example}

\theoremstyle{remark}

\newtheoremstyle{mytheorem}
  {0.75em} {0.75em} {\itshape} {} {\bfseries} {.} {1em} {}
\theoremstyle{mytheorem}
\newtheorem{theorem}{Theorem}

\theoremstyle{mytheorem}
\newtheorem{lemma}{Lemma}

\theoremstyle{mytheorem}
\newtheorem{corollary}{Corollary}

\theoremstyle{mytheorem}
\newtheorem{proposition}{Proposition}

\usepackage{soulutf8} 

\usepackage{definitions}

\definecolor{brass}{rgb}{0.71, 0.65, 0.26}
\definecolor{blue-green}{rgb}{0.0, 0.87, 0.87}
\definecolor{blue-green-v1}{rgb}{0.0, 0.77, 0.9}
\definecolor{blue-green-v2}{rgb}{0.0, 0.9, 0.77}

\newlength{\YTboxsize}
\newcommand{\halfgraybox}{%
  \tikz[baseline=($(box.south)!0.5!(box.center)$)]{%
    \node[inner sep=0pt, minimum size=\YTboxsize] (box) {};
    \fill[gray!50] (box.south west) rectangle ($(box.south east)!0.5!(box.north east)$);
  }%
}

\newcommand{\dashedbox}{%
  \tikz[baseline=($(box.south)!0.5!(box.center)$)]{%
    \node[inner sep=0pt, minimum size=\YTboxsize] (box) {};
    \draw[dashed,gray] (box.south west) rectangle (box.north east);
  }%
}

\newcommand{\halfgraydashedbox}{%
  \tikz[baseline=($(box.south)!0.5!(box.center)$)]{%
    \node[inner sep=0pt, minimum size=\YTboxsize] (box) {};
    \fill[gray] (box.north west) rectangle ($(box.south east)!0.5!(box.north east)$);
    \draw ($(box.south west)!0.5!(box.north west)$) -- (box.north west) -- (box.north east) -- ($(box.south east)!0.5!(box.north east)$) -- cycle;
    \draw[dashed,gray] ($(box.south west)!0.5!(box.north west)$) -- (box.south west) -- (box.south east) -- ($(box.south east)!0.5!(box.north east)$);
  }%
}

\newcommand{\foldarrowvertical}[2][]{%
  \tikz[baseline=(C.base),#1]{%
    \coordinate (C) at (0,0);
    \draw[->,>=stealth,thick] (0,-0.6) to[out=80,in=-80] (0,0.6);
    \if\relax\detokenize{#2}\relax\else
      \node[above, font=\footnotesize] at (0,0.6) {#2};
    \fi
  }%
}

\begin{document}

\title[Extendibility of Brauer states]{Extendibility of Brauer states}


\author[1,2]{\fnm{Adrian} \sur{Solymos}}\email{solymos.adrian@wigner.hun-ren.hu}

\author[3]{\fnm{Dávid} \sur{Jakab}}\email{david.jakab@duke.edu}

\author[1,4,5,6]{\fnm{Zoltán} \sur{Zimborás}}\email{zoltan.zimboras@helsinki.fi}

\affil[1]{\small \orgdiv{Quantum Computing and Quantum Information Research Group}, \orgname{HUN-REN Wigner Research Centre for Physics}, \orgaddress{\street{Konkoly–Thege Miklós út 29-33}, \city{Budapest}, \postcode{1525}, \country{Hungary}}}

\affil[2]{\small \orgdiv{Department of Physics of Complex Systems}, \orgname{E\"otv\"os Lor\'and University}, \orgaddress{\street{Pázmány Péter sétány 1/A}, \city{Budapest}, \postcode{1117}, \country{Hungary}}}

\affil[3]{\small \orgdiv{Pratt School of Engineering}, \orgname{Duke University}, \orgaddress{\street{305 Teer Engineering Building Box 90271}, \city{Durham}, \postcode{27708}, \state{North Carolina}, \country{USA}}}

\affil[4]{\small \orgdiv{Department of Programming Languages and Compilers}, \orgname{E\"otv\"os Lor\'and University}, \orgaddress{\street{Pázmány Péter sétány 1/C}, \city{Budapest}, \postcode{1117}, \country{Hungary}}}

\affil[5]{\small \orgname{Algorithmiq Ltd}, \orgaddress{\street{Kanavakatu 3C}, \city{Helsinki}, \postcode{00160}, \country{Finland}}}

\affil[6]{\small \orgname{Department of Physics, University of Helsinki}, \orgaddress{\street{Gustaf Hällströmin katu 2}, \city{Helsinki}, \postcode{00014}, \country{Finland}}}


\abstract{
    We investigate the extendibility problem for Brauer states, focusing on the symmetric two-sided extendibility and the de~Finetti extendibility. By employing the representation theory of the unitary and orthogonal groups, we provide a general recipe for determining the set of $(n,m)$-extendible and $n$-de~Finetti-extendible Brauer states. From the concrete form of the commutant of the diagonal action of the orthogonal group, we explicitly determine the set of parameters for which the Brauer states are \mbox{$(1,2)$-,} $(1,3)$- and $(2,2)$-extendible in any dimension $d$ and find that Brauer states extend with a non-trivial trade-off in $n$ and $m$. Using the same recipe we also provide an estimate of the set of $(1,m)$-extendible Brauer states for any $m$ and dimension $d$. Finally, using the branching rules from $\U(d)$ to $\Ort(d)$, we obtain the set of $n$-de~Finetti-extendible Brauer states in any dimension, and also analytically describe the $n\to\infty$ limiting shape which turns out not to be a polygon for odd dimensions.
}


\maketitle

\section{Introduction}\label{sec:intro}

State extension problems have played a significant role in quantum information theory from its earliest stages, with various versions found in the literature today \cite{werner1989application, terhal2003symmetric, doherty2014entanglement, christandl2004squashed, allerstorfer2023monogamy}. In general, the idea in every state extension scheme is to share one (or more) given state(s) between more parties than they originally belonged to in such a way that the original state(s) appear(s) in the shared version. More concretely, in the particular case of the so-called \textit{two-sided extendibility}, a bipartite state of Alice and Bob is said to be $(n,m)$-extendible if there exists a state shared between $n$ number of Alices and $m$ number of Bobs such that the reduced state of each pair is the original one \cite{johnson2013compatible}.

Unlike classical states, quantum states cannot necessarily be arbitrarily extended to multiple parties in such a way that the two-particle reduced states are all identical to the original state. In fact, only separable states can be arbitrarily extended in this manner \cite{fannes1988symmetric,raggio1989quantum, schumacher2002conjecture, doherty2004complete}. On the other end of the spectrum, pure states that are entangled are not extendible at all. These ideas are behind the \textit{monogamy of entanglement}, which states, in loose terms, that if two systems are highly entangled, they cannot be entangled with any other system \cite{terhal2004entanglement,koashi2004monogamy}.

As a consequence of these properties, extension problems have significance in entanglement and non-locality theory \cite{werner1989application, werner1990remarks, doherty2002distinguishing,terhal2003symmetric, christandl2004squashed,doherty2004complete}. The degree of extendibility, or extendibility number of a state, describes the maximal number of parties the state can be extended to. It can be thought of as an entanglement measure as it is 1-LOCC monotone \cite{nowakowski2016symmetric}. It is also related to the family of measures called unextendible entanglement \cite{wang2024quantifying}, or the squashed entanglement for which it serves as a lower bound \cite{brandao2011faithful, li2018squashed}.

Applications beyond entanglement theory that build on the idea of state extensions include quantum data hiding \cite{brandao2011faithful} and quantum key distribution \cite{moroder2006one,myhr2009symmetric,khatri2016symmetric}. A complete resource theory of extendibility has also been worked out \cite{kaur2019extendibility}.

However, concrete calculations of extendibility numbers have only been done for a few types of states \cite{Ranade2009symmetric, lami2019extendibility, krumnow2024extendibility,negari2025extendibility}. As is often the case, it has been useful to make the problem more tractable by focusing on states with symmetry, such as Werner, isotropic, or Brauer states \cite{werner1989states, keyl2002fundamentals, horodecki1999reduction}. These are families of states described by one or two parameters, invariant to some unitary or orthogonal actions on their Hilbert space. The extension problem for these has been first generally examined by Johnson and Viola in \cite{johnson2013compatible}. Since then, there have been many advancements in this field, e.g.~by the authors of the present work in \cite{jakab2022extendibility} and by Allerstorfer et al. in \cite{allerstorfer2023monogamy}.

Other than being the test subject for many schemes in entanglement theory \cite{vollbrecth2001entanglement}, states with symmetry, more particularly extendible ones, have applications in the quantum Max-Cut problem \cite{anshu2020beyond, parekh2021application, king2023improved}, the quantum cloning problem \cite{werner1998optimal,keyl1999optimal, nechita2023asymmetic}, quantum position verification \cite{buhrman2014position, kent2011quantum, unruh2014quantum} and many body physics \cite{jakab2021phase}.

The structure of this article is as follows: Section \ref{sec:ext} introduces the basic definitions and properties of extendibility; Section \ref{sec:brauer} describes the set of Brauer states; Sections \ref{sec:recipe} and \ref{sec:defi} provide the recipes to determine the set of $(n,m)$-extendible and $n$-de~Finetti-extendible Brauer states, respectively. Sections \ref{sec:result1} and \ref{sec:estimate} present concrete results by applying the first recipe: the sets of $(1,2)$-, $(2,2)$- and $(1,3)$-extendible Brauer states for any dimension and an estimate of the set of $(1,m)$-extendible Brauer states for any dimension is determined. Finally, Section \ref{sec:result2} presents concrete results by applying the second recipe: the sets of $n$-de~Finetti-extendible Brauer states for any dimension and the set of $\infty$-de~Finetti-extendible Brauer states for any dimension are determined. See also Reference \cite{allerstorfer2023monogamy} for their results in this latter topic, which are in line with our findings.

Note that throughout the article $\hil$ denotes a finite-dimensional complex Hilbert space and $\states{\hil}$ denotes the set of quantum states on $\hil$.

\section{Extendibility}\label{sec:ext}

There exist many notions of extendibility (also called shareability) all closely related to the quantum marginal problem \cite{ruskai1969representability,haapasalo2021quantum} which asks the following: Given reduced states, is there a global state such that these reduced states are its marginals? In the case of extendibility, the reduced states are usually assumed to be the same bipartite state. In the following, we present two versions of extendibility: two-sided extendibility and de~Finetti extendibility. In the language of Reference \cite{allerstorfer2023monogamy} these are, respectively, extendibility on the bipartite complete graph and the complete graph.

\subsection{Two-sided extendibility}

In general, when talking about extendibility, the literature usually refers to the notion of one-sided extendibility \cite{terhal2004entanglement, doherty2004complete}, as it is the most well known. However, two-sided extendibility has also gained attention \cite{terhal2003symmetric, johnson2013compatible, jakab2022extendibility} and is a generalisation of the former.

\begin{definition}
    A bipartite quantum state $\rho\in\states{\hil_A\ot\hil_B}$ on the Hilbert-space $\hil_A\ot\hil_B$ is said to be $(n,m)$-extendible if there exists a quantum state $\hat{\rho}\in\states{\hil_A^{\ot n}\ot \hil_B^{\ot m}}$ such that restricting it to any $\hil_A\ot\hil_B$ pair gives the original state. That is, for all $i\in[n]$ and $j\in[n+1,n+m]$\footnote{The notation $[n]$ denotes the set of positive natural numbers up to $n$, $[n]\coloneqq\{1,2,\ldots, n\}$. Similarly, $[n+1,n+m]\coloneqq\{n+1, n+2, \ldots, n+ m\}$.}:
    \begin{equation}
        \underset{\substack{\mathcal{I}\setminus \{i,j\}}}{\tr} \left (\hat{\rho} \right )=\rho\text{,}
    \end{equation}
    where $\mathcal{I}$ denotes the set of indices for the Hilbert spaces, and thus the trace denotes a partial trace to all but the $i$-th and $j$-th Hilbert spaces.
\end{definition}

A state $\hat{\rho}\in\states{\hil_A^{\ot n}\ot \hil_B^{\ot m}}$ will either be referred to as an extended form or an extending state of $\rho$, and it is usually not unique. Note that unlike the original definition of Reference \cite{terhal2003symmetric} for \textit{two-sided symmetric extendibility}, here permutation invariance is not required. However, if one is given an $(n,m)$-extending state which does not have permutation invariance for the two types of Hilbert spaces (Alice's and Bob's), then it can be made permutation invariant by simply twirling with all permutations from the group $\sym_n\times\sym_m$ \cite{johnson2013compatible}, where $\sym_n$ and $\sym_m$ denote the symmetric groups of degree $n$ and $m$, respectively.
 
One should think of two-sided extendibility like this: Given a bipartite state, can one find a state on a larger Hilbert space (with $n+m$ parts) such that its appropriate pairwise reductions reproduce the original state? If yes, the original state is said to be $(n,m)$-extendible.

The degree of a state's extendibility is closely related to entanglement, as attested by the following theorem \cite{fannes1988symmetric,raggio1989quantum,schumacher2002conjecture,doherty2004complete}:

\begin{theorem}
    A bipartite quantum state $\rho\in\states{\hil_A\ot\hil_B}$ is separable if and only if it is $(1,m)$-extendible for any $m$.
\end{theorem}

Note that obviously the theorem also holds true for `left-extensions', as $\hil_B$ does not have a distinguished role. The theorem also implies that, in the case of non-separable states (i.e., entangled states), there exist maximal extendibility numbers above which the state can be shared no more. This could be thought of as an embodiment of the monogamy of entanglement \cite{terhal2004entanglement, koashi2004monogamy}. In particular, a pure entangled state cannot be shared at all, i.e., they are only $(1,1)$-extendible.

Note that, naturally, if a state is $(n,m)$-extendible, it is also $(n',m')$-extendible for all $n'\leq n$ and $m'\leq m$. Thus, a partial order can be introduced on the extendibility numbers, and this is what leads to the maximal extendibility numbers that characterise a state's entanglement. Furthermore, if a family of bipartite states is symmetric, i.e., invariant to the flip (swap) operator, it is always enough to consider $(n,m)$-extendibility with $n\leq m$, since the $(n,m)$-extendible subset of states is also the $(m,n)$-extendible subset of states.

It is conjectured that LOCC transformations cannot decrease the maximal extendibility numbers of a state, however, so far this has only been rigorously proven for 1-LOCC (i.e., local operations with only a single round of classical communication) transformations in Reference \cite{nowakowski2016symmetric}. This fact makes maximal extendibility numbers a candidate to be an entanglement monotone.

A state's maximal extendibility numbers also serve as an upper bound for its distance from the set of separable states \cite{navascues2009complete,brandao2011faithful, brandao2017quantum}. References \cite{brandao2017quantum, doherty2014entanglement} actually argue that it may be more natural to consider the 1-LOCC norm when discussing the distance between the set of $(1,m)$-extendible states and the set of separable states. This norm is linked to the optimal probability for distinguishing two equiprobable states with a quantum measurement that can be realised as 1-LOCC operations, with the communication from $B$ to $A$ \cite{brandao2011faithful,brandao2017quantum}:

\begin{theorem}
    Let $\rho\in\states{\hil_A\ot\hil_B}$ be a $(1,m)$-extendible quantum state. Then there exists a separable state $\sigma\in\states{\hil_A\ot\hil_B}$ such that
    \begin{equation}
        \norm{\rho-\sigma}_{\mathrm{LOCC}^{\leftarrow}} \leq \sqrt{\frac{2 \log (d_A)}{m}}\text{,}
    \end{equation}
    where $\norm{.}_{\mathrm{LOCC}^{\leftarrow}}$ denotes the 1-LOCC norm with communication going from B to A, and $d_A$ is the dimension of the Hilbert space $\hil_A$.
\end{theorem}

It is worth noting that the 1-LOCC norm is more forgiving than the trace norm, i.e., more states can be regarded as approximately separable if one uses the 1-LOCC norm.

We have argued that a state's maximal extendibility numbers give relevant information about the state's entangledness. It is well known that calculating any entanglement measure for a mixed state is difficult and requires a convex optimization process. In a similar fashion, calculating the maximal extendibility numbers for an arbitrary state is difficult and can only be done numerically \cite{doherty2002distinguishing, doherty2004complete, navascues2009complete, boyd2004convex, benson2000solving}. However, restricting the analysis to states with symmetry allows us to use representation-theoretic techniques to derive analytic results. This in turn also means that the maximal extendibility numbers are calculable even for mixed states with such symmetries.

\subsection{de~Finetti extendibility}

While two-sided extendibility can be defined for any bipartite state, de~Finetti extendibility (also called exchangeability) can only be defined for symmetric bipartite states, i.e., states invariant to the flip operator, such as Werner states, for example. In the case of de~Finetti extendibility \cite{konig2005deFinetti,christandl2007oneandahalf}, only one number describes the extension and it is required that every bipartite reduction gives the original state.
\begin{definition}
    A symmetric bipartite quantum state $\rho\in\states{\hil\ot\hil}$ on the Hilbert space $\hil\ot\hil$ is said to be $n$-de~Finetti-extendible ($n\geq 2$) if there exists a quantum state $\hat{\rho}\in\states{\hil^{\ot n}}$ such that any pairwise reduction gives the original state. That is, for all $i,j\in[n]$, $i\neq j$:
    \begin{equation}
        \underset{\substack{\mathcal{I}\setminus \{i,j\}}}{\tr} \left (\hat{\rho} \right )=\rho\text{,}
    \end{equation}
    where $\mathcal{I}$ denotes the set of indices for the Hilbert spaces, and thus the trace denotes a partial trace to all but the $i$-th and $j$-th Hilbert spaces.
\end{definition}

A state $\hat{\rho}\in\states{\hil^{\ot n}}$ will either be referred to as a de~Finetti extended form or a de~Finetti extending state of $\rho$, and it is usually not unique. Note that permutation invariance is not required. However, if one is given an $n$-de~Finetti-extending state which does not have permutation invariance for the Hilbert spaces, then it can be made permutation invariant by simply twirling with all permutations from the symmetric group $\sym_n$.

Unlike one-sided or two-sided extendibility, de~Finetti extendibility is such a strong requirement that not even every separable state is de~Finetti extendible for any $n$. Indeed, only the convex combinations of flip-invariant product states (i.e., of the form $\sigma \otimes \sigma $) are de~Finetti extendible for any $n$, which is a direct consequence of the quantum de~Finetti theorem \cite{christandl2007oneandahalf}.

\section{Brauer states}\label{sec:brauer}

Brauer states (also called orthogonal or OO-states) have come about as a generalisation of Werner and isotropic states \cite{vollbrecth2001entanglement, keyl2002fundamentals, allerstorfer2023monogamy}. Famously, Werner states are bipartite states that are invariant to local unitary transformations of the form $U \ot U$. On the other hand, isotropic states are invariant to local unitaries with the second unitary elementwise complex conjugated with respect to a fixed basis, that is $U \ot \overline{U}$. Brauer states are those states which are invariant to the intersection of the above-introduced two sets of transformations, which is nothing else but the set of local orthogonal transformation, i.e., unitaries of the form $O\ot O$. More formally:
\begin{definition}
    A bipartite quantum state $\brauer \in \states{\hil\ot \hil}$ (where $\dim \hil=d>1$) is called a Brauer state if it is invariant under the diagonal action of the defining representation of the $\Ort(d)$ group on $\hil\ot\hil$. That is,
    \begin{equation}
        \ytableausetup{mathmode, boxsize=\Yt}
        \comm{\brauer}{\Repp_{\ydiagram{1}}(g)\ot\Repp_{\ydiagram{1}}(g)}=0\text{,} \qquad \forall g\in\Ort(d)\text{,}
\end{equation}
    where $\Repp_{\ydiagram{1}}:\Ort(d)\to \GL(d,\comp)$ is the defining representation of the real orthogonal group of degree $d$ given a fixed real structure.
\end{definition}

Note that a fixed real structure is needed to be able to define transposes on $\hil$. Fixing a real structure can be done by fixing a basis with respect to which the transposition is being done. This is exactly the same basis with respect to which the transformations of isotropic states are complex conjugated.

Observe that every Werner and every isotropic state is automatically a Brauer state: Since the Werner and isotropic states are invariant to a larger set of group elements that also contain every orthogonal group element, they are indeed Brauer states as well.

Unlike Werner and isotropic states, Brauer states are described by two parameters. Using the representation theory of orthogonal groups and Schur's lemma, Brauer states can be parametrised in the following way:
\begin{equation}
    \brauer(\mu_1,\mu_2)=(1-\mu_1-\mu_2)\PROJ_0+\mu_1\frac{\PROJ_{\ydiagram{1,1}}}{\dim\left(\ytableausetup{mathmode, boxsize=0.5em, aligntableaux=center} \ydiagram{1,1}\right)}+\mu_2\frac{\PROJ_{\ytableausetup{mathmode, boxsize=\Yt} \ydiagram{2}}}{\dim(\ytableausetup{mathmode,boxsize=0.5em} \ydiagram{2})}\text{,}
\end{equation}
where $\dim\left(\ydiagram{1,1}\right)=\frac{d(d-1)}{2}$, $\dim(\ydiagram{2})=\frac{d(d+1)}{2}-1$, $0\leq \mu_1,\mu_2$, and $\mu_1+\mu_2\leq 1$. The $\PROJ$ operators denote projectors onto the three invariant subspaces under all $\Repp_{\ytableausetup{mathmode, boxsize=\Yt} \ydiagram{1}}(g)\ot\Repp_{\ydiagram{1}}(g)\cong \Repp_0(g) \OP \Repp_{\ydiagram{1,1}}(g) \OP \Repp_{\ydiagram{2}}(g)$ transformations. To see this, one can rely on representation theory (see Appendix \ref{sec:app_rep}).

However, instead of the above parametrisation, this article uses that of Keyl \cite{keyl2002fundamentals}, which is based on the Brauer-Schur-Weyl duality for the orthogonal group \cite{brauer1937algebras}. It states that the commutant of the diagonal action of two instances of the defining representation of the orthogonal group is the Brauer algebra $\mathcal{B}_2$, which is generated by the set $\{\id_{\hil}\ot \id_{\hil}, \flip, \fflip\}$, where $\id_{\hil}\ot\id_{\hil}$ is the identity on $\hil\ot\hil$, $\flip$ is the flip (or swap) operator and $\fflip=\flip^{t_2}$ is the partially transposed flip operator. These last two operators are nothing else than the operators that generate, along with the identity, the Werner and the isotropic states, respectively. Together, all three can be used to describe Brauer states. The relationship between these operators and the projectors from above is the following:
\begin{equation}
    \PROJ_0= \frac{\fflip}{d}\eqqcolon \bb \text{,} \qquad \qquad
    \PROJ_{\ydiagram{1,1}}=\frac{\id_{\hil}\ot\id_{\hil}-\flip}{2}\text{,} \qquad \qquad
    \PROJ_{\ydiagram{2}}=\frac{\id_{\hil}\ot\id_{\hil}+\flip}{2}-\frac{\fflip}{d}\text{.}
\end{equation}
Here, the notation $\bb$ has been introduced. Instead of $\mu_1$ and $\mu_2$, it is more convenient to parametrise Brauer states with $f\coloneqq \tr(\brauer \flip)$ and $b\coloneqq \tr(\brauer \bb)$. It is easy to see that these parameters lie in the following range:
\begin{align}
    f&=1-2\mu_1\in[-1,1]\text{,}\\
    b&=1-\mu_1-\mu_2\in[0,1]\text{,}
\end{align}
with the following restriction:
\begin{equation}
    b\leq \frac{1}{2}(f+1)\text{.}
\end{equation}

\vspace{0.5em}

The advantage of using $\tr(\brauer\bb)$ instead of $\tr(\brauer\fflip)$ is that the latter has a range between $[0,d]$, therefore using $\bb$ normalises this range. In this parametrisation the set of separable Brauer states is well known to be the rectangle with vertices $(0,0) - (0,\frac{1}{d}) - (1,\frac{1}{d}) - (1,0)$ \cite{keyl2002fundamentals}. Note that these are exactly the PPT states; thus, for Brauer states, the Peres-Horodecki criterion \cite{peres1996separability, horodecki1996separability} suffices to find the set of separable states.

Figure \ref{fig:brauer} illustrates the set of Brauer states.

\begin{figure}[H]
    \centering
    \includegraphics[width=0.5\textwidth]{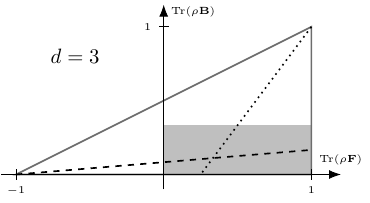}
    \caption{\small The set of Brauer states at $d=3$ in the $\tr(\rho\flip) - \tr(\rho\bb)$ parametrisation appearing as the grey triangle. The set of Werner states is illustrated as the dashed line, while the set of isotropic states as the dotted line. The set of separable Brauer states constitute the grey rectangle.}
    \label{fig:brauer}
\end{figure}

\vspace{-0.5cm}

Note that while the set is the grey triangle for any dimension $d$, the exact lines of Werner and isotropic states change depending on $d$, as well as the set of separable Brauer states.

\section{General recipe for determining the set of $(n,m)$-extendible Brauer states} \label{sec:recipe}

\subsection{The $n\neq m$ case}

Our goal is to describe the general procedure to find the subset of parameters for each dimension $d$ which describe the set of $(n,m)$-extendible Brauer states. Naturally, this is a convex subset, as any convex sum of two $(n,m)$-extendible states is also $(n,m)$-extendible.

Obviously, if one is somehow given the set of $n+m$ partite states that are such that they extend a Brauer state, one could obtain the set of $(n,m)$-extendible Brauer states by partial trace. More straightforwardly, given the properties of marginals, one would only have to calculate the trace of each of them with $\flip_{ij}$ and $\bb_{kl}$ for any $i,k\in[n]$ and $j,l\in[n+1,n+m]$ to get the set of parameters describing the $(n,m)$-extendible Brauer states.\footnote{$\flip_{ij}$ describes the $\flip$ operator acting between the $i$-th and $j$-th Hilbert space while acting as identity everywhere else. Similarly for $\bb_{kl}$.}

It turns out that, using representation theory, it is possible to obtain a \textit{relevant} subset of $(n,m)$-extending states in a nicely parametrised manner as the commutant of $O^{\ot (n+m)}$ transformations. More precisely, the following requirements are prescribed for a \textit{relevant} $(n,m)$-extending state $\hat{\rho}\in\states{\hil^{\ot n}\ot\hil^{\ot m}}$:
\begin{enumerate}
    \item $\comm{\hat{\rho}}{\Repp_{\ydiagram{1}}^{\ot (n+m)}(g)}\overset{!}{=}0$ for all $g\in\Ort(d)$,
    \item $\comm{\hat{\rho}}{\hat{\Repppp}^{\mathrm{L}}_{n}(h)\ot\hat{\Repppp}^{\mathrm{R}}_{m}(k)}\overset{!}{=}0$ for all $h\in \sym_n$ and $k\in \sym_m$,
\end{enumerate}
where $\Repp_{\ydiagram{1}}:\Ort(d)\to \GL(d,\comp)$ is the defining representation of the orthogonal group and $\hat{\Repppp}_n:\sym_n\to \GL(\hil^{\ot n})$ is the natural representation of the symmetric group on $\hil^{\ot n}$ that permutes the Hilbert spaces:
\begin{equation}
    \hat{\Repppp}_n(h) (v_1 \ot v_2 \ot \cdots \ot v_n )\coloneqq v_{h^{-1}(1)} \ot v_{h^{-1}(2)} \ot \cdots \ot v_{h^{-1}(n)}\text{,}
\end{equation}
where $h\in\sym_n$, $v_i\in\hil$ and the `hat' symbol over the representation highlights the fact that the representation is, in general, not irreducible. The upper index L or R is there to distinguish the representation of $\sym_n$ and $\sym_m$ acting on the Hilbert spaces of the Alices and Bobs, respectively. Note that even for the case of $n=m$, these representations cannot be fused as they are representations of $\sym_n \times \sym_m$ and not simply $\sym_{n+m}$. However, there is a special extra symmetry for $(n,n)$-extendible states that is described in Section \ref{sec:nn}. Note also that the second requirement is nothing more than that of the two-sided symmetric extendibility.

These are sufficient conditions for $\hat{\rho}$ to be an $(n,m)$-extending state of Brauer state, however, at first sight they might not be necessary. Indeed, requirement 1 guarantees that each appropriate bipartite reduction of $\hat{\rho}$ is a Brauer state, and requirement 2 guarantees that these appropriate bipartite reductions are the same, however, these might exclude some $(n,m)$-extending states. The following lemma, based on Reference \cite{vollbrecth2001entanglement}, guarantees that no information is lost, and therefore it is indeed enough to examine only the above set of \textit{relevant} $(n,m)$-extending states. The idea behind it is that given a general extending state, it can be made to obey the above requirements by twirling, without changing the Brauer state it extends:

\begin{lemma}\label{lm:proj}
    Let $\brauer\in\states{\hil\ot\hil}$ be an $(n,m)$-extendible Brauer state, and let $\hat{\rho}'\in\states{\hil^{\ot n}\ot\hil^{\ot m}}$ be any $(n,m)$-extending state of $\brauer$. Let a twirl operation $\twirl:\states{\hil^{\ot n}\ot \hil^{\ot m}}\to \states{\hil^{\ot n}\ot \hil^{\ot m}}$ act on $\hat{\rho}'$ as described by the above requirements:
    \begin{equation}
        \hat{\rho}\coloneqq \twirl(\hat{\rho}')\coloneqq\frac{1}{n! m!} \! \! \! \sum_{\substack {h\in \sym_n \\ k \in \sym_m}} \! \! \! \int_{\Ort(d)} \! \! \! (\Repp_{\ydiagram{1}}(g)^{\ot (n+m)}) (\hat{\Repppp}^{\mathrm{L}}_{n}(h)\ot\hat{\Repppp}^{\mathrm{R}}_{m}(k)) \hat{\rho}' (\hat{\Repppp}^{\mathrm{L}}_{n}(h)\ot\hat{\Repppp}^{\mathrm{R}}_{m}(k))^{\dagger} (\Repp_{\ydiagram{1}}(g)^{\ot (n+m)})^{\dagger}\mathrm{d}g\text{,}
    \end{equation}
    where the integral is with respect to the normalised Haar measure over $\Ort(d)$. The resulting state $\hat{\rho}$ is also an $(n,m)$-extending state of the same $\brauer$ and is invariant to the groups $\Ort(d)$ and $\sym_n\times\sym_m$ as described by the above requirements.
\end{lemma}

\begin{proof}
    By twirling with the group $\sym_n \times \sym_m \times \Ort(d)$ represented in the above way, we guarantee that $\hat{\rho}$ is the $(n,m)$-extended form of some Brauer state. The only thing that needs to be proven is that it is the same Brauer state as $\brauer$. We shall prove that it gives the same parameters when traced with the appropriate $\flip$ and $\bb$ operators. Let $\tr(\brauer \flip)\eqqcolon f$ and $\tr(\brauer \bb) \eqqcolon b$. As $\hat{\rho}'$ is an extending state of $\brauer$ the following is true for any $i,k\in[n]$ and $j,l\in[n+1, n+ m]$:
    \begin{equation}
        \tr(\hat{\rho}' \flip_{ij})=f\text{,} \qquad \qquad \tr(\hat{\rho}' \bb_{kl})=b\text{.}
    \end{equation}

    Given that quantum states are self-adjoint, the above traces can be thought of as Hilbert-Schmidt inner products, e.g., $\tr(\hat{\rho}' \flip_{ij})=\inn{\hat{\rho}'}{\flip_{ij}}_{\text{HS}}$. As the twirl superoperator $\twirl$ is self-adjoint and the operators $\flip_{ij}$ and $\bb_{kl}$ are invariant to orthogonal transformations (given that they are part of the Brauer algebra), the following is true:
    \begin{equation}
    \begin{split}
        \tr(\hat{\rho}\flip_{ij})&=\tr(\twirl(\hat{\rho}')\flip_{ij})=\inn{\twirl(\hat{\rho}')}{\flip_{ij}}_{\text{HS}}=\inn{\hat{\rho}'}{\twirl(\flip_{ij})}_{\text{HS}}\\
        &=\Biggl\langle \hat{\rho}'\;,\; \frac{1}{n!m!} \!\!\!\!\!\!\sum_{\substack{\alpha\in[n] \\ \beta \in [n+1,n+m]}}\!\!\!\!\!\! \flip_{\alpha\beta}\Biggr \rangle_{\text{HS}}=\tr(\hat{\rho}' \flip_{ij})\text{,}
    \end{split}
    \end{equation}
    where, in the last step, the property of $\hat{\rho}'$ being an extending state has been used. The proof is similar for $\bb_{kl}$. This proves that $\hat{\rho}=\twirl(\hat{\rho}')$ is an extending state of the same Brauer state $\brauer$.
\end{proof}

Therefore, no information is lost if the above requirements are prescribed, and it is enough to find every extending state satisfying them to characterise the set of $(n,m)$-extending Brauer states.

The commutant of the set of the above-mentioned operators can be calculated using representation theory, see Appendix \ref{sec:app_sw}. Of course, as $n$ and $m$ grow larger, it becomes increasingly difficult to give the exact form of a commuting operator. Nevertheless, the general form of a commuting state can be given as:
\begin{equation}\label{eq:commutant}
    \hat{\rho} \cong \bigoplus_{\Lambda} \mu_{\Lambda} \rho_{\Lambda} \ot \frac{\id_{\lil_{\Lambda}}}{\dim(\lil_{\Lambda})} \text{,}
\end{equation}
where the right-hand side reflects the canonical decomposition of the tensor product of representations into the direct sum of isotypic components: $\hil^{\ot (n+m)}\cong \oplus_{\Lambda} \kil_{\Lambda}\ot\lil_{\Lambda}$. Given an isotypic component, the irreducible representation (irrep) of the orthogonal and symmetric groups labelled by $\Lambda$ acts on the right Hilbert space ($\lil_{\Lambda}$) of the tensor product, while the left Hilbert space ($\kil_{\Lambda}$) describes the multiplicity space. In accordance with Schur's lemma, the commutant is non-trivial only on the multiplicity spaces: $\rho_{\Lambda}\in\states{\kil_{\Lambda}}$ are arbitrary quantum states, $\mu_{\Lambda}\geq 0$ and $\sum_{\Lambda} \mu_{\Lambda}=1$. In other words, $\hat{\rho}$ can be written as a convex sum of special quantum states with respect to the isotypic decomposition of the given representation. This gives rise to a rather nice parametrisation of the relevant subset of $(n,m)$-extending Brauer states.

Note that the label $\Lambda$ can be split into three sublabels, corresponding to the tensor product of three irreps, one for each group in the direct product $\sym_n \times \sym_m \times \Ort(d)$, see Appendices \ref{sec:app_rep} and \ref{sec:app_sw}. Note also that if $\rho_{\Lambda}$ is a state on a 1-dimensional Hilbert space, that is, the irrep labelled by $\Lambda$ has multiplicity one, then simply $\rho_{\Lambda}=1$. In this case, one is left with a projector onto the isotypic component.

If one can calculate the relevant traces of these constituent special quantum states with $\flip_{ij}$ and $\bb_{kl}$, then one gets a subset of the parameter space of Brauer states whose convex hull describes the set of $(n,m)$-extendible Brauer states. Inspired by Reference \cite{jakab2022extendibility}, the traces with the following operators will be calculated:
\begin{equation}
    \smf_{n,m} \coloneqq \frac{1}{nm}\sum_{i=1}^{n}\sum_{j=n+1}^{n+m} \flip_{ij}\text{,} \qquad \qquad \smb_{n,m} \coloneqq \frac{1}{nm}\sum_{i=1}^{n}\sum_{j=n+1}^{n+m} \bb_{ij}\text{.}
\end{equation}

The reason for using these new averaged operators instead of any particular $\flip_{ij}$ or $\bb_{kl}$ operator is that one can take advantage of some representation-theoretic tools, namely Casimir operators, to express them in a calculable manner. In terms of Casimir operators, these averaged operators take the following form (see Appendices \ref{sec:app_cas} and \ref{sec:app_calcs} for the derivation):
\begin{align}
\begin{split}\label{eq:smf}
    \smf_{n,m}=&\frac{2}{nm(\chi(\ydiagram{2})-\chi(\ydiagram{1,1}))}\left[\rep_{\ydiagram{1}}^{\times (n+m)}(C^{\ualg(d)}) - \rep_{\ydiagram{1}}^{\times n}(C^{\ualg(d)}) \ot \id^{\ot m}_\hil - \id^{\ot n}_\hil \ot \rep_{\ydiagram{1}}^{\times m}(C^{\ualg(d)}) \right]\text{,}
\end{split}
\end{align}
\begin{align}
\begin{split}\label{eq:smb}
    \smb_{n,m}=&\frac{\xi(\ydiagram{2})-\xi(\ydiagram{1,1})}{2\xi(\ydiagram{2})}\smf_{n,m} -\frac{1}{nm\xi(\ydiagram{2})}\left[\repp_{\ydiagram{1}}^{\times (n+m)}(C^{\So(d)}) - \repp_{\ydiagram{1}}^{\times n}(C^{\So(d)}) \ot \id^{\ot m}_\hil - \id^{\ot n}_\hil \ot \repp_{\ydiagram{1}}^{\times m}(C^{\So(d)})\right]\text{,}
\end{split}
\end{align}
where $\chi(\lambda)$ and $\xi(\lambda)$ denote the quadratic Casimir eigenvalues of $\ualg(d)$ or $\So(d)$, respectively, in the irreducible representation labelled by the Young diagram $\lambda$, and $\rep_{\ydiagram{1}}^{\times n}(C^{\ualg(d)})$ and $\repp_{\ydiagram{1}}^{\times n}(C^{\So(d)})$ denote the quadratic Casimir operator in the $n$-fold tensor product of the defining representation of the unitary and the special orthogonal Lie algebra, respectively.

Given the above forms of $\smf_{n,m}$ and $\smb_{n,m}$, to do the calculations one needs to know the trace of a given constituent quantum state ($\rho_{\Lambda} \ot \id_{\lil_{\Lambda}}/ \dim(\lil_{\Lambda})$) and the Casimir operators in the above representations (e.g., $\repp_{\ydiagram{1}}^{\times n}(C^{\So(d)}) \ot \id^{\ot m}_\hil$). The case of $\repp_{\ydiagram{1}}^{\times (n+m)}(C^{\So(d)})$ is trivial, as the operator can be brought to the same form as the one in Equation \eqref{eq:commutant}: $\repp_{\ydiagram{1}}^{\times (n+m)}(C^{\So(d)})\cong \bigoplus_{\Lambda} \id_{\kil_{\Lambda}} \ot \xi(\lambda)\id_{\lil_{\Lambda}}$, where $\lambda$ is the label of the $\So(d)$ irrep associated with $\Lambda$. However, it is less straightforward for the Casimir operators in the other representations as not all of them will have a constant eigenvalue on an isotypic component like the operator $\repp_{\ydiagram{1}}^{\times (n+m)}(C^{\So(d)})$ does.

In practice, one has to fix distinguished bases on $\kil_{\Lambda}$ that make it so that most representations of the Casimir operators are diagonal in them, although they might not have the same eigenvalue for every basis vector. There are two notable bases to talk of. Each of them starts by diagonalising the Casimir operator represented as $\rep_{\ydiagram{1}}^{\times n}(C^{\ualg(d)}) \ot \id^{\ot m}_\hil$ and $\id^{\ot n}_\hil \ot \rep_{\ydiagram{1}}^{\times m}(C^{\ualg(d)})$.

The first one comes about by fusing together the irreps of $\ualg(d)$ so that the represented Casimir operator $\rep_{\ydiagram{1}}^{\times (n+m)}(C^{\ualg(d)})$ can also be diagonalised. As the last step, the $\ualg(d)$ irreps are restricted to $\So(d)$ irreps giving rise to the diagonalised Casimir operator represented as $\repp_{\ydiagram{1}}^{\times (n+m)}(C^{\So(d)})$. However, in this basis, the Casimir operator represented as $\repp_{\ydiagram{1}}^{\times n}(C^{\So(d)}) \ot \id^{\ot m}_\hil$ and $\id^{\ot n}_\hil \ot \repp_{\ydiagram{1}}^{\times m}(C^{\So(d)})$ is generally not diagonal.

The second comes about by restricting the Casimir operator represented as $\rep_{\ydiagram{1}}^{\times n}(C^{\ualg(d)}) \ot \id^{\ot m}_\hil$ and $\id^{\ot n}_\hil \ot \rep_{\ydiagram{1}}^{\times m}(C^{\ualg(d)})$ to $\So(d)$ irreps giving rise to the diagonalised Casimir operator represented as $\repp_{\ydiagram{1}}^{\times n}(C^{\So(d)}) \ot \id^{\ot m}_\hil$ and $\id^{\ot n}_\hil \ot \repp_{\ydiagram{1}}^{\times m}(C^{\So(d)})$. Finally, the $\So(d)$ irreps are fused together so that the Casimir operator represented as $\repp_{\ydiagram{1}}^{\times (n+m)}(C^{\So(d)})$ is also diagonal. Naturally, in this basis, the Casimir operator represented as $\rep_{\ydiagram{1}}^{\times (n+m)}(C^{\ualg(d)})$ is generally not diagonal. The two processes are illustrated in Figure \ref{fig:casimirs}.

\begin{figure}[H]
    \centering
	\begin{tikzpicture}[scale=1]

    \newcommand{\Widtha}{3.75}
    \newcommand{\Widthb}{5}
    \newcommand{\Heighta}{-1.75}
    \newcommand{\smalla}{0.5}

    \draw (\Widtha+\Widthb/2,\smalla) -- (\Widtha+\Widthb/2,2*\Heighta-\smalla);

    \node at (0,0) {$\rep_{\ydiagram{1}}^{\times n}(C^{\ualg(d)}) \ot \id^{\ot m}_\hil$};
    \node at (\Widtha,0) {$\id^{\ot n}_\hil \ot \rep_{\ydiagram{1}}^{\times m}(C^{\ualg(d)})$};

    \node at (\Widtha+\Widthb,0) {$\rep_{\ydiagram{1}}^{\times n}(C^{\ualg(d)}) \ot \id^{\ot m}_\hil$};
    \node at (2*\Widtha+\Widthb,0) {$\id^{\ot n}_\hil \ot \rep_{\ydiagram{1}}^{\times m}(C^{\ualg(d)})$};

    \draw[double distance=3pt, arrows = {-Latex[length=0pt 2.5 0]}] (\smalla,-\smalla)--(\Widtha/2-\smalla,\Heighta+\smalla);
    \draw[double distance=3pt, arrows = {-Latex[length=0pt 2.5 0]}] (\Widtha-\smalla,-\smalla)--(\Widtha/2+\smalla,\Heighta+\smalla);
    \node at (\Widtha/2,\Heighta+2*\smalla) {$\bigotimes$};

    \draw[double distance=3pt, arrows = {-Latex[length=0pt 2.5 0]}] (\Widtha+\Widthb,-\smalla/2)--(\Widtha+\Widthb,\Heighta+\smalla);
    \draw[double distance=3pt, arrows = {-Latex[length=0pt 2.5 0]}] (\Widtha+\Widtha+\Widthb,-\smalla/2)--(\Widtha+\Widtha+\Widthb,\Heighta+\smalla);
    \node at (\Widtha/2+\Widthb + \Widtha, \Heighta/2+\smalla/3) {\small $\ualg(d) \to \So(d)$};

    \node at (\Widtha/2,\Heighta) {$\rep_{\ydiagram{1}}^{\times (n+m)}(C^{\ualg(d)})$};

    \node at (\Widtha+\Widthb,\Heighta) {$\repp_{\ydiagram{1}}^{\times n}(C^{\So(d)}) \ot \id^{\ot m}_\hil$};
    \node at (2*\Widtha+\Widthb,\Heighta) {$\id^{\ot n}_\hil \ot \repp_{\ydiagram{1}}^{\times m}(C^{\So(d)})$};

    \draw[double distance=3pt, arrows = {-Latex[length=0pt 2.5 0]}] (\Widtha/2,\Heighta-\smalla/2)--(\Widtha/2,2*\Heighta+\smalla);
    \node at (\Widtha/2+2.5*\smalla, \Heighta+\Heighta/2+\smalla/3) {\small $\ualg(d) \to \So(d)$};

    \draw[double distance=3pt, arrows = {-Latex[length=0pt 2.5 0]}] (\Widtha+\Widthb+\smalla,\Heighta-\smalla)--(\Widtha+\Widthb+\Widtha/2-\smalla,2*\Heighta+\smalla);
    \draw[double distance=3pt, arrows = {-Latex[length=0pt 2.5 0]}] (\Widtha+\Widthb+\Widtha-\smalla,\Heighta-\smalla)--(\Widtha+\Widthb+\Widtha/2+\smalla,2*\Heighta+\smalla);
    \node at (\Widtha+\Widthb+\Widtha/2,2*\Heighta+2*\smalla) {$\bigotimes$};

    \node at (\Widtha/2,2*\Heighta) {$\repp_{\ydiagram{1}}^{\times (n+m)}(C^{\So(d)})$};

    \node at (\Widtha+\Widthb+\Widtha/2,2*\Heighta) {$\repp_{\ydiagram{1}}^{\times (n+m)}(C^{\So(d)})$};

    \node at (0,2.5*\Heighta) {};
    
	\end{tikzpicture}
    \caption{\small Illustration of the two types of diagonalisation processes of the represented Casimir operators, starting from the same point. On the left side the irreps are fused together and then a restriction happens to the subalgebra $\So(d)$, while on the right side the process is done the other way around.}
    \label{fig:casimirs}
\end{figure}

\vspace{-0.5cm}

Ultimately, the two different choices of basis can be equated to two different irreducible decompositions of the representation of $\sym_n \times \sym_m \times \Ort(d)$ in the sense that while both contain the same individual irreps, these irreps are labelled by either which $\ualg(d)$ irrep they came from or which $\So(d)$ irreps fused together to form them.

Given that not every Casimir representation will generally be diagonal in the two distinguished bases, the unitary operator that links them will be needed. From Schur's lemma, the unitary operator $\Vry$ describing the change of basis between these two bases has to be an intertwiner. This means that it has the following form:
\begin{equation}\label{eq:uni_intertwiner}
    \Vry = \bigoplus_{\Lambda} \Vry_{\Lambda} \ot \id_{\lil_{\Lambda}} \text{,}
\end{equation}
where $\Vry_{\Lambda}\in\bound{\kil_{\Lambda}}$ are unitary operators. Note that if $\Vry_{\Lambda}$ acts on a 1-dimensional Hilbert space (which means that the multiplicity is one), then it is simply a phase factor $e^{i\alpha}$, $\alpha\in[0,2\pi)$, which is irrelevant as it factors out. Therefore, if the multiplicity is one, the Casimir representations can always be diagonalised simultaneously.

To be able to perform the calculations, the relevant parameters of these unitaries have to be known. In general, this is a hard problem, but the relevant parameters can be calculated for low multiplicity cases using methods that build on Young symmetrisers. To this end, the first type of basis is favourable.

Appendix \ref{sec:app_res} shows a concrete calculation based on the principles presented here.

\subsection{The $(n,n)$-extendible Brauer states}\label{sec:nn}

In the case of $n=m$, that is, $(n,n)$-extendible states, there is an additional symmetry that can be exploited. Given that Brauer states are symmetric, meaning that they themselves are flip invariant, it can be required that their $(n,n)$-extensions are also invariant to the \textit{faction flip} in which the Alices and Bobs all change place. Note that this is also true for any flip invariant state, not just Brauer states.

This means that we are actually dealing with not simply a representation of $\sym_n\times\sym_n$ but a representation of a larger subgroup of $\sym_{2n}$. This larger subgroup has as its subgroup $\sym_n\times\sym_n$ with index 2, more precisely, it is a semidirect product of $\sym_n\times\sym_n$ and $\mathbb{Z}_2$. This is also referred to as the \textit{wreath product} in the literature: $\sym_n \wr \zahl_2$.

As a consequence, its representation theory can be inferred from that of $\sym_n\times\sym_n$ \cite{fulton2004representation}, similarly to how the representation theory of $\Ort(d)$ can be inferred from the representation theory of $\SO(d)$. There will be two cases:
\begin{enumerate}
    \item When the irreps of $\sym_n\times\sym_n$ are the same for the constituents (e.g., $\Repppp_\lambda \ot \Repppp_\lambda$) then they will lift up in two ways. This is related to whether the \textit{faction flip} is represented trivially or non-trivially.
    \item When the irreps of $\sym_n\times\sym_n$ are different on the constituents (e.g., $\Repppp_{\lambda_1} \ot \Repppp_{\lambda_2}$), then the two versions will lift up together to form a single irrep (concretely $\Repppp_{\lambda_1} \ot \Repppp_{\lambda_2}\OP \Repppp_{\lambda_2} \ot \Repppp_{\lambda_1}$).
\end{enumerate}

Naturally, to do concrete calculations, one has to go through the representation theory of the symmetric group for a given $n$. In this article, we present this for the $n=2$ case.

\section{General recipe for determining the set of {$n$-de~Finetti-extendible} Brauer states}\label{sec:defi}

Our goal is to describe the general procedure to find the subset of parameters for each dimension $d$ which describes the $n$-de~Finetti-extendible Brauer states. The ideas are very similar to those presented in the previous section, therefore, this section only presents the main steps in the procedure. The interested reader is directed to the previous section for more details.

The \textit{relevant} $n$-de~Finetti-extending states $\hat{\rho}\in\states{\hil^{\ot n}}$ obey the following requirements:
\begin{enumerate}
    \item $\comm{\hat{\rho}}{\Repp_{\ydiagram{1}}^{\ot n}(g)}\overset{!}{=}0$ for all $g\in\Ort(d)$, \vspace{0.2cm}
    \item $\comm{\hat{\rho}}{\hat{\Repppp}(h)}\overset{!}{=}0$ for all $h\in \sym_n$,
\end{enumerate}
where $\Repp_{\ydiagram{1}}:\Ort(d)\to \GL(d,\comp)$ is the defining representation of the orthogonal group and $\hat{\Repppp}_n:\sym_n\to \GL(\hil^{\ot n})$ is the natural representation of the symmetric group on $\hil^{\ot n}$ that permutes the Hilbert spaces.

Our previously presented Lemma \ref{lm:proj} (see also \cite{vollbrecth2001entanglement}) guarantees that it is enough to examine the above defined set of \textit{relevant} $n$-de~Finetti-extending states.

The commutant of the set of the above-mentioned operators can be calculated using representation theory. The general form of a commuting state can be calculated using Schur's lemma:
\begin{equation}
    \hat{\rho} \cong \bigoplus_{\Lambda} \mu_{\Lambda} \rho_{\Lambda} \ot \frac{\id_{\lil_{\Lambda}}}{\dim(\lil_{\Lambda})} \text{,}
\end{equation}
where the right-hand side reflects the canonical decomposition of the tensor product of representations into the direct sum of isotypic components: $\hil^{\ot n}\cong \oplus_{\Lambda} \kil_{\Lambda}\ot\lil_{\Lambda}$. Given an isotypic component, the irreducible representation of the orthogonal and symmetric groups labelled by $\Lambda$ acts on the right Hilbert space ($\lil_{\Lambda}$) of the tensor product, while the left Hilbert space ($\kil_{\Lambda}$) describes the multiplicity space. In accordance with Schur's lemma, the commutant is non-trivial only on the multiplicity spaces: $\rho_{\Lambda}\in\states{\kil_{\Lambda}}$ are arbitrary quantum states, $\mu_{\Lambda}\geq 0$ and $\sum_{\Lambda} \mu_{\Lambda}=1$. That is, $\hat{\rho}$ can be written as a convex sum of special quantum states with respect to the isotypic decomposition of the given representation. 

Inspired by \cite{jakab2022extendibility}, the traces with the following operators will be calculated:
\begin{equation}
    \smf_{n} \coloneqq \frac{2}{n(n-1)}\sum_{\substack{i,j=1 \\ i<j}}^{n}\flip_{ij}\text{,} \qquad \qquad \smb_{n} \coloneqq \frac{2}{n(n-1)} \sum_{\substack{i,j=1 \\ i<j}}^{n}\bb_{ij}\text{.}
\end{equation}

Based on the ideas presented in Appendices \ref{sec:app_cas} and \ref{sec:app_calcs}, the above operators can be expressed using Casimir operators:
\begin{align}
    \smf_{n}&=\frac{4}{n(n-1)(\chi(\ydiagram{2})-\chi(\ydiagram{1,1}))} \left[\rep_{\ydiagram{1}}^{\times n}(C^{\ualg(d)}) - n\chi(\ydiagram{1})\id_{\hil}^{\ot n} \right]\text{,}\\
    \smb_{n}&=\frac{\xi(\ydiagram{2})-\xi(\ydiagram{1,1})}{2\xi(\ydiagram{2})}\smf_{n}-\frac{2}{n(n-1)\xi(\ydiagram{2})}\left[\repp_{\ydiagram{1}}^{\times n}(C^{\So(d)}) - n\xi(\ydiagram{1})\id_{\hil}^{\ot n}\right]\text{,}
\end{align}
where $\chi(\lambda)$ and $\xi(\lambda)$ denote the quadratic Casimir eigenvalues of $\ualg(d)$ or $\So(d)$, respectively, in the irreducible representation labelled by the Young diagram $\lambda$, and $\rep_{\ydiagram{1}}^{\times n}(C^{\ualg(d)})$ and $\repp_{\ydiagram{1}}^{\times n}(C^{\So(d)})$ denote the quadratic Casimir operator in the $n$-fold tensor product of the defining representation of the unitary and the special orthogonal Lie algebra, respectively.

Given the above forms of $\smf_{n}$ and $\smb_{n}$, to do the calculations one needs to know the trace of a given constituent quantum state ($\rho_{\Lambda} \ot \id_{\lil_{\Lambda}}/ \dim(\lil_{\Lambda})$) and the represented Casimir operators above. In this case, this is trivial, as one can simply use Schur-Weyl duality to find the relevant irrep decomposition of $\U(d)$ and then restrict it to $\SO(d)$ (see Appendices \ref{sec:app_rep} and \ref{sec:app_sw}). This will ensure that the irreducible subspaces are exactly labelled by which $\ualg(d)$ and $\So(d)$ irrep they come from. Once the results for every constituent quantum state are known, one can simply take the convex hull of those points to get the result.

Note also that even if multiplicities appear, they are irrelevant as they do not give different results for $\smf_{n}$ and $\smb_{n}$.

In practice, for given extendibility number $n$ and dimension $d$ one goes through all projections onto the invariant subspaces of the appearing $\ualg(d)$ irreps (these are labelled by Young diagrams with $n$ boxes and at most $d$ rows) and calculates the trace of the projection with $\smf_n$ for them. Then, using Appendix \ref{sec:app_rep}, one can find the minimal and maximal $\So(d)$ irrep in the given $\ualg(d)$ irrep, which will give the minimal and maximal values when taking the trace of the projector with $\smb_n$. Thus, a polygon is obtained in the $\tr(\rho\flip)-\tr(\rho\bb)$ parameter space. However, this polygon is generally not convex and one has to take the convex hull of it. Because of this, and given that for large $n$ the number of calculations increases drastically, it is easiest to do the calculations numerically, nevertheless extracting exact results.

\section{Results for $(1,2)$-, $(1,3)$- and $(2,2)$-extendible Brauer states}\label{sec:result1}

As an application of the general recipe given in Section \ref{sec:recipe}, we have explicitly determined the sets of $(1,2)$-, $(1,3)$- and $(2,2)$-extendible Brauer states. In this Section, we present the qualitative results, while the quantitative results can be found in Appendix \ref{sec:app_res}, along with a detailed derivation for the case of $(1,2)$-extendible Brauer states.

First, we present the irreducible decompositions of the representations to which the extending states need to be invariant in the $(1,2)$-, $(1,3)$- and $(2,2)$-extendibility scenarios. This is best done when the dimension is large enough, more precisely when $2(n + m) \leq d$ as the fusion rules of $\U(d)$ and $\Ort(d)$ become independent of the specific dimension by that point.

For the $(1,2)$-extendibility scenario, we have
\begin{equation}
\begin{split}
    \left(\hat{\Repppp}^{\mathrm{L}}_1\Repp_{\ydiagram{1}} \right)\ot \left( \hat{\Repppp}^{\mathrm{R}}_2 \Repp_{\ydiagram{1}}^{\ot 2}\right) \cong \; & \Repppp^{\mathrm{L}}_{\ydiagram{1}} \ot \Repppp^{\mathrm{R}}_{\ydiagram{1,1}} \ot \left(\Repp_{\ydiagram{1}} \OP \Repp_{\ydiagram{1,1,1}} \OP \Repp_{\ydiagram{2,1}} \right) \OP \\
    & \Repppp^{\mathrm{L}}_{\ydiagram{1}} \ot \Repppp^{\mathrm{R}}_{\ydiagram{2}}\ot\left({\color{red}\Repp_{\ydiagram{1}}} \OP {\color{red}\Repp_{\ydiagram{1}}} \OP \Repp_{\ydiagram{2,1}}\OP \Repp_{\ydiagram{3}}\right) \text{,}
\end{split}
\end{equation}
and for the $(1,3)$-extendibility scenario, we have
\begin{equation}
\begin{split}
    \left(\hat{\Repppp}^{\mathrm{L}}_1\Repp_{\ydiagram{1}} \right)\ot \left( \hat{\Repppp}^{\mathrm{R}}_3 \Repp_{\ydiagram{1}}^{\ot 3}\right) \cong \; & \Repppp^{\mathrm{L}}_{\ydiagram{1}} \ot \Repppp^{\mathrm{R}}_{\ydiagram{1,1,1}}\ot \left(\Repp_{\ydiagram{1,1}} \OP \Repp_{\ydiagram{1,1,1,1}} \OP \Repp_{\ydiagram{2,1,1}} \right) \OP \\
    & \Repppp^{\mathrm{L}}_{\ydiagram{1}} \ot \Repppp^{\mathrm{R}}_{\ydiagram{2,1}}\ot \left(\Repp_{0} \OP {\color{orange}\Repp_{\ydiagram{1,1}}}\OP {\color{orange}\Repp_{\ydiagram{1,1}}}\OP {\color{brass}\Repp_{\ydiagram{2}}} \OP {\color{brass}\Repp_{\ydiagram{2}}} \OP \Repp_{\ydiagram{2,1,1}} \OP \Repp_{\ydiagram{2,2}} \OP \Repp_{\ydiagram{3,1}} \right) \OP \\
    & \Repppp^{\mathrm{L}}_{\ydiagram{1}} \ot \Repppp^{\mathrm{R}}_{\ydiagram{3}}\ot \left(\Repp_{0} \OP \Repp_{\ydiagram{1,1}} \OP {\color{brown}\Repp_{\ydiagram{2}}}\OP {\color{brown}\Repp_{\ydiagram{2}}} \OP \Repp_{\ydiagram{3,1}} \OP \Repp_{\ydiagram{4}} \right)\text{.}
\end{split}
\end{equation}
and finally, for the $(2,2)$-extendibility scenario, we have
\begin{equation}
\begin{split}
    \left(\hat{\Repppp}_4 \Rep_{\ydiagram{1}}^{\ot 4}\right)\Big\rvert_{\D_4\times \Ort(d)} \cong \; & \Repppp_{1}\ot\left[{\color{blue}\Repp_{0}}\OP {\color{green}\Repp_{\ydiagram{2}}} \OP \Repp_{\ydiagram{2,2}} \OP {\color{blue}\Repp_{0}} \OP {\color{green}\Repp_{\ydiagram{2}}} \OP \Repp_{\ydiagram{4}}\right] \OP \\
    &\Repppp_{1^{\! *}}\ot \left[\Repp_{\ydiagram{1,1}} \OP \Repp_{\ydiagram{2}} \OP \Repp_{\ydiagram{3,1}}\right] \OP \\
    &\Repppp_{\bar{1}} \ot \left[\Repp_{\ydiagram{1,1,1,1}} \OP \Repp_{0} \OP \Repp_{\ydiagram{2}} \OP \Repp_{\ydiagram{2,2}}\right] \OP \\
    & \Repppp_{\bar{1}^{\! *}}\ot\left[\Repp_{\ydiagram{1,1}} \OP \Repp_{\ydiagram{2,1,1}} \right] \OP \\
    &\Repppp_{2}\ot\left[{\color{blue-green}\Repp_{\ydiagram{1,1}}} \OP \Repp_{\ydiagram{2,1,1}} \OP {\color{blue-green}\Repp_{\ydiagram{1,1}}} \OP \Repp_{\ydiagram{2}} \OP \Repp_{\ydiagram{3,1}} \right]\text{,}
\end{split}
\end{equation}
where the coloured representations have multiplicity greater than one. Note also that in the $(2,2)$-extendibility case, the relevant subgroup of $\sym_4$ is $\D_4$ and its irreps are labelled accordingly, as explained in Section \ref{sec:nn}.

Each of these irreps is linked to a special quantum state (as described in Section \ref{sec:recipe}) with multiplicity one or two. The multiplicity-one irreps are easiest to deal with, and their projection onto the parameter space shows up as a point. On the other hand, multiplicity two irreps are linked to states on a 2-dimensional Hilbert space (qubits), and thus their projection onto the parameter space shows up as the shadow of a Bloch sphere: an ellipse.

Note that since the convex hull of these points and ellipses is needed to describe the subset of $(n,m)$-extendible Brauer states, not every point and ellipse will appear on the boundary, some may be internal points.

Figure \ref{fig:altogether} showcases the results for $d=2,3,4$ and $d\to\infty$. In most cases the sets of $(1,2)$- and $(2,2)$-extendible states coincide, while the set of $(1,3)$-extendible states is always different distinct. However, in the $d=3$ case, all three cases differ, and therefore the $(1,3)$- and the $(2,2)$-extendible subsets are presented next to each other in the appropriate figure. As it turns out, for dimensions $d\geq n+m=4$ the sets of $(1,2)$- and $(2,2)$-extendible Brauer states coincide again.

\begin{figure}[H]
    \begin{subfigure}{\textwidth}
        \centering
        \includegraphics[width=0.6\textwidth]{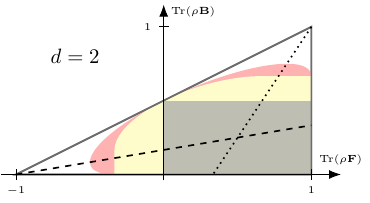}
    \end{subfigure}
    
    \begin{subfigure}{.5\textwidth}
        \centering
        \includegraphics[width=.9\linewidth]{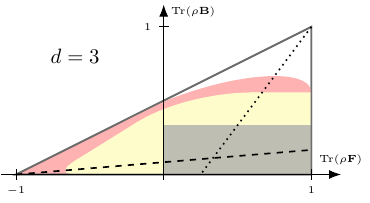}
    \end{subfigure}
    \begin{subfigure}{.5\textwidth}
        \centering
        \includegraphics[width=.9\linewidth]{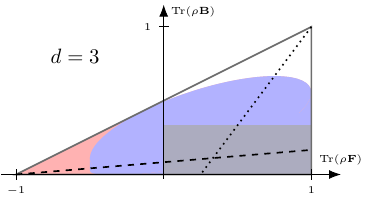}
    \end{subfigure}
    
    \begin{subfigure}{\textwidth}
        \centering
        \includegraphics[width=0.6\textwidth]{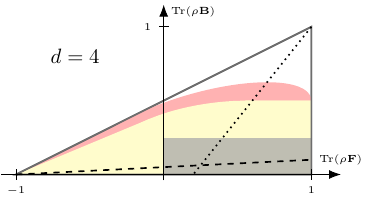}
    \end{subfigure}

    \begin{subfigure}{\textwidth}
        \centering
        \includegraphics[width=0.6\textwidth]{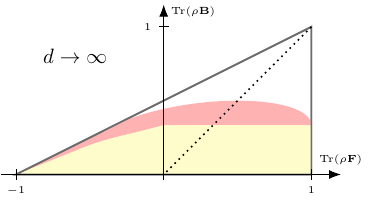}
    \end{subfigure}
    \caption{\small The set of Brauer states at $d=2,3,4$ and $d\to\infty$ in the $\tr(\rho\flip) - \tr(\rho\bb)$ parametrisation appearing as the grey triangle. In most figures subset of $(1,2)$- and $(2,2)$-extendible Brauer states appears in red in the background, and the subset of $(1,3)$-extendible Brauer states appears in yellow in the foreground. However, in the $d=3$ case the subset of $(2,2)$-extendible Brauer states appears in blue as it does not coincide with the subset of $(1,2)$-extendible Brauer states. The set of Werner states is illustrated as the dashed line, while the set of isotropic states as the dotted line. The set of separable Brauer states constitute the grey rectangle. However, in the $d\to\infty$ case the set of Werner states and the set of separable Brauer states have merged into the lower side.}
    \label{fig:altogether}
\end{figure}

In all figures, one should observe that, as expected, the $(1,3)$- and $(2,2)$-extendable Brauer states form a subset of the $(1,2)$-extendible Brauer states. What is more interesting is that the $(1,3)$-extendible Brauer states form a proper subset, while the $(2,2)$-extendible ones do not necessarily. Furthermore, the set of separable states, that is, the $(\infty,\infty)$-extendible Brauer states, is also a proper subset of $(1,3)$-extendible Brauer states, as expected. Note also that the behaviours of the Werner and isotropic states confirm the results of \cite{jakab2022extendibility}.

Observe that every Werner state (illustrated as the dashed line) is $(1,2)$-extendible for $d\geq 3$. Furthermore, they are $(1,3)$- and $(2,2)$-extendible for $d\geq 4$. In general \cite{jakab2022extendibility} confirms that every Werner state is $(n,m)$-extendible if $d\geq n+m$. In our method, this is apparent because the multiplicity-one projector linked to the $\U(d)$ irrep on the totally antisymmetric space appears first when $d=n+m$. This has a singular restriction to the orthogonal group and this irrep contributes the $(-1,0)$ point to the parameter space. Given that the upper-leftmost vertex of the set of separable states, namely $(0,1/d)$, is higher than the respective Werner state at the same $\tr(\rho\flip)=0$ coordinate, namely ($0,1/d(d+1)$), and that the set of $(n,m)$-extendible Brauer states form a convex set, we have that every Werner state is $(n,m)$-extendible for $d\geq n+m$. 

On the other hand, isotropic states have a subset that is not even $(1,2)$-extendible for any $d$, and indeed become less extendible as $d$ increases. This relates to the fact that they are a mixture of the maximally mixed and a maximally entangled state, this latter being unextendible. This is also confirmed by \cite{jakab2022extendibility}.

As a qualitative review of the result, it should also be observed in the last subfigure in Figure \ref{fig:altogether} that the sets do not change substantially as $d$ increases. Indeed, they get to a limiting shape as $d\to \infty$.

\section{Estimate of the set of $(1,m)$-extendible Brauer states}\label{sec:estimate}

As an other application of the general recipe in Section \ref{sec:recipe}, we have determined an estimate for the set of $(1,m)$-extendible Brauer states. The results are presented in this Section and the derivation in Appendix \ref{sec:app_estimate}.

The estimate hinges on identifying the most relevant multiplicity-one irreps in the irrep decomposition to which $(1,m)$-extending states need to be invariant to. This is done by finding a set of Young diagrams (labels) that give extremal contributions to the parameter space. Therefore, we have mostly focused on $\U(d)$ irreps such that their restriction to $\Ort(d)$ contains the trivial representation, as this can only have multiplicity one.

More precisely, if we consider a $\U(d)$ irrep $\Rep_{\nu}$ on $\hil^{\ot m}$ and a corresponding $\U(d)$ irrep $\Rep_{\lambda}$ on $\hil^{\ot (1+m)}$, then we wish to find $\Ort(d)$ irreps $\Repp_{\nu'}$ and $\Repp_{\lambda'}$ such that they appear in the restriction of the respective $\U(d)$ irreps and $\lambda'=[0]$ or $[0]^*$. With this in mind, we find the following set of extremal points:
\begin{align}
    \text{For $m\geq 1$:}\qquad &(1,0)\text{,}\\
    \text{For $m\geq 1$:}\qquad &\begin{cases} \text{if $m$ is odd:} & \left(1,\frac{m+d-1}{dm}\right)\text{,}\\ \text{if $m$ is even:} & \left(1,\frac{m+d}{d(m+1)}\right)\text{,}\end{cases}\\
    \text{For $m<d-1$:} \qquad &(-1,0)\text{,}\\
    \text{For $m\geq d-1$:}\qquad &\begin{cases} \text{if $m-d$ is odd:} & \left(\frac{1-d}{m},0\right)\text{,}\\ \text{if $m-d$ is even:} & \left(\frac{1-d}{m+1},0\right)\text{,}\end{cases}\\
    \text{For $m<2d-2$:}\qquad &\begin{cases}\text{if $m$ is odd:} & \left(\frac{3-m}{2m},\frac{2d-m+1}{2dm}\right)\text{,} \\ \text{if $m$ is even:} & \left(\frac{2-m}{2(m+1)},\frac{2d-m}{2d(m+1)}\right)\text{,} \end{cases}\\
    \text{For $m\geq 2d-2$:}\qquad &\begin{cases}\text{if $m$ is odd:} & \left(\frac{2-d}{m},\frac{1}{dm}\right)\text{,}\\ \text{if $m$ is even} & \left(\frac{2-d}{m+1},\frac{1}{d(m+1)}\right)\text{.} \end{cases}
\end{align}
Observe that in the cases of $m$ being even, or $m-d$ being even, the best estimates are the results for $m+1$. Thus, the sets of $(1,m)$-extendible and $(1,m+1)$-extendible Brauer states often coincide in our estimate for even $m$-s.

Figure \ref{fig:12_13_estimate} showcases the validity of the estimate by comparing it to the exact results for $(1,2)$- and $(1,3)$-extendible Brauer states. Observe that the estimates for $m=2$ and $m=3$ disagree only when $d=3$. The estimate perfectly describes where the sets of $(1,2)$- and $(1,3)$-extendible Brauer states meet the legs of the right triangle. However, note that for $d<4$ the estimate does not include the set of separable states in its entirety.

\vspace{-0.3cm}

\begin{figure}[H]
    \begin{subfigure}{.5\textwidth}
      \centering
      \includegraphics[width=\linewidth]{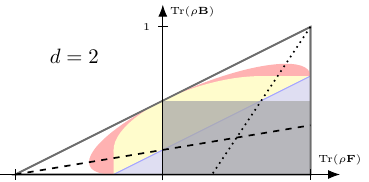}
    \end{subfigure}\hfill
    \begin{subfigure}{.5\textwidth}
      \centering
      \includegraphics[width=\linewidth]{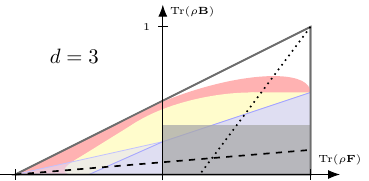}
    \end{subfigure}
    
    \begin{subfigure}{.5\textwidth}
      \centering
      \includegraphics[width=\linewidth]{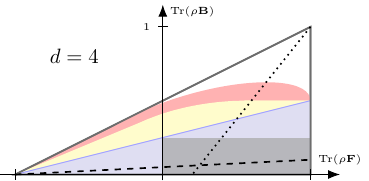}
    \end{subfigure}\hfill
    \begin{subfigure}{.5\textwidth}
      \centering
      \includegraphics[width=\linewidth]{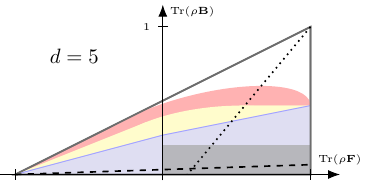}
    \end{subfigure}
    \caption{\small The set of Brauer states at $d=2,3,4$ and $5$ in the $\tr(\rho\flip) - \tr(\rho\bb)$ parametrisation appearing as the grey triangle. The subset of $(1,2)$-extendible Brauer states appears in red in the background, and the subset of $(1,3)$-extendible Brauer states appears in yellow in its foreground. The estimates for the sets of $(1,2)$- and $(1,3)$-extendible Brauer states, appear as the blue polygons with differing shades, contained in each other. The set of Werner states is illustrated as the dashed line, while the set of isotropic states as the dotted line. The set of separable Brauer states constitute the grey rectangle.}
  \label{fig:12_13_estimate}
\end{figure}

\vspace{-0.6cm}

A qualitative presentation of the results for $m\in[2,10]$ and dimensions $d=3$ and $10$, and $d\to\infty$ can be found in Figure \ref{fig:1m_d3_d10_dinfty}.

\vspace{-0.5cm}

\begin{figure}[H]
        \centering
        \includegraphics[width=0.65\textwidth]{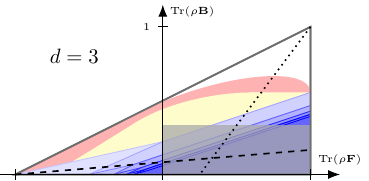} 
\end{figure}
    
\begin{figure}[H]
    \centering
    \includegraphics[width=0.65\textwidth]{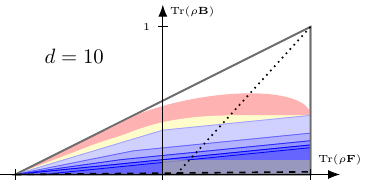}
\end{figure}

\begin{figure}[H]
    \centering
    \includegraphics[width=0.65\textwidth]{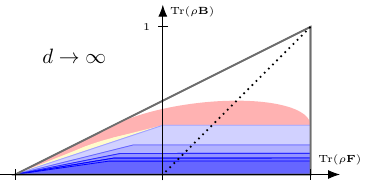} 
    \caption{\small The set of Brauer states at $d=3,10$ and $d\to\infty$ in the $\tr(\rho\flip) - \tr(\rho\bb)$ parametrisation appearing as the grey triangle. The subset of $(1,2)$-extendible Brauer states appears in red in the background, and the subset of $(1,3)$-extendible Brauer states appears in yellow in its foreground. The estimates of the sets of $(1,m)$-extendible Brauer states, $m\in[2,10]$, appear as the blue polygons with differing shades, contained in each other. The set of Werner states is illustrated as the dashed line, while the set of isotropic states as the dotted line. The set of separable Brauer states constitute the grey rectangle. However, in the $d\to\infty$ case the set of Werner states and the set of separable Brauer states have merged into the lower side.}
    \label{fig:1m_d3_d10_dinfty}
\end{figure}

\vspace{-0.5cm}

Note that, for fixed and odd $m$, in the $d\to\infty$ limit, we obtain a polygon with vertices $(-1,0)$, $\left(\frac{3-m}{2m},\frac{1}{m}\right)$, $\left(1,\frac{1}{m}\right)$, and $(1,0)$. On the other hand, for fixed dimension $d$, in the $m\to\infty$ limit, we obtain three of the four vertices of the rectangle of separable states: $(0,0)$, $(1,0)$, and $(1,1/d)$. However, the upper left corner, $(0,1/d)$, is elusive and most likely always corresponds to irreps with multiplicity greater than one, as attested by the exact results for $m=2$ and $m=3$.

Thus, the estimate could be improved by including the point corresponding to the upper left corner of the set of separable states and the points coming from the known results for $(1,m)$-extendible Werner and isotropic states \cite{jakab2022extendibility}. Furthermore, the estimate might also be improved by considering less obvious multiplicity-one irreps. Nevertheless, it gives a good first-order approximation to the set of $(1,m)$-extendible Brauer states.

\section{Results for $n$-de~Finetti-extendible Brauer states}\label{sec:result2}

\ytableausetup{mathmode, boxsize=1em, aligntableaux=center}

In the following, the exact results for the de~Finetti extendibility of Brauer states are presented.

\subsection{The finite $n$ case in dimension $d$}

As was noted in Section \ref{sec:defi}, for every $n$ and $d$ one has to take every Young diagram (taken as a label for a $\U(d)$ irrep), with $n$ boxes and at most $d$ rows, and determine the $\Ort(d)$ irrep(s) in its restriction, which are related to $\So(d)$ irreps with minimal and maximal quadratic Casimir eigenvalues. The following theorem contains our main result in relation to this issue:

\begin{restatable}[]{theorem}{minmax}\label{th:minmax}
    Let $\lambda$ be an at most $d$-row Young diagram labelling a $\U(d)$ irrep $\Rep_\lambda$, such that:
    \begin{equation}
    \lambda=[\lambda_1,\; \ldots,\; \lambda_d]=[2x_1+y_1,\; \ldots,\; 2x_d+y_d]\text{,}
    \end{equation}
    where, for all $i\in[d]$, the numbers $x_i$ are non-negative integers and $y_i\in\{0,1\}$.
    Then, the labels $\mu$ of the $\Ort(d)$ irreps in the restriction $\Rep_\lambda\big|_{\Ort(d)}$ corresponding to the $\So(d)$ irrep(s) with minimal or maximal quadratic Casimir eigenvalue, are the following:
    \begin{align}
    \text{Minimal: }\qquad & \mu=[1^{\sum_{i=1}^{d} y_i}]\text{,}\\
    \text{Maximal: }\qquad& \mu=\begin{cases}[\lambda_1-\lambda_d, \; \ldots, \; \lambda_{i}-\lambda_{d-i+1}, \; \ldots,\; \lambda_{k}-\lambda_{k+1}]^{a}\text{,} & \text{if $d=2k$,}\\
    [\lambda_1-\lambda_d, \; \ldots, \; \lambda_{i}-\lambda_{d-i+1}, \; \ldots, \; \lambda_{k}-\lambda_{k+2}]^{a}\text{,} & \text{if $d=2k+1$}\text{,}\end{cases}
    \end{align}
    where the notation $[1^{N}]$ describes an $N$-tall column of single boxes. Note that the label of the minimal irrep might need modification, while for the maximal case $a=\ast$, meaning that we are dealing with the associated representation, if $\lambda_{k+1}$ is odd, otherwise we are not dealing with the associated representation.
\end{restatable}

The proof of this theorem is lengthy and therefore can be found in Appendix \ref{sec:app_proof}. Illustratively, for the diagram labelling the irrep(s) with minimal quadratic Casimir eigenvalue, the theorem states that we should form a column using the single boxes in $\lambda$ that we find at the end of rows with an odd number of boxes in them. This relates to the fact that the $\Ort(d)$ irrep related to the $\So(d)$ irrep(s) with minimal quadratic Casimir eigenvalue is the one where in the restriction formula, Equation \eqref{eq:littlewood}, we use the maximal $2\kappa$ diagram to `fill out' $\lambda$ as much as possible. With a concrete Young diagram $\lambda=[5,4,4,2,1]$ we have the following:

\vspace{-0.1cm}

\begin{equation}
    \lambda=\ydiagram{5,4,4,2,1} \qquad \to \qquad \text{Minimal:} \quad
    \begin{ytableau}
    *(gray)\empty & *(gray)\empty & *(gray)\empty & *(gray)\empty & \empty\\
    *(gray)\empty & *(gray)\empty & *(gray)\empty & *(gray)\empty \\
    *(gray)\empty & *(gray)\empty & *(gray)\empty & *(gray)\empty \\
    *(gray)\empty & *(gray)\empty \\
    \empty \\
    \end{ytableau} \sim \; \ydiagram{1,1}=\mu\text{.}
\end{equation}

On the other hand, for the diagram labelling the irrep(s) with maximal quadratic Casimir eigenvalue, the theorem states that we should fold the Young diagram $\lambda$ in half such that the $d$-th row of $\lambda$ overlaps with its first row. Any overlaps cancel each other out, and the rows we are left with are the rows of the Young diagram labelling the desired irrep. Using the previous example, in dimension $d=5$ and $d=6$ we have:
\begin{equation}
    \lambda=\ydiagram{5,4,4,2,1} \qquad \to \qquad \text{Maximal:} \quad
    \begin{ytableau}
    \empty & \empty & \empty & \empty & \empty\\
    \empty & \empty & \empty & \empty \\
    \halfgraybox & \halfgraybox & \halfgraybox & \halfgraybox \\
    *(gray!50)\empty & *(gray!50)\empty \\
    *(gray!50)\empty \\
    \end{ytableau} \; \foldarrowvertical[black]{fold} \sim \;
    \begin{ytableau}
    *(gray)\empty & \empty & \empty & \empty & \empty\\
    *(gray)\empty & *(gray)\empty & \empty & \empty \\
    \none[\halfgraydashedbox] & \none[\halfgraydashedbox] & \none[\halfgraydashedbox] & \none[\halfgraydashedbox] \\
    \none[\dashedbox] & \none[\dashedbox] \\
    \none[\dashedbox] \\
    \end{ytableau} \sim \;
    \ydiagram{4,2}=\mu\text{,}
\end{equation}
\begin{equation}
    \lambda=\ydiagram{5,4,4,2,1} \qquad \to \qquad \text{Maximal:} \quad
    \begin{ytableau}
    \empty & \empty & \empty & \empty & \empty\\
    \empty & \empty & \empty & \empty \\
    \empty & \empty & \empty & \empty \\
    *(gray!50)\empty & *(gray!50)\empty \\
    *(gray!50)\empty \\
    \none[\dashedbox]\\
    \end{ytableau} \; \foldarrowvertical[black]{fold} \sim \;
    \begin{ytableau}
    \empty & \empty & \empty & \empty & \empty\\
    *(gray)\empty & \empty & \empty & \empty \\
    *(gray)\empty & *(gray)\empty & \empty & \empty \\
    \none[\dashedbox] & \none[\dashedbox] \\
    \none[\dashedbox]\\
    \none[\dashedbox]\\
    \end{ytableau} \sim \;
    \ydiagram{5,3,2}=\mu\text{.}
\end{equation}

In conclusion, for each $n$-box $\U(d)$ irrep label $\lambda$ one can determine the $\tr(\rho\flip)$ parameter and, the related minimal and maximal $\tr(\rho\bb)$ parameter. However, as was also stated in Section \ref{sec:defi}, the resulting polygon is generally not convex, and thus one needs to take its convex hull, which is easier to do numerically.

In Figure \ref{fig:definetti_d2_d5_d10} we illustrate the exact results for the sets of $n$-de~Finetti-extendible Brauer states for $n=3,4,5,20$ and dimensions $d=2,5,10$, although the calculations can be done for any $n$ and any $d$. One should notice that a limiting shape appears as $n$ grows larger, which is discussed in the next section.

\vspace{-0.3cm}

\begin{figure}[H]
    \begin{subfigure}{.5\textwidth}
      \centering
      \includegraphics[width=\linewidth]{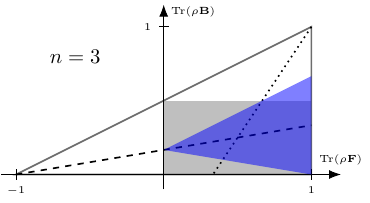}
    \end{subfigure}
    \hfill
    \begin{subfigure}{.5\textwidth}
      \centering
      \includegraphics[width=\linewidth]{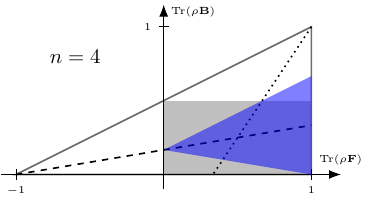}
    \end{subfigure}
    
    \begin{subfigure}{.5\textwidth}
      \centering
      \includegraphics[width=\linewidth]{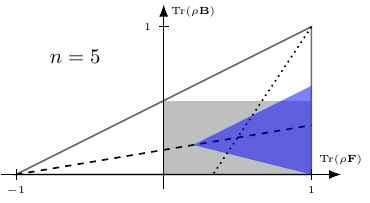}
    \end{subfigure}
    \hfill
    \begin{subfigure}{.5\textwidth}
      \centering
      \includegraphics[width=\linewidth]{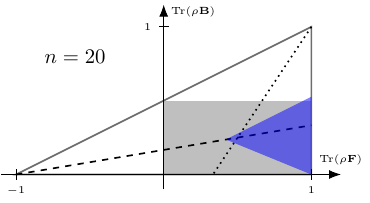}
    \end{subfigure}
\end{figure}

\begin{figure}[H]
    \begin{subfigure}{.5\textwidth}
      \centering
      \includegraphics[width=\linewidth]{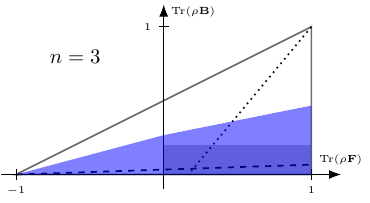}
    \end{subfigure}
    \begin{subfigure}{.5\textwidth}
      \centering
      \includegraphics[width=\linewidth]{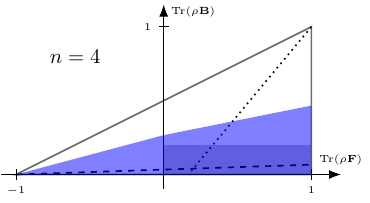}
    \end{subfigure}
    
    \begin{subfigure}{.5\textwidth}
      \centering
      \includegraphics[width=\linewidth]{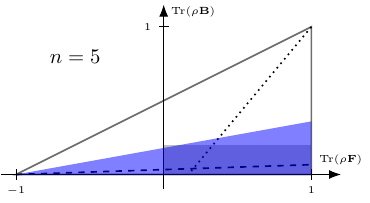}
    \end{subfigure}
    \begin{subfigure}{.5\textwidth}
      \centering
      \includegraphics[width=\linewidth]{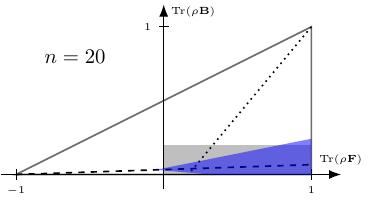}
    \end{subfigure}
\end{figure}

\begin{figure}[H]
    \begin{subfigure}{.5\textwidth}
      \centering
      \includegraphics[width=\linewidth]{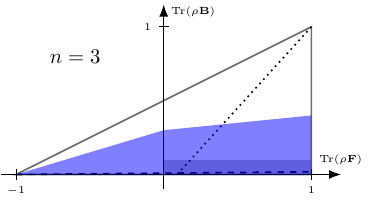}
    \end{subfigure}
    \begin{subfigure}{.5\textwidth}
      \centering
      \includegraphics[width=\linewidth]{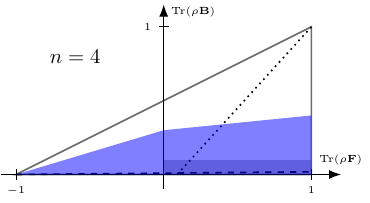}
    \end{subfigure}

    \begin{subfigure}{.5\textwidth}
      \centering
      \includegraphics[width=\linewidth]{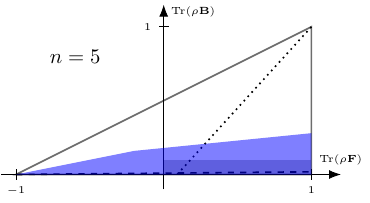}
    \end{subfigure}
    \begin{subfigure}{.5\textwidth}
      \centering
      \includegraphics[width=\linewidth]{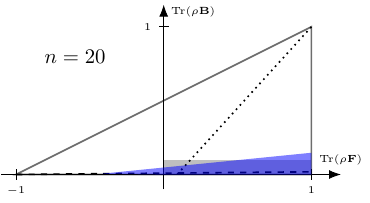}
    \end{subfigure}
    \caption{\small The set of Brauer states at $d=2,5,10$ in the $\tr(\rho\flip) - \tr(\rho\bb)$ parametrisation appearing as the grey triangle, grouped by dimension into sets of four. The subset of $n$-de~Finetti-extendible Brauer states is presented in blue for $n=3,4,5,20$. The set of Werner states is illustrated as the dashed line, while the set of isotropic states as the dotted line. The set of separable Brauer states constitute the grey rectangle.}
    \label{fig:definetti_d2_d5_d10}
\end{figure}

\vspace{-1cm}

\subsection{The $n\to\infty$ case in dimension $d$}

\ytableausetup{mathmode, boxsize=\Yt}

Based on the results of the previous section, one concludes that a limiting shape is reached as $n$ grows larger. For example, the limiting shape for $d=3$ is presented in Figure \ref{fig:definetti_limit}, circumscribed by three straight lines (in black, yellow, and red in the figure) and a non-trivial function (green in the figure).

\begin{figure}[H]
    \centering
    \includegraphics[width=0.6\textwidth]{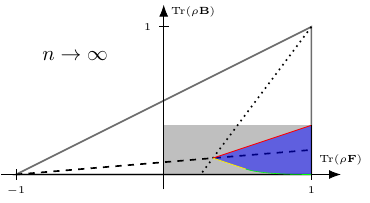}
    \caption{\small The set of Brauer states at $d=3$ in the $\tr(\rho\flip) - \tr(\rho\bb)$ parametrisation appearing as the grey triangle. The subset of $1000$-de~Finetti-extendible Brauer states is presented in blue. The set of Werner states is illustrated as the dashed line, while the set of isotropic states as the dotted line. The set of separable Brauer states constitute the grey rectangle. The area circumscribed by the red-black-green-yellow lines is the limiting shape of the $n\to\infty$ case.}
    \label{fig:definetti_limit}
\end{figure}

\vspace{-0.5cm}

The form of the limiting shape in any dimension can be obtained by the following argument: From the quantum de~Finetti theorem, we know that only Brauer states $\brauer$ that can be decomposed as convex combinations of flip-invariant product states $\brauer=\sum_i q_i \sigma_i \otimes \sigma_i $ (with $q_i \ge 0$ and $\sum_i q_i =1 $) can belong to this set. However, note that the $\Repp_{\ydiagram{1}}(g)\ot\Repp_{\ydiagram{1}}(g)$-twirl on a product state $\sigma_i \otimes \sigma_i $ yields a Brauer state $\brauer^{(i)}$ that has the same $\flip$ and $\bb$ expectation values as before the twirl, i.e., $\tr(\flip\sigma_i \otimes \sigma_i) = \tr(\flip\brauer^{(i)})$ and $\tr(\bb\sigma_i \otimes \sigma_i) = \tr(\bb\brauer^{(i)})$. This means that for each extremal point of the set of infinitely de~Finetti-extendible Brauer states there exists a flip-invariant product state with the given expectation values.

By taking the eigendecomposition of one part ($\sigma$) of the flip invariant product states ($\sigma\ot\sigma$) we can realise that the $\tr(\flip \sigma \ot \sigma)$ value is the sum of the squared eigenvalues, while the $\tr(\bb \sigma\ot\sigma)$ value can be linked to the quadratic form of a doubly stochastic matrix with the vector of eigenvalues. By the Birkhoff-von Neumann theorem any doubly stochastic matrix can be written as the convex sum of permutation matrices which leads to the realisation that the $\tr(\bb\sigma\ot\sigma)$ value can be be upper and lower bounded by using the rearrangement inequalities. See Appendix \ref{sec:app_definetti_infty} for a detailed proof.

Applying this line of reasoning, the following results can be obtained for the upper and lower limits:
\begin{minipage}{0.22\textwidth}
    \begin{flalign}
    \begin{aligned}
      & b_{\text{upper}}(f)= \frac{f}{d}\text{,} \\
      & \text{if $f\in\left[\frac{1}{d},1\right]$,} &
    \end{aligned}
    \end{flalign}
\end{minipage}
\hfill
\begin{minipage}{0.7\textwidth}
    \begin{align}
        b_{\text{lower}}^{d\text{ even}}(f)=&\begin{cases} \frac{2-df}{d^2}\text{,} & \text{if $f\in\left[\frac{1}{d},\frac{2}{d}\right]$,} \\
        0\text{,} & \text{if $f\in\left[\frac{2}{d},1\right]$,} \end{cases}\\
        b_{\text{lower}}^{d\text{ odd}}(f)=&\begin{cases} \frac{2-df}{d^2}\text{,} & \text{if $f\in\left[\frac{1}{d},\frac{2d-1}{d^2}\right]$,} \\
         \frac{\left(2-\sqrt{(d^2-1)f-2(d-1)}\right)^2}{d(d+1)^2}\text{,} & \text{if $f\in\left[\frac{2d-1}{d^2},\frac{2}{d-1}\right]$,} \vspace{0.2cm} \\
        0\text{,} & \text{if $f\in\left[\frac{2}{d-1},1\right]$.} \end{cases}
    \end{align}
\end{minipage}

\vspace{0.2cm}

Interestingly, in the case of odd dimensions, one obtains not a polygon but a shape which has some flat sides and one described by a non-trivial function. Observe Figure \ref{fig:definetti_infty} for a parametric illustration of the results. The upper limiting line is drawn in red in both the even and the odd cases. The lower limiting line has two parts for the even case, illustrated in yellow and blue, while for the odd case these are yellow, green, and blue. The same colour corresponds to the same type of limiting line. The black vertical line represents the convex closure of the area. As expected, the upper and lower limits meet at the maximally mixed state.

With this, we have answered Conjecture 4.16.~from \cite{allerstorfer2023monogamy}, in which they posited that in the $n\to\infty$ limit the set of $n$-de~Finetti-extendible Brauer states is not a polytope. Our answer proves their conjecture with the added subtlety that the conjecture is only true for odd dimensions.

Given the above results, one can draw the conclusion that Werner states with $\tr(\wern\flip)\geq 1/d$ are de~Finetti extendible for any $n$, while no non-trivial isotropic state is de~Finetti extendible for any $n$.

\begin{figure}[H]
    \centering
    \includegraphics[width=0.95\textwidth]{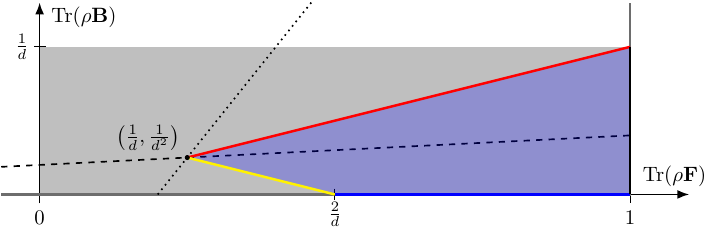}\vspace{0.2cm}
        \includegraphics[width=0.95\textwidth]{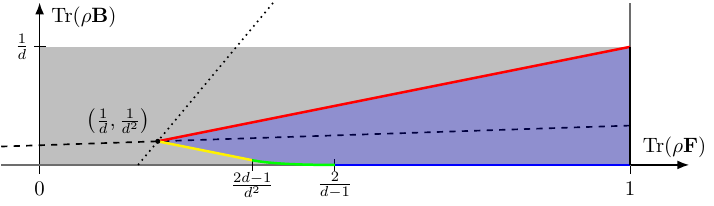}
    \caption{\small Parametric figures describing the set of de~Finetti-extendible Brauer states in the $n\to\infty$ limit in the $\tr(\rho\flip) - \tr(\rho\bb)$ parametrisation, for respectively the even- and odd-dimensional cases. The figures only contain the neighbourhood of the set of separable states which constitute the grey rectangle. The upper boundary of the set is described by the red line in both cases. The lower boundary is described by the yellow and blue straight lines in the even case, and a further non-trivial green line in the odd case. The states that are de~Finetti extendible for any $n$ are in the blue shaded area circumscribed by these lines. The set of Werner states is illustrated as the dashed line, while the set of isotropic states as the dotted line.}
    \label{fig:definetti_infty}
\end{figure}

\vspace{-1cm}

\section{Discussion}

We have presented a recipe for determining the set of $(n,m)$-extendible Brauer states. To demonstrate this method, we have calculated the set of $(1,2)$-, $(1,3)$- and $(2,2)$-extendible Brauer states for any local dimension $d$ and have provided an estimate of the set of $(1,m)$-extendible Brauer states in any dimension $d$. Additionally, we have put forward a method to determine the set of $n$-de~Finetti-extendibile Brauer states for any $n$ and dimension $d$, and we have presented results for the low-dimensional ($d=2,5,10$) cases. Furthermore, we have analytically described the limiting set of Brauer states of any dimension $d$ which are $n$-de~Finetti-extendible for any $n$. Our methods relied heavily on the concrete representation theory of the unitary, orthogonal, and symmetric groups.

Our findings for the two-sided extendibility are in line with our previous result in \cite{jakab2022extendibility}. We could confirm that for large enough dimension ($d\geq n+m$) all Werner states are $(n,m)$-extendible. We could also confirm the triviality that isotropic states have a subset that is not extendible at all, for any dimension $d$, which is in line with the fact that they contain a maximally entangled state.

We have made the interesting observation that as the local dimension increases, the sets of $(1,2)$-, $(1,3)$- and $(2,2)$-extendible Brauer states do not change significantly, while the set of separable Brauer states shrinks rapidly. It is well known that as the dimension increases, a randomly selected quantum state is more likely to be entangled than not. Therefore, the shrinking of the set of separable states is expected, however, interestingly, the degree of extendibility does not change significantly.

For the case of $n$-de~Finetti-extendibility, our results are in accordance with expectations. Our most notable result is an analytic description of the limiting shape in the $n\to \infty$ case which we derive from first principles. It provides the answer to Conjecture 4.16.~from \cite{allerstorfer2023monogamy}.

Further research directions are abundant. It would be interesting to push the extendibility numbers further. Given the relationship between the $(1,2)$- and $(2,2)$-extendible Brauer states, one is tempted to make a weak conjecture that for large enough $d$ all $(1,n)$-extendible Brauer states are also $(n,n)$-extendible. By preforming the calculations for the $(2,3)$- and $(3,3)$-extendibility cases, this conjecture could be destroyed or buttressed.

As is always the case with extendibility problems, it would also be interesting to define tripartite Brauer states and examine the multipartite extendibility problem for them. 

Another direction of research would be that of states invariant to the unitary symplectic group. The representation theory of this group is very similar to that of the orthogonal one, and such bipartite states are also described by two parameters, which further underlines the parallels with Brauer states. Some authors of this contribution are currently working together with authors from \cite{allerstorfer2023monogamy} to examine this research direction.

\section{Acknowledgement}
The authors thank D.~Grinko for engaging talks on the topic.
A.S.~and Z.Z.~express gratitude to I.~Marvian for hosting them at Duke University and for the stimulating discussions on extendibility. A.S.~also thanks I.~Vona for several helpful technical discussions on symbolic computation. 

This research was supported by the Ministry of Culture and Innovation and the National Research, Development and Innovation Office through the Quantum Information National Laboratory of Hungary (Grant No.~2022-2.1.1-NL-2022-00004), and
the grants TKP-2021-NVA-29 and FK 135220. Z.Z.
was also partially supported by the
QuantERA II project HQCC-101017733.

\bibliography{biblio.bib}


\begin{thebibliography}{73}
\ifx \bisbn   \undefined \def \bisbn  #1{ISBN #1}\fi
\ifx \binits  \undefined \def \binits#1{#1}\fi
\ifx \bauthor  \undefined \def \bauthor#1{#1}\fi
\ifx \batitle  \undefined \def \batitle#1{#1}\fi
\ifx \bjtitle  \undefined \def \bjtitle#1{#1}\fi
\ifx \bvolume  \undefined \def \bvolume#1{\textbf{#1}}\fi
\ifx \byear  \undefined \def \byear#1{#1}\fi
\ifx \bissue  \undefined \def \bissue#1{#1}\fi
\ifx \bfpage  \undefined \def \bfpage#1{#1}\fi
\ifx \blpage  \undefined \def \blpage #1{#1}\fi
\ifx \burl  \undefined \def \burl#1{\textsf{#1}}\fi
\ifx \doiurl  \undefined \def \doiurl#1{\url{https://doi.org/#1}}\fi
\ifx \betal  \undefined \def \betal{\textit{et al.}}\fi
\ifx \binstitute  \undefined \def \binstitute#1{#1}\fi
\ifx \binstitutionaled  \undefined \def \binstitutionaled#1{#1}\fi
\ifx \bctitle  \undefined \def \bctitle#1{#1}\fi
\ifx \beditor  \undefined \def \beditor#1{#1}\fi
\ifx \bpublisher  \undefined \def \bpublisher#1{#1}\fi
\ifx \bbtitle  \undefined \def \bbtitle#1{#1}\fi
\ifx \bedition  \undefined \def \bedition#1{#1}\fi
\ifx \bseriesno  \undefined \def \bseriesno#1{#1}\fi
\ifx \blocation  \undefined \def \blocation#1{#1}\fi
\ifx \bsertitle  \undefined \def \bsertitle#1{#1}\fi
\ifx \bsnm \undefined \def \bsnm#1{#1}\fi
\ifx \bsuffix \undefined \def \bsuffix#1{#1}\fi
\ifx \bparticle \undefined \def \bparticle#1{#1}\fi
\ifx \barticle \undefined \def \barticle#1{#1}\fi
\bibcommenthead
\ifx \bconfdate \undefined \def \bconfdate #1{#1}\fi
\ifx \botherref \undefined \def \botherref #1{#1}\fi
\ifx \url \undefined \def \url#1{\textsf{#1}}\fi
\ifx \bchapter \undefined \def \bchapter#1{#1}\fi
\ifx \bbook \undefined \def \bbook#1{#1}\fi
\ifx \bcomment \undefined \def \bcomment#1{#1}\fi
\ifx \oauthor \undefined \def \oauthor#1{#1}\fi
\ifx \citeauthoryear \undefined \def \citeauthoryear#1{#1}\fi
\ifx \endbibitem  \undefined \def \endbibitem {}\fi
\ifx \bconflocation  \undefined \def \bconflocation#1{#1}\fi
\ifx \arxivurl  \undefined \def \arxivurl#1{\textsf{#1}}\fi
\csname PreBibitemsHook\endcsname

\bibitem[\protect\citeauthoryear{Werner}{1989}]{werner1989application}
\begin{barticle}
\bauthor{\bsnm{Werner}, \binits{R.F.}}:
\batitle{An application of {B}ell's inequalities to a quantum state extension
  problem}.
\bjtitle{Lett. Math. Phys.}
\bvolume{17},
\bfpage{359}--\blpage{363}
(\byear{1989})
\doiurl{10.1103/PhysRevLett.123.070502}
\end{barticle}
\endbibitem

\bibitem[\protect\citeauthoryear{Terhal et~al.}{2003}]{terhal2003symmetric}
\begin{barticle}
\bauthor{\bsnm{Terhal}, \binits{B.M.}},
\bauthor{\bsnm{Doherty}, \binits{A.C.}},
\bauthor{\bsnm{Schwab}, \binits{D.}}:
\batitle{Symmetric extensions of quantum states and local hidden variable
  theories}.
\bjtitle{Phys. Rev. Lett.}
\bvolume{90},
\bfpage{157903}
(\byear{2003})
\doiurl{10.1103/PhysRevLett.90.157903}
\end{barticle}
\endbibitem

\bibitem[\protect\citeauthoryear{Doherty}{2014}]{doherty2014entanglement}
\begin{barticle}
\bauthor{\bsnm{Doherty}, \binits{A.C.}}:
\batitle{Entanglement and the shareability of quantum states}.
\bjtitle{J. Phys. A: Math. Theor.}
\bvolume{47},
\bfpage{424004}
(\byear{2014})
\doiurl{10.1088/1751-8113/47/42/424004}
\end{barticle}
\endbibitem

\bibitem[\protect\citeauthoryear{Christandl and
  Winter}{2004}]{christandl2004squashed}
\begin{barticle}
\bauthor{\bsnm{Christandl}, \binits{M.}},
\bauthor{\bsnm{Winter}, \binits{A.}}:
\batitle{{“Squashed entanglement”: An additive entanglement measure}}.
\bjtitle{J. Math. Phys.}
\bvolume{45},
\bfpage{829}--\blpage{840}
(\byear{2004})
\doiurl{10.1063/1.1643788}
\end{barticle}
\endbibitem

\bibitem[\protect\citeauthoryear{Allerstorfer
  et~al.}{2024}]{allerstorfer2023monogamy}
\begin{botherref}
\oauthor{\bsnm{Allerstorfer}, \binits{R.}},
\oauthor{\bsnm{Christandl}, \binits{M.}},
\oauthor{\bsnm{Grinko}, \binits{D.}},
\oauthor{\bsnm{Nechita}, \binits{I.}},
\oauthor{\bsnm{Ozols}, \binits{M.}},
\oauthor{\bsnm{Rochette}, \binits{D.}},
\oauthor{\bsnm{Lunel}, \binits{P.V.}}:
Monogamy of highly symmetric states
(2024).
\url{https://arxiv.org/abs/2309.16655}
\end{botherref}
\endbibitem

\bibitem[\protect\citeauthoryear{Johnson and
  Viola}{2013}]{johnson2013compatible}
\begin{barticle}
\bauthor{\bsnm{Johnson}, \binits{P.D.}},
\bauthor{\bsnm{Viola}, \binits{L.}}:
\batitle{Compatible quantum correlations: Extension problems for {W}erner and
  isotropic states}.
\bjtitle{Phys. Rev. A}
\bvolume{88},
\bfpage{032323}
(\byear{2013})
\doiurl{10.1103/PhysRevA.88.032323}
\end{barticle}
\endbibitem

\bibitem[\protect\citeauthoryear{Fannes et~al.}{1988}]{fannes1988symmetric}
\begin{barticle}
\bauthor{\bsnm{Fannes}, \binits{M.}},
\bauthor{\bsnm{Lewis}, \binits{J.T.}},
\bauthor{\bsnm{Verbeure}, \binits{A.}}:
\batitle{Symmetric states of composite systems}.
\bjtitle{Lett. Math. Phys.}
\bvolume{15},
\bfpage{255}--\blpage{260}
(\byear{1988})
\doiurl{10.1007/BF00398595}
\end{barticle}
\endbibitem

\bibitem[\protect\citeauthoryear{Raggio and Werner}{1989}]{raggio1989quantum}
\begin{barticle}
\bauthor{\bsnm{Raggio}, \binits{G.A.}},
\bauthor{\bsnm{Werner}, \binits{R.F.}}:
\batitle{Quantum statistical mechanics of general mean field systems}.
\bjtitle{Helv. Phys. Acta}
\bvolume{62},
\bfpage{980}--\blpage{1003}
(\byear{1989})
\doiurl{10.5169/seals-116175}
\end{barticle}
\endbibitem

\bibitem[\protect\citeauthoryear{Schumacher}{}]{schumacher2002conjecture}
\begin{botherref}
\oauthor{\bsnm{Schumacher}, \binits{B.}}
conjectured in private communication to B.~M.~Terhal. See Reference 11.
\end{botherref}
\endbibitem

\bibitem[\protect\citeauthoryear{Doherty et~al.}{2004}]{doherty2004complete}
\begin{barticle}
\bauthor{\bsnm{Doherty}, \binits{A.C.}},
\bauthor{\bsnm{Parrilo}, \binits{P.A.}},
\bauthor{\bsnm{Spedalieri}, \binits{F.M.}}:
\batitle{Complete family of separability criteria}.
\bjtitle{Phys. Rev. A}
\bvolume{69},
\bfpage{022308}
(\byear{2004})
\doiurl{10.1103/PhysRevA.69.022308}
\end{barticle}
\endbibitem

\bibitem[\protect\citeauthoryear{Terhal}{2004}]{terhal2004entanglement}
\begin{barticle}
\bauthor{\bsnm{Terhal}, \binits{B.M.}}:
\batitle{Is entanglement monogamous?}
\bjtitle{IBM Journal of Research and Development}
\bvolume{48},
\bfpage{71}--\blpage{78}
(\byear{2004})
\doiurl{10.1147/rd.481.0071}
\end{barticle}
\endbibitem

\bibitem[\protect\citeauthoryear{Koashi and Winter}{2004}]{koashi2004monogamy}
\begin{barticle}
\bauthor{\bsnm{Koashi}, \binits{M.}},
\bauthor{\bsnm{Winter}, \binits{A.}}:
\batitle{Monogamy of quantum entanglement and other correlations}.
\bjtitle{Phys. Rev. A}
\bvolume{69},
\bfpage{022309}
(\byear{2004})
\doiurl{10.1103/PhysRevA.69.022309}
\end{barticle}
\endbibitem

\bibitem[\protect\citeauthoryear{Werner}{1990}]{werner1990remarks}
\begin{barticle}
\bauthor{\bsnm{Werner}, \binits{R.F.}}:
\batitle{Remarks on a quantum state extension problem}.
\bjtitle{Lett. Math. Phys.}
\bvolume{19},
\bfpage{319}--\blpage{326}
(\byear{1990})
\doiurl{10.1007/BF00429951}
\end{barticle}
\endbibitem

\bibitem[\protect\citeauthoryear{Doherty
  et~al.}{2002}]{doherty2002distinguishing}
\begin{barticle}
\bauthor{\bsnm{Doherty}, \binits{A.C.}},
\bauthor{\bsnm{Parrilo}, \binits{P.A.}},
\bauthor{\bsnm{Spedalieri}, \binits{F.M.}}:
\batitle{Distinguishing separable and entangled states}.
\bjtitle{Phys. Rev. Lett.}
\bvolume{88},
\bfpage{187904}
(\byear{2002})
\doiurl{10.1103/PhysRevLett.88.187904}
\end{barticle}
\endbibitem

\bibitem[\protect\citeauthoryear{Nowakowski}{2016}]{nowakowski2016symmetric}
\begin{barticle}
\bauthor{\bsnm{Nowakowski}, \binits{M.L.}}:
\batitle{The symmetric extendibility of quantum states}.
\bjtitle{J. Phys. A: Math. Theor.}
\bvolume{49},
\bfpage{385301}
(\byear{2016})
\doiurl{10.1088/1751-8113/49/38/385301}
\end{barticle}
\endbibitem

\bibitem[\protect\citeauthoryear{Wang et~al.}{2024}]{wang2024quantifying}
\begin{barticle}
\bauthor{\bsnm{Wang}, \binits{K.}},
\bauthor{\bsnm{Wang}, \binits{X.}},
\bauthor{\bsnm{Wilde}, \binits{M.M.}}:
\batitle{Quantifying the unextendibility of entanglement}.
\bjtitle{New J. Phys.}
\bvolume{26},
\bfpage{033013}
(\byear{2024})
\doiurl{10.1088/1367-2630/ad264e}
\end{barticle}
\endbibitem

\bibitem[\protect\citeauthoryear{Brand\~{a}o
  et~al.}{2011}]{brandao2011faithful}
\begin{barticle}
\bauthor{\bsnm{Brand\~{a}o}, \binits{F.G.S.L.}},
\bauthor{\bsnm{Christandl}, \binits{M.}},
\bauthor{\bsnm{Yard}, \binits{J.}}:
\batitle{Faithful squashed entanglement}.
\bjtitle{Commun. Math. Phys.}
\bvolume{306},
\bfpage{805}--\blpage{830}
(\byear{2011})
\doiurl{10.1007/s00220-011-1302-1}
\end{barticle}
\endbibitem

\bibitem[\protect\citeauthoryear{Li and Winter}{2018}]{li2018squashed}
\begin{barticle}
\bauthor{\bsnm{Li}, \binits{K.}},
\bauthor{\bsnm{Winter}, \binits{A.}}:
\batitle{Squashed entanglement, k-extendibility, quantum {M}arkov chains, and
  recovery maps}.
\bjtitle{Found. Phys.}
\bvolume{48},
\bfpage{910}--\blpage{924}
(\byear{2018})
\doiurl{10.1007/s10701-018-0143-6}
\end{barticle}
\endbibitem

\bibitem[\protect\citeauthoryear{Moroder et~al.}{2006}]{moroder2006one}
\begin{barticle}
\bauthor{\bsnm{Moroder}, \binits{T.}},
\bauthor{\bsnm{Curty}, \binits{M.}},
\bauthor{\bsnm{L{\"u}tkenhaus}, \binits{N.}}:
\batitle{One-way quantum key distribution: Simple upper bound on the secret key
  rate}.
\bjtitle{Phys. Rev. A}
\bvolume{74},
\bfpage{052301}
(\byear{2006})
\doiurl{10.1103/PhysRevA.74.052301}
\end{barticle}
\endbibitem

\bibitem[\protect\citeauthoryear{Myhr et~al.}{2009}]{myhr2009symmetric}
\begin{barticle}
\bauthor{\bsnm{Myhr}, \binits{G.O.}},
\bauthor{\bsnm{Renes}, \binits{J.M.}},
\bauthor{\bsnm{Doherty}, \binits{A.C.}},
\bauthor{\bsnm{L{\"u}tkenhaus}, \binits{N.}}:
\batitle{Symmetric extension in two-way quantum key distribution}.
\bjtitle{Phys. Rev. A}
\bvolume{79},
\bfpage{042329}
(\byear{2009})
\doiurl{10.1103/PhysRevA.79.042329}
\end{barticle}
\endbibitem

\bibitem[\protect\citeauthoryear{Khatri}{2016}]{khatri2016symmetric}
\begin{botherref}
\oauthor{\bsnm{Khatri}, \binits{S.}}:
Symmetric extendability of quantum states and the extreme limits of quantum key
  distribution.
Master's thesis,
University of Waterloo
(2016)
\end{botherref}
\endbibitem

\bibitem[\protect\citeauthoryear{Kaur et~al.}{2019}]{kaur2019extendibility}
\begin{barticle}
\bauthor{\bsnm{Kaur}, \binits{E.}},
\bauthor{\bsnm{Das}, \binits{S.}},
\bauthor{\bsnm{Wilde}, \binits{M.M.}},
\bauthor{\bsnm{Winter}, \binits{A.}}:
\batitle{Extendibility limits the performance of quantum processors}.
\bjtitle{Phys. Rev. Lett.}
\bvolume{123},
\bfpage{070502}
(\byear{2019})
\doiurl{10.1103/PhysRevLett.123.070502}
\end{barticle}
\endbibitem

\bibitem[\protect\citeauthoryear{Ranade}{2009}]{Ranade2009symmetric}
\begin{barticle}
\bauthor{\bsnm{Ranade}, \binits{K.S.}}:
\batitle{Symmetric extendibility for a class of qudit states}.
\bjtitle{J. Phys. A: Math. Theor.}
\bvolume{42},
\bfpage{425302}
(\byear{2009})
\doiurl{10.1088/1751-8113/42/42/425302}
\end{barticle}
\endbibitem

\bibitem[\protect\citeauthoryear{Lami et~al.}{2019}]{lami2019extendibility}
\begin{barticle}
\bauthor{\bsnm{Lami}, \binits{L.}},
\bauthor{\bsnm{Khatri}, \binits{S.}},
\bauthor{\bsnm{Adesso}, \binits{G.}},
\bauthor{\bsnm{Wilde}, \binits{M.M.}}:
\batitle{Extendibility of bosonic {G}aussian states}.
\bjtitle{Phys. Rev. Lett.}
\bvolume{123},
\bfpage{050501}
(\byear{2019})
\doiurl{10.1103/PhysRevLett.123.050501}
\end{barticle}
\endbibitem

\bibitem[\protect\citeauthoryear{Krumnow
  et~al.}{2024}]{krumnow2024extendibility}
\begin{botherref}
\oauthor{\bsnm{Krumnow}, \binits{C.}},
\oauthor{\bsnm{Zimbor{\'a}s}, \binits{Z.}},
\oauthor{\bsnm{Eisert}, \binits{J.}}:
Extendibility of fermionic states and rigorous ground state approximations of
  interacting fermionic systems.
arXiv preprint arXiv:2410.08322
(2024)
\end{botherref}
\endbibitem

\bibitem[\protect\citeauthoryear{Negari and
  Salek}{2025}]{negari2025extendibility}
\begin{botherref}
\oauthor{\bsnm{Negari}, \binits{A.-R.}},
\oauthor{\bsnm{Salek}, \binits{F.}}:
Extendibility of fermionic gaussian states.
arXiv preprint arXiv:2508.18532
(2025)
\end{botherref}
\endbibitem

\bibitem[\protect\citeauthoryear{Werner}{1989}]{werner1989states}
\begin{barticle}
\bauthor{\bsnm{Werner}, \binits{R.F.}}:
\batitle{Quantum states with {E}instein-{P}odolsky-{R}osen correlations
  admitting a hidden-variable model}.
\bjtitle{Phys. Rev. A}
\bvolume{40},
\bfpage{4277}--\blpage{4281}
(\byear{1989})
\doiurl{10.1103/PhysRevA.40.4277}
\end{barticle}
\endbibitem

\bibitem[\protect\citeauthoryear{Keyl}{2002}]{keyl2002fundamentals}
\begin{barticle}
\bauthor{\bsnm{Keyl}, \binits{M.}}:
\batitle{Fundamentals of quantum information theory}.
\bjtitle{Physics Reports}
\bvolume{369},
\bfpage{431}--\blpage{548}
(\byear{2002})
\doiurl{10.1016/S0370-1573(02)00266-1}
\end{barticle}
\endbibitem

\bibitem[\protect\citeauthoryear{Horodecki and
  Horodecki}{1999}]{horodecki1999reduction}
\begin{barticle}
\bauthor{\bsnm{Horodecki}, \binits{M.}},
\bauthor{\bsnm{Horodecki}, \binits{P.}}:
\batitle{Reduction criterion of separability and limits for a class of
  distillation protocols}.
\bjtitle{Phys. Rev. A}
\bvolume{59},
\bfpage{4206}--\blpage{4216}
(\byear{1999})
\doiurl{10.1103/physreva.59.4206}
\end{barticle}
\endbibitem

\bibitem[\protect\citeauthoryear{Jakab et~al.}{2022}]{jakab2022extendibility}
\begin{botherref}
\oauthor{\bsnm{Jakab}, \binits{D.}},
\oauthor{\bsnm{Solymos}, \binits{A.}},
\oauthor{\bsnm{Zimborás}, \binits{Z.}}:
Extendibility of {W}erner States
(2022).
\url{https://arxiv.org/abs/2208.13743}
\end{botherref}
\endbibitem

\bibitem[\protect\citeauthoryear{Vollbrecht and
  Werner}{2001}]{vollbrecth2001entanglement}
\begin{barticle}
\bauthor{\bsnm{Vollbrecht}, \binits{K.G.H.}},
\bauthor{\bsnm{Werner}, \binits{R.F.}}:
\batitle{Entanglement measures under symmetry}.
\bjtitle{Phys. Rev. A}
\bvolume{64},
\bfpage{062307}
(\byear{2001})
\doiurl{10.1103/PhysRevA.64.062307}
\end{barticle}
\endbibitem

\bibitem[\protect\citeauthoryear{Anshu et~al.}{2020}]{anshu2020beyond}
\begin{bchapter}
\bauthor{\bsnm{Anshu}, \binits{A.}},
\bauthor{\bsnm{Gosset}, \binits{D.}},
\bauthor{\bsnm{Morenz}, \binits{K.}}:
\bctitle{{Beyond Product State Approximations for a Quantum Analogue of Max
  Cut}}.
In: \beditor{\bsnm{Flammia}, \binits{S.T.}} (ed.)
\bbtitle{15th Conference on the Theory of Quantum Computation, Communication
  and Cryptography (TQC 2020)}.
\bsertitle{Leibniz International Proceedings in Informatics (LIPIcs)},
vol. \bseriesno{158},
pp. \bfpage{7}--\blpage{1715}.
\bpublisher{Schloss Dagstuhl -- Leibniz-Zentrum f{\"u}r Informatik},
\blocation{Dagstuhl, Germany}
(\byear{2020}).
\doiurl{10.4230/LIPIcs.TQC.2020.7}
\end{bchapter}
\endbibitem

\bibitem[\protect\citeauthoryear{Parekh and
  Thompson}{2021}]{parekh2021application}
\begin{bchapter}
\bauthor{\bsnm{Parekh}, \binits{O.}},
\bauthor{\bsnm{Thompson}, \binits{K.}}:
\bctitle{{Application of the Level-2 Quantum Lasserre Hierarchy in Quantum
  Approximation Algorithms}}.
In: \beditor{\bsnm{Bansal}, \binits{N.}},
\beditor{\bsnm{Merelli}, \binits{E.}},
\beditor{\bsnm{Worrell}, \binits{J.}} (eds.)
\bbtitle{48th International Colloquium on Automata, Languages, and Programming
  (ICALP 2021)}.
\bsertitle{Leibniz International Proceedings in Informatics (LIPIcs)},
vol. \bseriesno{198},
pp. \bfpage{102}--\blpage{110220}.
\bpublisher{Schloss Dagstuhl -- Leibniz-Zentrum f{\"u}r Informatik},
\blocation{Dagstuhl, Germany}
(\byear{2021}).
\doiurl{10.4230/LIPIcs.ICALP.2021.102}
\end{bchapter}
\endbibitem

\bibitem[\protect\citeauthoryear{King}{2023}]{king2023improved}
\begin{barticle}
\bauthor{\bsnm{King}, \binits{R.}}:
\batitle{An {I}mproved {A}pproximation {A}lgorithm for {Q}uantum {M}ax-{C}ut on
  {T}riangle-{F}ree {G}raphs}.
\bjtitle{{Quantum}}
\bvolume{7},
\bfpage{1180}
(\byear{2023})
\doiurl{10.22331/q-2023-11-09-1180}
\end{barticle}
\endbibitem

\bibitem[\protect\citeauthoryear{Werner}{1998}]{werner1998optimal}
\begin{barticle}
\bauthor{\bsnm{Werner}, \binits{R.F.}}:
\batitle{Optimal cloning of pure states}.
\bjtitle{Phys. Rev. A}
\bvolume{58},
\bfpage{1827}--\blpage{1832}
(\byear{1998})
\doiurl{10.1103/PhysRevA.58.1827}
\end{barticle}
\endbibitem

\bibitem[\protect\citeauthoryear{Keyl and Werner}{1999}]{keyl1999optimal}
\begin{barticle}
\bauthor{\bsnm{Keyl}, \binits{M.}},
\bauthor{\bsnm{Werner}, \binits{R.F.}}:
\batitle{{Optimal cloning of pure states, testing single clones}}.
\bjtitle{J. Math. Phys.}
\bvolume{40},
\bfpage{3283}--\blpage{3299}
(\byear{1999})
\doiurl{10.1063/1.532887}
\end{barticle}
\endbibitem

\bibitem[\protect\citeauthoryear{Nechita et~al.}{2023}]{nechita2023asymmetic}
\begin{botherref}
\oauthor{\bsnm{Nechita}, \binits{I.}},
\oauthor{\bsnm{Pellegrini}, \binits{C.}},
\oauthor{\bsnm{Rochette}, \binits{D.}}:
{The asymmetric quantum cloning region}.
Lett. Math. Phys.
\textbf{113}
(2023)
\doiurl{10.1007/s11005-023-01694-8}
\end{botherref}
\endbibitem

\bibitem[\protect\citeauthoryear{Buhrman et~al.}{2014}]{buhrman2014position}
\begin{barticle}
\bauthor{\bsnm{Buhrman}, \binits{H.}},
\bauthor{\bsnm{Chandran}, \binits{N.}},
\bauthor{\bsnm{Fehr}, \binits{S.}},
\bauthor{\bsnm{Gelles}, \binits{R.}},
\bauthor{\bsnm{Goyal}, \binits{V.}},
\bauthor{\bsnm{Ostrovsky}, \binits{R.}},
\bauthor{\bsnm{Schaffner}, \binits{C.}}:
\batitle{Position-based quantum cryptography: Impossibility and constructions}.
\bjtitle{SIAM Journal on Computing}
\bvolume{43},
\bfpage{150}--\blpage{178}
(\byear{2014})
\doiurl{10.1137/130913687}
\end{barticle}
\endbibitem

\bibitem[\protect\citeauthoryear{Kent et~al.}{2011}]{kent2011quantum}
\begin{barticle}
\bauthor{\bsnm{Kent}, \binits{A.}},
\bauthor{\bsnm{Munro}, \binits{W.J.}},
\bauthor{\bsnm{Spiller}, \binits{T.P.}}:
\batitle{Quantum tagging: Authenticating location via quantum information and
  relativistic signaling constraints}.
\bjtitle{Phys. Rev. A}
\bvolume{84},
\bfpage{012326}
(\byear{2011})
\doiurl{10.1103/PhysRevA.84.012326}
\end{barticle}
\endbibitem

\bibitem[\protect\citeauthoryear{Unruh}{2014}]{unruh2014quantum}
\begin{bchapter}
\bauthor{\bsnm{Unruh}, \binits{D.}}:
\bctitle{Quantum position verification in the random oracle model}.
In: \beditor{\bsnm{Garay}, \binits{J.A.}},
\beditor{\bsnm{Gennaro}, \binits{R.}} (eds.)
\bbtitle{Advances in Cryptology -- CRYPTO 2014},
pp. \bfpage{1}--\blpage{18}.
\bpublisher{Springer},
\blocation{Berlin}
(\byear{2014})
\end{bchapter}
\endbibitem

\bibitem[\protect\citeauthoryear{Jakab and Zimbor\'as}{2021}]{jakab2021phase}
\begin{barticle}
\bauthor{\bsnm{Jakab}, \binits{D.}},
\bauthor{\bsnm{Zimbor\'as}, \binits{Z.}}:
\batitle{Quantum phases of collective {SU(3)} spin systems with bipartite
  symmetry}.
\bjtitle{Phys. Rev. B}
\bvolume{103},
\bfpage{214448}
(\byear{2021})
\doiurl{10.1103/PhysRevB.103.214448}
\end{barticle}
\endbibitem

\bibitem[\protect\citeauthoryear{Ruskai}{1969}]{ruskai1969representability}
\begin{barticle}
\bauthor{\bsnm{Ruskai}, \binits{M.B.}}:
\batitle{$n$-representability problem: Conditions on geminals}.
\bjtitle{Phys. Rev.}
\bvolume{183},
\bfpage{129}--\blpage{141}
(\byear{1969})
\doiurl{10.1103/PhysRev.183.129}
\end{barticle}
\endbibitem

\bibitem[\protect\citeauthoryear{Haapasalo et~al.}{2021}]{haapasalo2021quantum}
\begin{barticle}
\bauthor{\bsnm{Haapasalo}, \binits{E.}},
\bauthor{\bsnm{Kraft}, \binits{T.}},
\bauthor{\bsnm{Miklin}, \binits{N.}},
\bauthor{\bsnm{Uola}, \binits{R.}}:
\batitle{Quantum marginal problem and incompatibility}.
\bjtitle{{Quantum}}
\bvolume{5},
\bfpage{476}
(\byear{2021})
\doiurl{10.22331/q-2021-06-15-476}
\end{barticle}
\endbibitem

\bibitem[\protect\citeauthoryear{Navascu\'es
  et~al.}{2009}]{navascues2009complete}
\begin{barticle}
\bauthor{\bsnm{Navascu\'es}, \binits{M.}},
\bauthor{\bsnm{Owari}, \binits{M.}},
\bauthor{\bsnm{Plenio}, \binits{M.B.}}:
\batitle{Complete criterion for separability detection}.
\bjtitle{Phys. Rev. Lett.}
\bvolume{103},
\bfpage{160404}
(\byear{2009})
\doiurl{10.1103/PhysRevLett.103.160404}
\end{barticle}
\endbibitem

\bibitem[\protect\citeauthoryear{Brand\~{a}o and
  Harrow}{2017}]{brandao2017quantum}
\begin{barticle}
\bauthor{\bsnm{Brand\~{a}o}, \binits{F.G.S.L.}},
\bauthor{\bsnm{Harrow}, \binits{A.W.}}:
\batitle{Quantum de {F}inetti theorems under local measurements with
  applications}.
\bjtitle{Communications in Mathematical Physics}
\bvolume{353}(\bissue{2}),
\bfpage{469}--\blpage{506}
(\byear{2017})
\doiurl{10.1007/s00220-017-2880-3}
\end{barticle}
\endbibitem

\bibitem[\protect\citeauthoryear{Boyd and Vandenberghe}{2004}]{boyd2004convex}
\begin{bbook}
\bauthor{\bsnm{Boyd}, \binits{S.P.}},
\bauthor{\bsnm{Vandenberghe}, \binits{L.}}:
\bbtitle{Convex Optimization}.
\bpublisher{Cambridge University Press},
\blocation{Cambridge}
(\byear{2004})
\end{bbook}
\endbibitem

\bibitem[\protect\citeauthoryear{Benson et~al.}{2000}]{benson2000solving}
\begin{barticle}
\bauthor{\bsnm{Benson}, \binits{S.J.}},
\bauthor{\bsnm{Ye}, \binits{Y.}},
\bauthor{\bsnm{Zhang}, \binits{X.}}:
\batitle{Solving large-scale sparse semidefinite programs for combinatorial
  optimization}.
\bjtitle{SIAM Journal on Optimization}
\bvolume{10}(\bissue{2}),
\bfpage{443}--\blpage{461}
(\byear{2000})
\doiurl{10.1137/S1052623497328008}
\end{barticle}
\endbibitem

\bibitem[\protect\citeauthoryear{König and Renner}{2005}]{konig2005deFinetti}
\begin{barticle}
\bauthor{\bsnm{König}, \binits{R.}},
\bauthor{\bsnm{Renner}, \binits{R.}}:
\batitle{{A de {F}inetti representation for finite symmetric quantum states}}.
\bjtitle{Journal of Mathematical Physics}
\bvolume{46}(\bissue{12}),
\bfpage{122108}
(\byear{2005})
\doiurl{10.1063/1.2146188}
{\href{https://arxiv.org/abs/https://pubs.aip.org/aip/jmp/article-pdf/doi/10.1063/1.2146188/16738139/122108\_1\_online.pdf}{{https://pubs.aip.org/aip/jmp/article-pdf/doi/10.1063/1.2146188/16738139/122108\_1\_online.pdf}}}
\end{barticle}
\endbibitem

\bibitem[\protect\citeauthoryear{Christandl
  et~al.}{2007}]{christandl2007oneandahalf}
\begin{barticle}
\bauthor{\bsnm{Christandl}, \binits{M.}},
\bauthor{\bsnm{König}, \binits{R.}},
\bauthor{\bsnm{Mitchison}, \binits{G.}},
\bauthor{\bsnm{Renner}, \binits{R.}}:
\batitle{One-and-a-half quantum de {F}inetti theorems}.
\bjtitle{Communications in Mathematical Physics}
\bvolume{273}(\bissue{2}),
\bfpage{473}--\blpage{498}
(\byear{2007})
\doiurl{10.1007/s00220-007-0189-3}
\end{barticle}
\endbibitem

\bibitem[\protect\citeauthoryear{Brauer}{1937}]{brauer1937algebras}
\begin{barticle}
\bauthor{\bsnm{Brauer}, \binits{R.}}:
\batitle{On algebras which are connected with the semisimple continuous
  groups}.
\bjtitle{Annals of Mathematics}
\bvolume{38},
\bfpage{857}--\blpage{872}
(\byear{1937})
\doiurl{10.2307/1968843}
\end{barticle}
\endbibitem

\bibitem[\protect\citeauthoryear{Peres}{1996}]{peres1996separability}
\begin{barticle}
\bauthor{\bsnm{Peres}, \binits{A.}}:
\batitle{Separability criterion for density matrices}.
\bjtitle{Phys. Rev. Lett.}
\bvolume{77},
\bfpage{1413}--\blpage{1415}
(\byear{1996})
\doiurl{10.1103/PhysRevLett.77.1413}
\end{barticle}
\endbibitem

\bibitem[\protect\citeauthoryear{Horodecki
  et~al.}{1996}]{horodecki1996separability}
\begin{barticle}
\bauthor{\bsnm{Horodecki}, \binits{M.}},
\bauthor{\bsnm{Horodecki}, \binits{P.}},
\bauthor{\bsnm{Horodecki}, \binits{R.}}:
\batitle{Separability of mixed states: necessary and sufficient conditions}.
\bjtitle{Physics Letters A}
\bvolume{223}(\bissue{1}),
\bfpage{1}--\blpage{8}
(\byear{1996})
\doiurl{10.1016/S0375-9601(96)00706-2}
\end{barticle}
\endbibitem

\bibitem[\protect\citeauthoryear{Fulton and
  Harris}{2004}]{fulton2004representation}
\begin{bbook}
\bauthor{\bsnm{Fulton}, \binits{W.}},
\bauthor{\bsnm{Harris}, \binits{J.}}:
\bbtitle{Representation Theory: A First Course}.
\bpublisher{Springer},
\blocation{New York}
(\byear{2004})
\end{bbook}
\endbibitem

\bibitem[\protect\citeauthoryear{Kirillov}{1976}]{kirillov1976elements}
\begin{bbook}
\bauthor{\bsnm{Kirillov}, \binits{A.A.}}:
\bbtitle{Elements of the Theory of Representation}.
\bpublisher{Springer},
\blocation{New York}
(\byear{1976})
\end{bbook}
\endbibitem

\bibitem[\protect\citeauthoryear{Goodman and
  Wallach}{2009}]{goodman2009symmetry}
\begin{bbook}
\bauthor{\bsnm{Goodman}, \binits{R.}},
\bauthor{\bsnm{Wallach}, \binits{N.R.}}:
\bbtitle{Symmetry, Representations, and Invariants}.
\bpublisher{Springer},
\blocation{New York}
(\byear{2009})
\end{bbook}
\endbibitem

\bibitem[\protect\citeauthoryear{Goodman and
  Wallach}{1998}]{goodman1998representation}
\begin{bbook}
\bauthor{\bsnm{Goodman}, \binits{R.}},
\bauthor{\bsnm{Wallach}, \binits{N.R.}}:
\bbtitle{Representations and Invariants of the Classical Groups}.
\bpublisher{Cambridge University Press},
\blocation{Cambridge}
(\byear{1998})
\end{bbook}
\endbibitem

\bibitem[\protect\citeauthoryear{Baez}{2023}]{baez2023young}
\begin{botherref}
\oauthor{\bsnm{Baez}, \binits{J.C.}}:
Young Diagrams and Classical Groups
(2023).
\url{https://arxiv.org/abs/2302.07971}
\end{botherref}
\endbibitem

\bibitem[\protect\citeauthoryear{King}{1975}]{king1975branching}
\begin{barticle}
\bauthor{\bsnm{King}, \binits{R.C.}}:
\batitle{Branching rules for classical {L}ie groups using tensor and spinor
  methods}.
\bjtitle{Journal of Physics A: Mathematical and General}
\bvolume{8}(\bissue{4}),
\bfpage{429}--\blpage{449}
(\byear{1975})
\doiurl{10.1088/0305-4470/8/4/004}
\end{barticle}
\endbibitem

\bibitem[\protect\citeauthoryear{Littlewood and
  Richardson}{1934}]{littlewood1934group}
\begin{barticle}
\bauthor{\bsnm{Littlewood}, \binits{D.E.}},
\bauthor{\bsnm{Richardson}, \binits{A.R.}}:
\batitle{Group characters and algebra}.
\bjtitle{Philosophical Transactions of the Royal Society of London. Series A,
  Containing Papers of a Mathematical or Physical Character}
\bvolume{233}(\bissue{721-730}),
\bfpage{99}--\blpage{141}
(\byear{1934})
\doiurl{10.1098/rsta.1934.0015}
\end{barticle}
\endbibitem

\bibitem[\protect\citeauthoryear{Koike and Terada}{1987}]{koike1987young}
\begin{barticle}
\bauthor{\bsnm{Koike}, \binits{K.}},
\bauthor{\bsnm{Terada}, \binits{I.}}:
\batitle{Young-diagrammatic methods for the representation theory of the
  classical groups of type {$B_n$}, {$C_n$}, {$D_n$}}.
\bjtitle{Journal of Algebra}
\bvolume{107}(\bissue{2}),
\bfpage{466}--\blpage{511}
(\byear{1987})
\doiurl{10.1016/0021-8693(87)90099-8}
\end{barticle}
\endbibitem

\bibitem[\protect\citeauthoryear{Gao et~al.}{2021}]{gao2021newell}
\begin{barticle}
\bauthor{\bsnm{Gao}, \binits{S.}},
\bauthor{\bsnm{Orelowitz}, \binits{G.}},
\bauthor{\bsnm{Yong}, \binits{A.}}:
\batitle{{N}ewell-{L}ittlewood numbers}.
\bjtitle{Transactions of the American Mathematical Society}
\bvolume{374},
\bfpage{6331}--\blpage{6366}
(\byear{2021})
\end{barticle}
\endbibitem

\bibitem[\protect\citeauthoryear{King}{2003}]{king2003modification}
\begin{barticle}
\bauthor{\bsnm{King}, \binits{R.C.}}:
\batitle{{Modification Rules and Products of Irreducible Representations of the
  Unitary, Orthogonal, and Symplectic Groups}}.
\bjtitle{Journal of Mathematical Physics}
\bvolume{12}(\bissue{8}),
\bfpage{1588}--\blpage{1598}
(\byear{2003})
\doiurl{10.1063/1.1665778}
{\href{https://arxiv.org/abs/https://pubs.aip.org/aip/jmp/article-pdf/12/8/1588/11285655/1588\_1\_online.pdf}{{https://pubs.aip.org/aip/jmp/article-pdf/12/8/1588/11285655/1588\_1\_online.pdf}}}
\end{barticle}
\endbibitem

\bibitem[\protect\citeauthoryear{Littlewood}{1958}]{littlewood1958products}
\begin{barticle}
\bauthor{\bsnm{Littlewood}, \binits{D.E.}}:
\batitle{Products and plethysms of characters with orthogonal, symplectic and
  symmetric groups}.
\bjtitle{Canadian Journal of Mathematics}
\bvolume{10},
\bfpage{17}--\blpage{32}
(\byear{1958})
\doiurl{10.4153/CJM-1958-002-7}
\end{barticle}
\endbibitem

\bibitem[\protect\citeauthoryear{Howe et~al.}{2005}]{howe2005stable}
\begin{barticle}
\bauthor{\bsnm{Howe}, \binits{R.}},
\bauthor{\bsnm{Tan}, \binits{E.-C.}},
\bauthor{\bsnm{Willenbring}, \binits{J.F.}}:
\batitle{Stable branching rules for classical symmetric pairs}.
\bjtitle{Transactions of the American Mathematical Society}
\bvolume{357}(\bissue{4}),
\bfpage{1601}--\blpage{1626}
(\byear{2005}).
Accessed 2023-06-20
\end{barticle}
\endbibitem

\bibitem[\protect\citeauthoryear{Jakab}{2022}]{jakab2022interplay}
\begin{botherref}
\oauthor{\bsnm{Jakab}, \binits{D.}}:
The interplay of unitary and permutation symmetries in composite quantum
  systems.
PhD thesis,
University of Pécs
(2022)
\end{botherref}
\endbibitem

\bibitem[\protect\citeauthoryear{Hall}{2015}]{hall2015lie}
\begin{bbook}
\bauthor{\bsnm{Hall}, \binits{B.}}:
\bbtitle{{L}ie Groups, {L}ie Algebras, and Representations, Second Edition}.
\bpublisher{Springer},
\blocation{New York}
(\byear{2015})
\end{bbook}
\endbibitem

\bibitem[\protect\citeauthoryear{Iachello}{2015}]{iachello2015lie}
\begin{bbook}
\bauthor{\bsnm{Iachello}, \binits{F.}}:
\bbtitle{{L}ie Algebras and {A}pplications, Second Edition}.
\bpublisher{Springer},
\blocation{New York}
(\byear{2015})
\end{bbook}
\endbibitem

\bibitem[\protect\citeauthoryear{Newell}{1951}]{newell1951modification}
\begin{barticle}
\bauthor{\bsnm{Newell}, \binits{M.J.}}:
\batitle{Modification rules for the orthogonal and symplectic groups}.
\bjtitle{Proceedings of the Royal Irish Academy. Section A: Mathematical and
  Physical Sciences}
\bvolume{54},
\bfpage{153}--\blpage{163}
(\byear{1951}).
Accessed 2024-05-03
\end{barticle}
\endbibitem

\bibitem[\protect\citeauthoryear{King}{1971}]{king1971dimension}
\begin{barticle}
\bauthor{\bsnm{King}, \binits{R.C.}}:
\batitle{The dimensions of irreducible tensor representations of the orthogonal
  and symplectic groups}.
\bjtitle{Canadian Journal of Mathematics}
\bvolume{23}(\bissue{1}),
\bfpage{176}--\blpage{188}
(\byear{1971})
\doiurl{10.4153/CJM-1971-017-2}
\end{barticle}
\endbibitem

\bibitem[\protect\citeauthoryear{Brylawski}{1973}]{brylawski1973lattice}
\begin{barticle}
\bauthor{\bsnm{Brylawski}, \binits{T.}}:
\batitle{The lattice of integer partitions}.
\bjtitle{Discrete Mathematics}
\bvolume{6}(\bissue{3}),
\bfpage{201}--\blpage{219}
(\byear{1973})
\doiurl{10.1016/0012-365X(73)90094-0}
\end{barticle}
\endbibitem

\bibitem[\protect\citeauthoryear{Macdonald}{1998}]{macdonald1998symmetric}
\begin{bbook}
\bauthor{\bsnm{Macdonald}, \binits{I.G.}}:
\bbtitle{Symmetric Functions and Hall Polynomials}.
\bpublisher{Oxford University Press},
\blocation{Oxford}
(\byear{1998})
\end{bbook}
\endbibitem

\bibitem[\protect\citeauthoryear{Sundaram}{1990}]{sundaram1990orthogonal}
\begin{barticle}
\bauthor{\bsnm{Sundaram}, \binits{S.}}:
\batitle{Orthogonal tableaux and an insertion algorithm for so(2n + 1)}.
\bjtitle{Journal of Combinatorial Theory, Series A}
\bvolume{53}(\bissue{2}),
\bfpage{239}--\blpage{256}
(\byear{1990})
\doiurl{10.1016/0097-3165(90)90059-6}
\end{barticle}
\endbibitem

\bibitem[\protect\citeauthoryear{Okada}{2016}]{okada2016pieri}
\begin{botherref}
\oauthor{\bsnm{Okada}, \binits{S.}}:
Pieri rules for classical groups and equinumeration between generalized
  oscillating tableaux and semistandard tableaux
(2016).
\url{https://arxiv.org/abs/1606.02375}
\end{botherref}
\endbibitem

\end{thebibliography}

\begin{appendices}

\section{The representation theory of $\Ort(d)$} \label{sec:app_rep}

The representation theory of $\Ort(d)$ is a vast topic, therefore, the interested reader is directed to \cite{kirillov1976elements,fulton2004representation,goodman2009symmetry,goodman1998representation} for a more in depth analysis. Here, only the essential parts, more precisely the tensor representations, the fusion rules, and the branching rules from $\U(d)$ to $\Ort(d)$ are presented. A recapitulation of the representation theory of $\U(d)$ is also presented.

Note that throughout this article, all representations are assumed to be unitary representations and that their domains and codomains will not always be explicitly stated, but will always be assumed to be appropriate ones. Representations related to $\U(d)$ will be denoted by $\Rep$, representations related to $\Ort(d)$ by $\Repp$, and representations related to $\sym_n$ by $\Repppp$.

\subsection{Recap: Irreducible tensor representations of $\U(d)$}

The unitary group of degree $d$, denoted by $\U(d)$, is the set of linear operators on a $d$-dimensional Hilbert space $\hil$ that preserve the inner product with the group operation being the product of linear operators. The defining (or tautologous) representation of $\U(d)$ is the representation over $\hil$ in which each unitary in $\U(d)$ is mapped to itself.

The relevant representations to our topic arise by taking the tensor product of the defining representation some $N$ number of times with itself and looking at the irrep decomposition of this representation. These are called irreducible tensor representations. In general, tensor representations can be built by taking the tensor products of these irreps.

A polynomial representation is one in which (given bases) the matrix entries of the representing matrix are polynomials of the matrix entries of the original matrix $U\in\U(d)$. As it turns out, in the case of $\U(d)$ (and more generally $\GL(d)$), the tensor representations are exactly the polynomial representations. The following theorem classifies all of these (see \cite{baez2023young} or Section 2 in \cite{king1975branching}):

\begin{theorem}
    The irreducible tensor representations $\Rep_\lambda$ of $\U(d)$ are in a one-to-one correspondence with the Young diagrams $\lambda$ of at most $d$ rows. Furthermore, the irrep decomposition of the $N$-fold tensor product of the defining representation contains exactly all Young diagrams with at most $d$ rows and $N$ number of boxes, with possibly greater than one multiplicity.
\end{theorem}

Note that a Young diagram is a finite collection of boxes arranged in left-justified rows in order of non-increasing row lengths.

The above theorem means that understanding how $\U(d)$ irreps labelled by Young diagrams fuse together and restrict to $\Ort(d)$ irreps is what is needed to recreate the calculations in this article. Following the above theorem, it is also easy to deduce that the defining representation of $\U(d)$ is $\Rep_{\ydiagram{1}}$. The trivial representation is labelled with a zero to represent a Young diagram with zero boxes: $\Rep_{0}$.

Note that the notation of this article does not specify the degree $d$ of the unitary group for which the representation $\Rep_\lambda$ stands for. This is because $d$ is treated as a parameter and the relevant representation theory (the fusion and restriction rules) of the $\U(d)$ irreps is the same for large enough $d$. Therefore, it is omitted from the notation. However, it is always implicitly fixed for a given calculation, and two irreps of different unitary groups will never figure in the same context.

\subsection{Irreducible tensor representations of $\Ort(d)$}

The real orthogonal group of degree $d$, denoted by $\Ort(d)$, is the set of linear operators on a $d$-dimensional real inner product space $V$ that preserve the inner product, with the group operation being the product of linear operators. The defining (or tautologous) representation of $\Ort(d)$ is the representation over $V$ in which each element of $\Ort(d)$ is mapped to itself.

To be able to define a representation of the orthogonal group on a complex inner product space $\hil$, one first has to fix a real structure on the vector space. There is no canonical way of doing this, as infinitely many unitarily conjugate real structures can be defined over the same complex vector space.\footnote{In our case this is related to the fact that there are infinitely many unitarily equivalent embeddings of $\Ort(d)$ into $\U(d)$.} One of the ways to fix a real structure is to fix an orthonormal basis $\{e_i\}_{i=1}^{d}$ for the $d$-dimensional Hilbert space $\hil$ and define the real structure on $\hil$ as the complex conjugation of the coefficients of vectors in this fixed basis.\footnote{Note that there are infinitely many ways to fix a basis.} This way, we have the real subspace $\hil_{\real}\coloneqq \mathrm{span}_{\real}(\{e_i\}_{i=1}^{d})$, and the following splitting is true $\hil=\hil_\real\OP i\hil_\real$, where $i$ denotes the imaginary unit. Once $\hil_\real$ is fixed, a real inner product naturally arises on it from the restriction of the usual inner product of $\hil$ to $\hil_\real$. Thus, $\Ort(d)$ is defined as those unitaries that preserve $\hil_\real$. Equivalently, we have that the set $\Ort(d)\subset \U(d)$ will be the set of unitaries with only real numbers as their matrix elements in the basis $\{e_i\}_{i=1}^{d}$.

Note that once an orthonormal basis is fixed, it can be rotated by any element $O\in\Ort(d)$ to get another orthonormal basis, which will lead to an equivalent real structure. One can construct the $\hil_\real$-basis-independent vector $\Psi\in\hil\ot\hil$ by defining $\Psi\coloneqq \sum_{i=1}^{d} e_i\ot e_i$. This is the unique vector in $\hil_\real \ot \hil_\real$ such that for all $v,w\in\hil_\real$ we have that $\inn{v}{w}_{\hil}=\inn{\Psi}{v\ot w}_{\hil\ot\hil}$. Note that $(O\otimes O)\Psi=\Psi$ for all $O\in\Ort(d)$. Furthermore, the projection onto $\Psi$ is exactly $\fflip$, or, equivalently, the unnormed maximally entangled state in the description of isotropic states. For further information, see Section 3.5 of \cite{fulton2004representation}.

Just as before, one is only interested in the tensor representations of $\Ort(d)$ (see Section 2 in \cite{king1975branching}):

\begin{theorem}
    Let $\Ort(d)$ be the real orthogonal group of degree $d$ with $d=2k$ or $d=2k+1$. For each Young diagram $\lambda$ of at most $k$ rows there exist at most two irreducible tensor representations of $\Ort(d)$ labelled by $\lambda$ and $\lambda^*$. The two representations are called associate representations, and the relationship between them is the following: $\Repp_\lambda\cong \det \ot \Repp_{\lambda^{*}}$. All irreducible tensor representations of $\Ort(d)$ are of this form. However, in the $d=2k$ case, there is only one representation linked to Young diagrams $\lambda$ with non-zero last row, that is, in this case: $\Repp_{\lambda}\cong \Repp_{\lambda^{*}}$.
\end{theorem}

It is non-trivial, but, similarly to the $\U(d)$ case, the defining representation is $\Repp_{\ydiagram{1}}$, and the trivial representation is $\Repp_{0}$. Note that $\Repp_{0^{*}}$ is the determinant representation. Importantly $\det(O)\in\{-1,1\}$ for any $O\in\Ort(d)$.

Note again that the notation of this article does not specify the dimension $d$ of the orthogonal group for which the representation $\Repp_\lambda$ or $\Repp_{\lambda^{*}}$ stands for. This is because $d$ is treated as a parameter and the relevant representation theory (the fusion rules) of $\Ort(d)$ irreps is the same for large enough $d$. Therefore, it is omitted from the notation, but is always implicitly fixed for a given calculation. Most importantly, only those unitary and orthogonal irreps are put into relations that represent the unitary or orthogonal group of the same degree $d$.

\subsection{Recap: Fusion rules of $\U(d)$ - The Littlewood-Richardson formula}

As stated before, the irrep decomposition of the tensor product of two tensor irreps is needed. For more information, see originally \cite{littlewood1934group} and for more modern discussions see page 402 in \cite{goodman2009symmetry} or Section 3 in \cite{king1975branching} or page 485 in \cite{koike1987young}.

\begin{theorem}
    Let $\Rep_{\mu}$ and $\Rep_{\nu}$ be tensor irreps of $\U(d)$ labelled by Young diagrams $\mu$ and $\nu$ with at most $d$ rows (their depth is at most $d$). Their tensor product decomposes as
    \begin{equation}\label{eq:lr-formula}
        \Rep_{\mu}\ot\Rep_{\nu}\cong \bigoplus_{\substack {\lambda \\ \mathrm{depth}(\lambda)\leq d}} c^{\lambda}_{\mu,\nu} \Rep_{\lambda}\text{,}
    \end{equation}
    where $\lambda$ is a Young diagram and $c^{\lambda}_{\mu,\nu}$ are nonnegative integers called Littlewood-Richardson coefficients that give the multiplicity of the $\Rep_{\lambda}$ irrep.
\end{theorem}

Importantly, the only non-zero Littlewood-Richarson coefficients are those for which the sum of the boxes in $\mu$ and $\nu$ equals the number of boxes in $\lambda$. The Littlewood-Richardson coefficients can be calculated following a combinatorial algorithm called the Littlewood-Richardson rule (see Section 9.3.5 in \cite{goodman2009symmetry}):

\begin{theorem}[Littlewood-Richardson]
    The Littlewood-Richardson coefficient $c^{\lambda}_{\mu,\nu}$ is the number of Littlewood-Richardson skew tableaux of shape $\lambda\setminus \mu$ and weight $\nu$.
\end{theorem}

Note also that $c^{\lambda}_{\mu,\nu}$ is not dependent on $d$ insofar as the dependence is encoded in the requirement of Equation \eqref{eq:lr-formula} that $\mathrm{depth}(\lambda)\leq d$. This is why it is unnecessary to specify the degree $d$ of the unitary group to which the representation $\Rep_\lambda$ belongs to.

Most relevant calculations can be done using Pieri's formula (see Corollary 9.2.4 in \cite{goodman2009symmetry}):

\begin{corollary}[Pieri's formula]\label{cor:pieri}
    Let $\Rep_{\mu}$ and $\Rep_{\nu}$ be tensor irreps of $\U(d)$ labelled by Young diagrams with $\mu$ of depth smaller than $d$ and $\nu$ of depth $1$. Their tensor product decomposes as
    \begin{equation}
        \Rep_{\mu}\ot\Rep_{\nu}\cong \bigoplus_{\lambda} \Rep_{\lambda}\text{,}
    \end{equation}
    where the summation goes through all Young diagrams $\lambda$ one can get from $\mu$ by adding $|\nu|$ boxes to it, no two in the same column.
\end{corollary}

\ytableausetup{mathmode, boxsize=0.5em}

In the case of $\nu=\ydiagram{1}$, this means that all valid Young diagrams appear that come from $\mu$ having a single extra box in one of its columns. Note that the `empty' column(s) after the rightmost non-empty column also needs to be taken into account in the above corollary.

\subsection{Fusion rules of $\Ort(d)$ - The Newell-Littlewood formula}\label{sec:nl}

Similarly to the case of $\U(d)$ one has the following theorem for the $\Ort(d)$ irreps labelled by ordinary (`non-starred') Young diagrams (see Equation (3) in \cite{gao2021newell} Equation (4.9) in \cite{king2003modification} and page 501 in \cite{koike1987young}): 

\begin{theorem}\label{th:nl-formula}
    Let $\Repp_{\mu}$ and $\Repp_{\nu}$ be tensor irreps of $\Ort(d)$ ($d=2k$ or $d=2k+1$) labelled by Young diagrams $\mu$ and $\nu$ such that the sum of their depths is smaller or equal to $k$. Their tensor product decomposes as
    \begin{equation}\label{eq:nl-formula}
        \Repp_{\mu}\ot\Repp_{\nu}\cong \bigoplus_{\lambda} N^{\lambda}_{\mu,\nu} \Repp_{\lambda}\text{,}
    \end{equation}
    where $\lambda$ is a Young diagram and $N^{\lambda}_{\mu,\nu}$ are nonnegative integers called Newell-Littlewood numbers that give the multiplicity of the $\Repp_{\lambda}$ irrep.
\end{theorem}

Due to the relation $\Repp_{\lambda^*} = \det \ot \Repp_\lambda$, the above formula can easily be generalised to include the associate diagrams: If only one of the diagrams is starred on the left-hand side of Equation \eqref{eq:nl-formula}, then all the diagrams on the right-hand side are starred. If both are starred, then non on the right-hand side are.

The Newell-Littlewood numbers can be calculated from the Littlewood-Richardson coefficients using the following formula:
\begin{equation}
    N^{\lambda}_{\mu,\nu}=\sum_{\alpha,\beta,\gamma} c^{\mu}_{\alpha,\beta} c^{\nu}_{\alpha,\gamma} c^{\lambda}_{\beta,\gamma},
\end{equation}
where $\alpha,\beta,\gamma$ are Young diagrams.

Note again that the Newell-Littlewood numbers are not dependent on $d$ insofar as that the dependence is encoded in the requirements of Theorem \ref{th:nl-formula}. This is why it is unnecessary to specify the degree $d$ of the orthogonal group to which the representation $\Repp_\lambda$ or $\Repp_{\lambda^{*}}$ belongs to.

Some facts about Newell-Littlewood numbers from Lemma 2.2 in \cite{gao2021newell}, where $|\mu|,|\nu|$ and $|\lambda|$ denote the number of boxes in the respective Young diagram:

\begin{enumerate}
    \item $N^{\lambda}_{\mu,\nu}$ is invariant under any $\sym_3$-permutation of the indices $(\mu,\nu,\lambda)$.
    \item $N^{\lambda}_{\mu,\nu}=c^{\lambda}_{\mu,\nu}$ if $|\mu|+|\nu|=|\lambda|$.
    \item $N^{\lambda}_{\mu,\nu}=0$ unless $|\mu|,|\nu|,|\lambda|$ satisfy the triangle inequalities, i.e., $|\mu|+|\nu|\geq |\lambda|$, $|\mu|+|\lambda|\geq |\nu|$ and $|\lambda|+|\nu|\geq |\mu|$.
    \item $N^{\lambda}_{\mu,\nu}=0$ unless $|\lambda|+|\mu|+|\nu|\equiv 0 \;(\mathrm{mod}\; 2)$.
\end{enumerate}

Note that in the case of $\mathrm{depth}(\mu)+\mathrm{depth}(\nu)>k$ the formula is still valid, but some modification rules are needed to make sense of a Young diagram with depth greater than $k$ as a non-standard label for an $\Ort(d)$ representation. See Appendix \ref{sec:app_mod}.

Similarly to the case of $\U(d)$, a Pieri-type rule can be stated, which is enough for most relevant calculations (see Proposition 2.4 in \cite{gao2021newell})

\begin{corollary}[Pieri-type rule]
    Let $\Repp_{\mu}$ and $\Repp_{\nu}$ be tensor irreps of $\Ort(d)$ ($d=2k$ or $d=2k+1$) labelled by Young diagrams with $\mu$ of depth smaller than $k$ and $\nu$ of depth $1$. Their tensor product decomposes as
    \begin{equation}\label{eq:ort_pieri}
        \Repp_{\mu}\ot\Repp_{\nu}\cong \bigoplus_{\lambda} \Repp_{\lambda}\text{,}
    \end{equation}
    where the sum is over the multiset of $\lambda$ obtained from $\mu$ by removing $0\leq j\leq |\nu|$ number of boxes each from different columns (while keeping it a Young diagram) and then adding $|\nu|-j$ number of boxes each to different columns (while keeping it a Young diagram). Equation \eqref{eq:ort_pieri} extends naturally to the cases with associate irreps.
\end{corollary}

Naturally, the number of columns $\mu$ possesses limits the number of boxes that can be removed from it, therefore, the range $0\leq j\leq |\nu|$ is not always exhausted. Note also that irreps with greater than one multiplicity might appear, leading to $\lambda$ being an element of a multiset.

In the case of $\nu=\ydiagram{1}\;$, this means that all diagrams appear that come from $\mu$ having a single box removed or added to its different columns, while keeping it a valid Young diagram.

\subsection{Branching rules from $\U(d)$ to $\Ort(d)$}\label{sec:app_branch}

The branching rules from $\U(d)$ to $\Ort(d)$ can be described in terms of Littlewood-Richardson coefficients.

According to Proposition 1.5.3 in \cite{koike1987young} we have (see also \cite{littlewood1958products} and Section 4 in \cite{king1975branching} and Theorem 1.1.~in \cite{howe2005stable}):

\begin{proposition}[Littlewood]
    Let $\lambda$ be a Young diagram whose depth is less than or equal to $k$ for either $d=2k$ or $d=2k+1$, and let $\Rep_\lambda$ denote the corresponding tensor irrep of $\U(d)$. Then, assuming the natural embedding of $\Ort(d)$ in $\U(d)$, the following is true:
    \begin{equation}\label{eq:littlewood}
        \Rep_{\lambda}\Big\rvert_{\Ort(d)} \cong \bigoplus_{\mu} \left(\sum_{\kappa} c_{2\kappa,\mu}^{\lambda}\right) \Repp_{\mu}\text{,}
    \end{equation}
    where $\mu$ and $\kappa$ are Young diagrams, $2\kappa$ denotes the Young diagram with double the boxes of $\kappa$ in each row, and $\Repp_{\mu}$ is the tensor irrep of $\Ort(d)$ associated to the $\mu$ label.
\end{proposition}

Note that in the case of $\mathrm{depth}(\lambda)>k$ the formula is still valid, but modification rules are needed to make sense of the $\Ort(d)$ representations labelled by Young diagrams with depth greater than $k$. See Appendix \ref{sec:app_mod}.

\subsection{Modification rules}\label{sec:app_mod}

In Appendices \ref{sec:nl} and \ref{sec:app_branch} it has become apparent that a way to interpret $\Ort(d)$ representations labelled by Young diagrams with more than $k$ rows is needed. This is where the modification rules come in. For more information, see Equation (2.3) in \cite{king1975branching}, or \cite{king2003modification}, or Section 2.4 in \cite{koike1987young}.

\begin{theorem}
    Let $d=2k$ or $d=2k+1$ be fixed and let $\lambda$ be a Young diagram with $p>k$ rows. The $\Ort(d)$ tensor irrep $\Repp_{\lambda}$ with non-standard label $\lambda$ can be interpreted as a tensor irrep with a standard label by (repeatedly) modifying the Young diagram $\lambda$ until it either gives a valid Young diagram or turns to zero (in which case there is no such representation). The modified diagram $\lambda'$ is described by the following equation:
    \begin{equation} 
        \lambda' = (-1)^{x-1} [\lambda - h]^{*} \text{,}
    \end{equation}
    where $h\coloneqq 2p-d$ and $[\lambda-h]$ describes the Young diagram $\lambda$ from which a continuous boundary hook has been removed: Starting from the bottom-leftmost box, mark off $h$ number of boxes along the outer boundary, always going to the box on the right if possible or to the box above if not. If the diagram $[\lambda-h]$ is a valid Young diagram (taken as a diagram with the same first column), then that is the modified diagram. If not, the procedure leads to zero, as there is no such representation. The number $x$ designates the column in which the continuous boundary hook removal procedure ends.
\end{theorem}

\begin{example}
    In the case of $\Ort(5)$ the following are true:
    \begin{equation}
        \ytableausetup{mathmode, boxsize=0.75em}
        \ydiagram{2,2,1} = \ydiagram{2,2}^{*}\text{,} \qquad \qquad \ydiagram{2,2,1,1} = \emptyset\text{.}
        \ytableausetup{mathmode, boxsize=0.5em}
    \end{equation}
\end{example}
\noindent Further examples can be inferred from \cite{king2003modification}.

Note that the second example above does not lead to the trivial ($0$) or the determinant ($0^{*}$) representation, it leads to no representation. Note also that representations having a negative coefficient only appear naturally alongside representations containing the same representation with a positive coefficient, such that these two cancel each other out.

In practice, after each step of a calculation, it is wisest to modify the appearing Young diagrams into their standard label form and continue with the calculation using those.

\setcounter{equation}{0}

\section{Using Schur-Weyl duality for the extendibility problem of Brauer states} \label{sec:app_sw}

This appendix deals with how to apply Schur-Weyl duality to the problem of extendibility of Brauer states. First, a recapitulation of Schur-Weyl duality is presented, but the interested reader is directed to Section 4.2.4 and Chapter 9 in \cite{goodman2009symmetry} and Chapter 9 in \cite{goodman1998representation} for a more in-depth analysis. Then Schur-Weyl duality is applied to the problem of extendibility of Brauer states with specific examples for the low extendibility number cases investigated in this article.

Note that there exists a type of Schur-Weyl duality for orthogonal groups discovered by Brauer \cite{brauer1937algebras}. However, it is not needed in this appendix.

\subsection{Recap: Schur-Weyl duality}

Central to this article is a theorem called Schur-Weyl duality, which (in its most general form) connects the irreducible finite-dimensional $N$-degree tensor irreps of the general linear group of degree $d$ and the irreducible representations of the symmetric group of degree $N$\footnote{For more information on the irreps of the symmetric group, the interested reader is directed to \cite{kirillov1976elements,fulton2004representation}.}, see Theorem 9.1.2 in \cite{goodman2009symmetry}. However, the theorem can also be stated for the unitary group, see Section 3.2.~in \cite{jakab2022interplay}, and this is the preferred form for this article:

\begin{theorem}[Schur-Weyl 1]
    Let $\hil$ be a complex $d$-dimensional Hilbert space, and $N\in\nat$. Let $\hat{\Rep}_N:\U(d) \to \GL(\hil^{\ot N})$ be the so called diagonal action of the defining representation of the unitary group on $\hil^{\ot N}$:
    \begin{equation}\ytableausetup{mathmode, boxsize=\Yt}
        \hat{\Rep}_N(g)\coloneqq \Rep_{\ydiagram{1}}(g)^{\ot N}\text{,}\vspace{0.1cm}
    \end{equation}where $g\in\U(d)$. Let $\hat{\Repppp}_N:\sym_N\to \GL(\hil^{\ot N})$ be the natural representation of the symmetric group on $\hil^{\ot N}$ that permutes the Hilbert spaces:
    \begin{equation}
        \hat{\Repppp}_N(h) (v_1 \ot v_2 \ot \cdots \ot v_N )\coloneqq v_{h^{-1}(1)} \ot v_{h^{-1}(2)} \ot \cdots \ot v_{h^{-1}(N)}\text{,}
    \end{equation}
    where $h\in\sym_n$ and $v_i\in\hil$.
    These two representations commute with each other and the joint action of $\U(d)$ and $\sym_N$ decomposes into irreducible components as:
    \begin{equation}
        \hat{\Repppp}_N(h)\hat{\Rep}_N(g) \cong \bigoplus_{\lambda \in \Yng(N,d)} \Repppp_{\lambda}(h) \ot \Rep_{\lambda}(g)\text{,}
    \end{equation}
    where $\Yng(N,d)$ denotes the Young diagrams with $N$ boxes and at most $d$ rows.
\end{theorem}

This also means that for the $S_N$ irrep labelled by $\lambda$, the dimension of $\Rep_\lambda$ gives the multiplicity of that given irrep in the decomposition, and similarly, for the $\U(d)$ irrep labelled by $\lambda$, the dimension of $\Repppp_\lambda$ gives the multiplicity of that given irrep in the decomposition.

The following is an equivalent formulation of the above theorem:

\begin{theorem}[Schur-Weyl 2]\label{th:sw2}
    The algebra $\Span_\comp\{\hat{\Rep}_N(g)\in\GL(\hil^{\ot n}) \mid g\in \U(d)\}$ and the algebra $\Span_\comp\{\hat{\Repppp}_N(h)\in\GL(\hil^{\ot n}) \mid h\in \sym_N\}$ are the mutual commutants of each other.
\end{theorem}

In the following sections, we present the application of Schur-Weyl duality to our extendibility problem. For ease of notation, the variables of the representations (such as $g$ or $h$) will often be suppressed. Beforehand, we illustrate Schur-Weyl duality with the following example:

\begin{example}
    For the case of $N=4$, Schur-Weyl duality gives the following for $d>3$:
    \begin{equation}
        \hat{\Repppp}_4 \hat{\Rep}_4 \cong \left[\Repppp_{\ydiagram{1,1,1,1}}\ot\Rep_{\ydiagram{1,1,1,1}}\right] \OP \left[\Repppp_{\ydiagram{2,1,1}}\ot\Rep_{\ydiagram{2,1,1}} \right] \OP \left[\Repppp_{\ydiagram{2,2}}\ot\Rep_{\ydiagram{2,2}}\right] \OP \left[\Repppp_{\ydiagram{3,1}}\ot\Rep_{\ydiagram{3,1}} \right]\OP \left[\Repppp_{\ydiagram{4}}\ot\Rep_{\ydiagram{4}} \right]\text{,}
    \end{equation}
    which is in agreement with the fusion rules of $\U(d)$. Here, one has to point out that $\Repppp_{\ydiagram{2,1,1}}$ and $\Repppp_{\ydiagram{3,1}}$ are 3-dimensional representations and carry the multiplicity information of $\Rep_{\ydiagram{2,1,1}}$ and $\Rep_{\ydiagram{3,1}}$, respectively. Similarly, $\Repppp_{\ydiagram{2,2}}$ is a 2-dimensional representation and carries the multiplicity information of $\Rep_{\ydiagram{2,2}}$. Because of this, it is unwise to suppress the $\sym_3$ irreps in the notation.
\end{example}

Note that for the lower-dimensional cases, one ought to consult Appendix \ref{sec:app_rep}.

\subsection{Application to the extendibility problem of Brauer states}

As detailed in Section \ref{sec:recipe}, to find the set of $(n,m)$-extendible Brauer states, the following requirements are prescribed for the set of extending states $\{\hat{\rho}\}$ acting on $\hil^{\ot n}\ot\hil^{\ot m}$:
\begin{enumerate}
    \item $\comm{\hat{\rho}}{\Repp_{\ydiagram{1}}^{\ot (n+m)}(g)}\overset{!}{=}0$ for all $g\in\Ort(d)$,
    \item $\comm{\hat{\rho}}{\hat{\Repppp}^{\mathrm{L}}_{n}(h)\ot\hat{\Repppp}^{\mathrm{R}}_{m}(k)}\overset{!}{=}0$ for all $h\in \sym_n$ and $k\in \sym_m$,
\end{enumerate}
where the upper index L or R is there to distinguish the representation of $\sym_n$ and $\sym_m$ even for the case of $n=m$: These representations cannot be fused as they are representations of $\sym_n \times \sym_n$ and not simply $\sym_n$. This is reflected in the fact that even for the $n=m$ case, the arguments are independent ($h$ and $k$). Contrast this with the representation of the orthogonal group that is supposed to be fused together, as attested by the fact that only one argument $g$ is considered. 

These two requirements can be thought of as one in the following way:
\begin{equation}
    \comm{\hat{\rho}}{\left(\hat{\Repppp}^{\mathrm{L}}_{n}(h)\hat{\Repp}_{n}(g)\right)\ot \left(\hat{\Repppp}^{\mathrm{R}}_{m}(k)\hat{\Repp}_{m}(g)\right)}\overset{!}{=}0\text{,} \qquad \forall g\in\Ort(d), \; \forall h\in\sym_n, \; \forall k\in\sym_m,
\end{equation}
where $\hat{\Repp}_n(g)\coloneqq \Repp_{\ydiagram{1}}(g)^{\ot n}$.
The irrep decomposition of the representations in parentheses can be calculated using Schur-Weyl duality. For general $N$, with the variables suppressed, we have:
\begin{equation}
    \hat{\Repppp}_{N}\hat{\Repp}_{N}=\hat{\Repppp}_N \hat{\Rep}_{N}\Big\rvert_{\Ort(d)} \cong \bigoplus_{\lambda \in \Yng(N,d)} \Repppp_{\lambda} \ot \Rep_{\lambda}\Big\rvert_{\Ort(d)}\text{.}
\end{equation}
Of course, one has to handle the restrictions carefully for small $d$ cases, see Appendix \ref{sec:app_rep}.

Note that the decomposition can also be done using Theorems 10.2.9 and 10.2.12 in \cite{goodman2009symmetry}. Nevertheless, using the above equation:
\begin{align}
    \left(\hat{\Repppp}^{\mathrm{L}}_{n}\hat{\Repp}_{n}\right)\ot \left(\hat{\Repppp}^{\mathrm{R}}_{m}\hat{\Repp}_{m}\right)&\cong \left(\bigoplus_{\mu \in \Yng(n,d)} \Repppp^{\mathrm{L}}_{\mu} \ot \Rep_{\mu}\Big\rvert_{\Ort(d)}\right) \ot \left(\bigoplus_{\nu \in \Yng(m,d)} \Repppp^{\mathrm{R}}_{\nu} \ot \Rep_{\nu}\Big\rvert_{\Ort(d)}\right)\\ 
    &\cong \bigoplus_{\substack{\mu \in \Yng(n,d) \\ \nu \in \Yng(m,d)}} \Repppp^{\mathrm{L}}_{\mu} \ot \Repppp^{\mathrm{R}}_{\nu} \ot \left(\Rep_{\mu}\Big\rvert_{\Ort(d)} \ot \Rep_{\nu}\Big\rvert_{\Ort(d)}\right)\text{,}
\end{align}
where one has to calculate the irrep decomposition of the fusion of the two orthogonal irreps in the last parentheses using the fusion rules. This is hard to do sensibly in the general case, but can be done on a case-by-case basis. If one knows the proper irrep decomposition on the right-hand side of the above equation, one can use Schur's lemma to calculate the general commutant element $\hat{\rho}$ and proceed with the recipe.

In the following, the cases for small $n$ and $m$ are discussed. Note that because of the flip invariance of Brauer states, it is enough to consider the $n\leq m$ cases.

\subsubsection{The case of (1,2)-extendibility}

Given the requirements imposed on the extending state $\hat{\rho}_{(1,2)}\in\states{\hil\ot\hil^{\ot 2}}$, the irrep decomposition of the representation it has to be invariant to is the following for $d>5$:
\begin{align}
    \left(\hat{\Repppp}^{\mathrm{L}}_1\hat{\Repp}_1 \right)\ot \left( \hat{\Repppp}^{\mathrm{R}}_2 \hat{\Repp}_2\right) \cong& 
    \left[\Repppp^{\mathrm{L}}_{\ydiagram{1}} \ot \Repp_{\ydiagram{1}} \right] \ot \left[\Repppp^{\mathrm{R}}_{\ydiagram{1,1}}\ot\Repp_{\ydiagram{1,1}} \OP \Repppp^{\mathrm{R}}_{\ydiagram{2}}\ot\left(\Repp_{0} \OP \Repp_{\ydiagram{2}}\right) \right]\\
    \begin{split}
        \cong & \; \Repppp^{\mathrm{L}}_{\ydiagram{1}} \ot \Repppp^{\mathrm{R}}_{\ydiagram{1,1}} \ot \left(\Repp_{\ydiagram{1}} \OP \Repp_{\ydiagram{1,1,1}} \OP \Repp_{\ydiagram{2,1}} \right) \OP \\
        & \; \Repppp^{\mathrm{L}}_{\ydiagram{1}} \ot \Repppp^{\mathrm{R}}_{\ydiagram{2}}\ot\left({\color{red}\Repp_{\ydiagram{1}}} \OP {\color{red}\Repp_{\ydiagram{1}}} \OP \Repp_{\ydiagram{2,1}}\OP \Repp_{\ydiagram{3}}\right) \text{,}
    \end{split}
\end{align}
where the only irrep with multiplicity greater than one is highlighted in red. Note that for smaller dimensions one has to consult Appendix \ref{sec:app_rep}.

\subsubsection{The case of (1,3)-extendibility}

Given the requirements imposed on the extending state $\hat{\rho}_{(1,3)}\in\states{\hil\ot\hil^{\ot 3}}$ the irrep decomposition of the representation it has to be invariant to is the following for $d>7$:
\begin{align}
    \left(\hat{\Repppp}^{\mathrm{L}}_1\hat{\Repp}_1 \right)\ot \left( \hat{\Repppp}^{\mathrm{R}}_3 \hat{\Repp}_3\right) \cong& \left[\Repppp^{\mathrm{L}}_{\ydiagram{1}} \ot \Repp_{\ydiagram{1}} \right] \ot \left[\Repppp^{\mathrm{R}}_{\ydiagram{1,1,1}}\ot\Repp_{\ydiagram{1,1,1}} \OP \Repppp^{\mathrm{R}}_{\ydiagram{2,1}}\ot\left(\Repp_{\ydiagram{1}}\OP\Repp_{\ydiagram{2,1}}\right) \OP \Repppp^{\mathrm{R}}_{\ydiagram{3}}\ot\left(\Repp_{\ydiagram{1}}\OP\Repp_{\ydiagram{3}}\right)\right] \\
    \begin{split}
        \cong & \; \Repppp^{\mathrm{L}}_{\ydiagram{1}} \ot \Repppp^{\mathrm{R}}_{\ydiagram{1,1,1}}\ot \left(\Repp_{\ydiagram{1,1}} \OP \Repp_{\ydiagram{1,1,1,1}} \OP \Repp_{\ydiagram{2,1,1}} \right) \OP \\
        & \; \Repppp^{\mathrm{L}}_{\ydiagram{1}} \ot \Repppp^{\mathrm{R}}_{\ydiagram{2,1}}\ot \left(\Repp_{0} \OP {\color{orange}\Repp_{\ydiagram{1,1}}}\OP {\color{orange}\Repp_{\ydiagram{1,1}}}\OP {\color{brass}\Repp_{\ydiagram{2}}} \OP {\color{brass}\Repp_{\ydiagram{2}}} \OP \Repp_{\ydiagram{2,1,1}} \OP \Repp_{\ydiagram{2,2}} \OP \Repp_{\ydiagram{3,1}} \right) \OP \\
        & \; \Repppp^{\mathrm{L}}_{\ydiagram{1}} \ot \Repppp^{\mathrm{R}}_{\ydiagram{3}}\ot \left(\Repp_{0} \OP \Repp_{\ydiagram{1,1}} \OP {\color{brown}\Repp_{\ydiagram{2}}}\OP {\color{brown}\Repp_{\ydiagram{2}}} \OP \Repp_{\ydiagram{3,1}} \OP \Repp_{\ydiagram{4}} \right)\text{.}
    \end{split}
\end{align}
Note that the coloured representations have multiplicity greater than one and that, again, for smaller dimensions one has to consult Appendix \ref{sec:app_rep}.

\subsubsection{The case of (2,2)-extendibility}

Given the basic requirements imposed on the extending state $\hat{\rho}_{(2,2)}\in\states{\hil^{\ot 2}\ot\hil^{\ot 2}}$ the irrep decomposition of the representation it has to be invariant to is the following for $d>7$:
\begin{align}
    \left(\hat{\Repppp}^{\mathrm{L}}_2\hat{\Repp}_2 \right)\ot \left( \hat{\Repppp}^{\mathrm{R}}_2 \hat{\Repp}_2\right) \cong& 
    \left[\Repppp^{\mathrm{L}}_{\ydiagram{1,1}}\ot\Repp_{\ydiagram{1,1}} \OP \Repppp^{\mathrm{L}}_{\ydiagram{2}}\ot\left(\Repp_{0} \OP \Repp_{\ydiagram{2}}\right) \right] \ot \left[\Repppp^{\mathrm{R}}_{\ydiagram{1,1}}\ot\Repp_{\ydiagram{1,1}} \OP \Repppp^{\mathrm{R}}_{\ydiagram{2}}\ot\left(\Repp_{0} \OP \Repp_{\ydiagram{2}}\right) \right] \\
    \begin{split}
        \cong & \; \Repppp^{\mathrm{L}}_{\ydiagram{1,1}}\ot \Repppp^{\mathrm{R}}_{\ydiagram{1,1}} \ot \left(\Repp_{0} \OP \Repp_{\ydiagram{1,1}} \OP \Repp_{\ydiagram{2}} \OP \Repp_{\ydiagram{1,1,1,1}} \OP \Repp_{\ydiagram{2,1,1}} \OP \Repp_{\ydiagram{2,2}} \right) \OP\\
        & \; \Repppp^{\mathrm{L}}_{\ydiagram{1,1}}\ot \Repppp^{\mathrm{R}}_{\ydiagram{2}} \ot \left({\color{blue-green-v1}\Repp_{\ydiagram{1,1}}} \OP {\color{blue-green-v1}\Repp_{\ydiagram{1,1}}} \OP \Repp_{\ydiagram{2}} \OP \Repp_{\ydiagram{2,1,1}} \OP \Repp_{\ydiagram{3,1}}\right) \OP \\
        & \; \Repppp^{\mathrm{L}}_{\ydiagram{2}}\ot \Repppp^{\mathrm{R}}_{\ydiagram{1,1}} \ot \left({\color{blue-green-v2}\Repp_{\ydiagram{1,1}}} \OP {\color{blue-green-v2}\Repp_{\ydiagram{1,1}}} \OP \Repp_{\ydiagram{2}} \OP \Repp_{\ydiagram{2,1,1}} \OP \Repp_{\ydiagram{3,1}} \right) \OP\\
        & \; \Repppp^{\mathrm{L}}_{\ydiagram{2}}\ot \Repppp^{\mathrm{R}}_{\ydiagram{2}}\ot \left({\color{blue}\Repp_{0}} \OP {\color{blue}\Repp_{0}} \OP \Repp_{\ydiagram{1,1}} \OP {\color{green} \Repp_{\ydiagram{2}}} \OP {\color{green} \Repp_{\ydiagram{2}}} \OP {\color{green} \Repp_{\ydiagram{2}}} \OP \Repp_{\ydiagram{2,2}} \OP \Repp_{\ydiagram{3,1}} \OP \Repp_{\ydiagram{4}} \right)\text{.}
    \end{split}
\end{align}
Note that the coloured representations have multiplicity greater than one, the green representation actually has multiplicity three! However, one should not forget that in the $(2,2)$-extendibility case an extra symmetry is present described in Section \ref{sec:nn} and explored in Appendix \ref{sec:app_res3}. This reduces the multiplicity of the green irreps to two and merges the $\Repp_{\ydiagram{1,1}}$ multiplicities together. Note again that for smaller dimensions one has to consult Appendix \ref{sec:app_rep}.

\setcounter{equation}{0}

\section{The quadratic Casimir operator of tensor irreps of $\U(d)$ and $\Ort(d)$}\label{sec:app_cas}

This appendix introduces the quadratic Casimir operator and the relevant $\ualg(d)$ and $\So(d)$ Casimir eigenvalues used in this work. For more information on Lie algebras and the Casimir operator, see References \cite{fulton2004representation,goodman2009symmetry, hall2015lie}.

\subsection{Recap: Lie algebra representations}

Every Lie group representation gives rise to a Lie algebra representation over the same vector space called the differential of the representation (see, e.g., Proposition 1.3.14.~in \cite{goodman2009symmetry}). So far, the relevant Lie group representations have been denoted by capital Greek letters ($\Rep$ for the $\U(d)$ representations and $\Repp$ for the $\Ort(d)$ ones). The corresponding Lie algebra representations will be denoted by lowercase Greek letters: $\rep$ for $\ualg(d)$ representations and $\repp$ for $\So(d)$ representations.

The following proposition describes when an irreducible Lie group representation induces an irreducible Lie algebra representation (see Proposition 1.7.7.~in \cite{goodman2009symmetry}):
\begin{proposition}
    Let $\Rep$ be a polynomial representation of a Lie group $G$ and let $\rep$ be its differential on the Lie algebra $\mathfrak{g}$. The differential representation is irreducible if and only if $\Rep$ is irreducible on the connected subgroup of $G$.
\end{proposition}

The unitary group is connected; therefore, any unitary tensor irrep will give rise to a Lie algebra irrep, thus they will be labelled by the same Young diagrams.

However, the case of the orthogonal group is more involved. The following is true for the connected component $\SO(d)$ (see Theorems 5.5.23.~and 5.5.24.~in \cite{goodman2009symmetry}):
\begin{theorem}
    Let $\Repp_{\lambda}$ be an $\Ort(d)$ tensor irrep where $d=2k$ or $d=2k+1$. Then the following is true for its restriction to $\SO(d)$:
    \begin{equation}
        \Repp_{\lambda} \Big\rvert_{\SO(d)}\cong \Repp_{\lambda^*}\Big\rvert_{\SO(d)} \cong \begin{cases} \tilde{\Repp}_{\lambda}\text{,} \quad &\text{ if $d$ is odd or $d$ is even but }\lambda_k=0\text{,}\\ \tilde{\Repp}_{\lambda} \OP \tilde{\Repp}_{\overline{\lambda}}\text{,} \quad &\text{ if $d$ is even and $\lambda_k\neq 0$,} \end{cases}
    \end{equation}
    where $\tilde{\Repp}$ denotes $\SO(d)$ representations and $\overline{\lambda}$ denotes the generalised Young diagram having a negative $\lambda_k$ number of boxes in its last row.
\end{theorem}

Therefore, in the case of the orthogonal group it is almost as easy as for the unitary group. Two things need to be kept in mind: firstly, there is no difference between the starred and non-starred representations. Secondly, representations of the orthogonal group of even degree that have non-zero boxes in the last line will split up into two irreps, one labelled by the same Young diagram and the other labelled by one that has a negative number of boxes in the last line. This will be represented as dots in the relevant boxes. See the following example for $d=4$:
\begin{equation}
    \Repp_{\ydiagram{2,2}} \Big\rvert_{\SO(d)}\cong \tilde{\Repp}_{\ydiagram{2,2}} \OP \tilde{\Repp}_{\ydiagram[\cdot]{0,2}*{2,0}}\text{,}
\end{equation}
where $\tilde{\Repp}$ denotes $\SO(d)$ representations. Note, however, that the irreps labelled by these generalised Young diagrams will not influence the calculations and results of this article as they have the same quadratic Casimir eigenvalues as their non-generalised counterparts. See Appendix \ref{sec:qCO} for more information.

A simple consequence of the relationship between Lie group representations and the corresponding Lie algebra representation is that they transform in sync. For example, if a Lie group representation $\Rep$ splits up in the following way:
\begin{equation}
    \Rep(g)=\Vry \left(\Rep_{1}(g) \OP \Rep_{2}(g) \right)\Vry^\dagger \text{,}
\end{equation}
for a group element $g$, some subrepresentations $\Rep_1$ and $\Rep_2$ and a fixed unitary $\Vry$. Then the same is true for the corresponding Lie algebra representations:
\begin{equation}
    \rep(X)=\Vry \left(\rep_{1}(X) \OP \rep_{2}(X) \right)\Vry^\dagger \text{,}
\end{equation}
where $X$ is an element of the Lie algebra.

In our calculations, the Lie algebra representation corresponding to the tensor product of Lie group representations will often be needed. The tensor product of Lie algebra representations shall be denoted by the symbol $\times$. This is so that it differentiates between ordinary tensor products, as the tensor product of Lie algebra representations is actually the so-called Kronecker sum of operators. More precisely, for the Lie algebra representations $\rep_1:\lie\to\gl(\hil_1)$ and $\rep_2:\lie\to\gl(\hil_2)$ their tensor product is the following:
\begin{equation}
    (\rep_{1}\times\rep_{2})(X)=\rep_1(X)\ot \id_{\hil_2} + \id_{\hil_1} \ot \rep_2(X)\text{,}
\end{equation}
where $X\in\lie$ is an element of the Lie algebra in question.

\subsection{Recap: The quadratic Casimir operator}

To determine the set $(n,m)$-extendible Brauer states we will be using the quadratic Casimir elements of the $\ualg(d)$ and $\So(d)$ Lie algebras. In general, the quadratic Casimir of the Lie algebra $\mathfrak{g}$ in the representation $\rep$ is the following:
\begin{equation}
    \rep(C^{\mathfrak{g}})=\sum_{i} \rep(X_i)\rep(X^i)\text{,}
\end{equation}
where $\{X_i\}_i$ is a basis in $\mathfrak{g}$ and $\{X^i\}_i$ is its dual\footnote{With respect to, e.g., the Killing form.}. If $\rep$ is an irrep labelled by $\mathfrak{i}$, then, because of Schur's lemma for Lie-algebras:
\begin{equation}
    \rep(C^{\mathfrak{g}})=c(\mathfrak{i})\id_{\mathfrak{i}}\text{,}
\end{equation}
where $c(\mathfrak{i})$ is the eigenvalue of the quadratic Casimir element for the irrep $\mathfrak{i}$.

Note that for simple Lie algebras, the $c(\mathfrak{i})$ eigenvalues are defined up to a global multiplicative constant, for example $\frac{c(\mathfrak{i_1})}{c(\mathfrak{i_2})}$, where defined, is the same for any choice of bilinear form. In this article, the constant for $\ualg(d)$ irreps is denoted by $\chi$ and for $\So(d)$ irreps by $\xi$. 

The eigenvalue of the quadratic Casimir operator of a semisimple Lie algebra is positive for any non-trivial representation. For the trivial representation, all generators of the Lie algebra act as zero and, therefore, the eigenvalue of the Casimir is always zero.

\subsection{The quadratic Casimir operator eigenvalues for $\ualg(d)$}

The formula for the eigenvalue of the quadratic Casimir of $\ualg(d)$ for the irrep with the Young diagram $\lambda$ as a standard label is the following (see Section 7.5 in \cite{iachello2015lie}):
\begin{equation}\label{eq:chi}
    \chi(\lambda)=\sum_{i=1}^{d} \lambda_i (\lambda_i + d + 1 - 2i)\text{,}
\end{equation}
where $\lambda$ has at most $d$ rows.

Thus, the relevant order-two Casimir element eigenvalues for $\ualg(d)$ irreps:
\begin{equation}
    \begin{gathered}
    \chi(0)=0\text{,} \qquad \chi(\ydiagram{1})=d\text{,} \qquad \chi(\ydiagram{1,1})=2(d-1)\text{,} \qquad \chi(\ydiagram{2})=2(d+1)\text{,}\\
    \chi(\ydiagram{1,1,1})=3(d-2)\text{,} \qquad \chi(\ydiagram{2,1})=3d\text{,} \qquad \chi(\ydiagram{3})=3(d+2)\text{,} \qquad \chi\left(\ydiagram{1,1,1,1}\right)=4(d-3)\text{,}\\ \chi(\ydiagram{2,1,1})=4(d-1)\text{,} \qquad \chi(\ydiagram{2,2})=4d\text{,} \qquad \chi(\ydiagram{3,1})=4(d+1)\text{,} \qquad \chi(\ydiagram{4})=4(d+3)\text{.}
    \end{gathered}
\end{equation}
Note that these eigenvalues are defined up to a global multiplicative constant.

\subsection{The quadratic Casimir operator eigenvalues for $\So(d)$} \label{sec:qCO}

The formula for the eigenvalue of the quadratic Casimir of $\So(d)$ for the irrep with the Young diagram $\lambda$ as a standard label is the following (see Section 7.5 in \cite{iachello2015lie}):
\begin{equation}\label{eq:xi}
    \xi(\lambda)=\sum_{i=1}^{k} 2 \lambda_i (\lambda_i + d - 2i)\text{,}
\end{equation}
where $\lambda$ has at most $k$ rows, where $k=\lfloor\frac{d}{2}\rfloor$. Note that the formula is quadratic in the last row of $\lambda$ and therefore, for the special even-dimensional cases $\xi([\lambda_1,\ldots,\lambda_k])=\xi([\lambda_1,\ldots,-\lambda_k])$.

Thus, the relevant order-two Casimir element eigenvalues of $\So(d)$:
\begin{equation}
\begin{gathered}
    \xi(0)=0\text{,} \qquad \xi(\ydiagram{1})=2(d-1)\text{,} \qquad \xi(\ydiagram{1,1})=4(d-2)\text{,} \qquad \xi(\ydiagram{2})=4d\text{,}\\
    \xi(\ydiagram{1,1,1})=6(d-3)\text{,} \qquad \xi(\ydiagram{2,1})=6(d-1)\text{,} \qquad \xi(\ydiagram{3})=6(d+1)\text{,} \qquad \xi\left(\ydiagram{1,1,1,1}\right)=8(d-4)\text{,}\\ \xi(\ydiagram{2,1,1})=8(d-2)\text{,} \qquad \xi(\ydiagram{2,2})=8(d-1)\text{,} \qquad \xi(\ydiagram{3,1})=8d\text{,} \qquad \xi(\ydiagram{4})=8(d+2)\text{.}
\end{gathered}
\end{equation}
Note that these eigenvalues are defined up to a global multiplicative constant.

Furthermore, one can obtain a theorem to deal with the cases of non-standard irrep labels, Young diagrams with more than $k$ rows:

\begin{theorem}\label{th:mod_cas}
    Fix the dimension $d=2k$ or $d=2k+1$ and let $\lambda$ be a Young diagram with more than $k$ rows, thus being a non-standard label for an $\Ort(d)$ tensor irrep $\Repp_{\lambda}$. If there exists a Young diagram $\lambda'$ which is a standard irrep label that can be obtained using the modification rules, then it is true that for the fixed dimension $d$ the quadratic Casimir eigenvalues of these two are the same, that is:
    \begin{equation}
    	\xi(\lambda)\vert_d=\xi(\lambda')\vert_d\text{,}
    \end{equation}
    where Equation \eqref{eq:xi} is adapted for $\lambda$ by extending the summation from $k$ to $d$.
\end{theorem}
\noindent The proof of this theorem follows from \cite{newell1951modification}.

Note that based on Equation \eqref{eq:xi} and Theorem \ref{th:mod_cas}, for a fixed dimension $d$ and a Young diagram $\lambda$ of at most $d$ rows that can be transformed into a valid $\Ort(d)$ irrep label we have that:
\begin{equation}\label{eq:chixi}
    \xi(\lambda)=2(\chi(\lambda)-|\lambda|)\text{,}
\end{equation}
where, again, the summation in Equation \eqref{eq:xi} is extended from $k$ to $d$.

\setcounter{equation}{0}

\section{Calculations for $(n,m)$-extendibility} \label{sec:app_calcs}

In this appendix, we determine the relationship between $\flip_{jk}$, $\bb_{jk}$ and the respective Casimirs. Let $\rep_{\ydiagram{1}}:\ualg(d)\to\gl(\hil)$ be the Lie algebra representation corresponding to the defining representation. We have the following for $d>1$\footnote{The $d=1$ is an almost trivial, uninteresting case.}:
\begin{align}
    \rep_{\ydiagram{1}}\times\rep_{\ydiagram{1}}\left(C^{\ualg(d)}\right)&=\sum_{i=1}^{d^2}\left[\rep_{\ydiagram{1}}(X_i)\ot\id_{\hil} + \id_{\hil}\ot\rep_{\ydiagram{1}}(X_i)\right]\left[\rep_{\ydiagram{1}}(X^i)\ot\id_{\hil} + \id_{\hil}\ot\rep_{\ydiagram{1}}(X^i)\right] \\
    &=2\chi(\ydiagram{1})\id_{\hil}\ot\id_{\hil} + \sum_{i=1}^{d^2}\rep_{\ydiagram{1}}(X_i)\ot\rep_{\ydiagram{1}}(X^i) + \rep_{\ydiagram{1}}(X^i)\ot \rep_{\ydiagram{1}}(X_i)\\
    &=\chi(\ydiagram{1,1}) \hat{\PROJ}_{\ydiagram{1,1}} + \chi(\ydiagram{2}) \hat{\PROJ}_{\ydiagram{2}}\\
    &=\left[\frac{\chi(\ydiagram{2})+\chi(\ydiagram{1,1})}{2}\right]\id_{\hil}\ot\id_{\hil} + \left[\frac{\chi(\ydiagram{2})-\chi(\ydiagram{1,1})}{2}\right] \flip\text{,}
\end{align}
where $\{X_i\}_{i=1}^{d^2}$ is a basis in $\ualg(d)$ and $\{X^i\}_{i=1}^{d}$ is the corresponding dual basis. In the third line we have used that the irrep decomposition of Lie algebra representations is in sync with that of the appropriate Lie group representations. Furthermore, given that the Casimir operator acts as a constant times the identity on any given irreducible component, we have introduced $\hat{\PROJ}_\lambda$ to denote the projection onto the $\lambda$ subspace which is invariant under all $\Rep_{\ydiagram{1}}\ot\Rep_{\ydiagram{1}}\cong\Rep_{\ydiagram{1,1}}\OP\Rep_{\ydiagram{2}}$ transformations. We also have:
\begin{equation}
    \rep_{\ydiagram{1}}^{\times n}\left(C^{\ualg(d)}\right)=\sum_{i=1}^{d^2}\sum_{j,k=1}^{n}\rep_{\ydiagram{1}}^{(j)}(X_i)\rep_{\ydiagram{1}}^{(k)}(X^i)\text{.}
\end{equation}
The sum on the right-hand side can be decomposed into four cases in the case of $\rep_{\ydiagram{1}}^{\times (n+m)}\left(C^{\ualg(d)}\right)$. One where $j=k$, one where $j,k\leq n$ but they are not equal, one where $n+1\leq j,k$ but they are not equal, and one where either $j\leq n$ and $n+1\leq k$ or the other way around:
\begin{align}
\begin{split}
    \rep_{\ydiagram{1}}^{\times (n+m)}\left(C^{\ualg(d)}\right)=\sum_{i=1}^{d^2}\left[\vphantom{\underset{\substack{j,k=1 \\ j=k}}{\sum^{n+m}}} \right. &\underset{\substack{j,k=1 \\ j=k}}{\sum^{n+m}}\rep_{\ydiagram{1}}^{(j)}(X_i)\rep_{\ydiagram{1}}^{(k)}(X^i) + \\
    &\underset{\substack{j,k=1 \\ j\neq k}}{\sum^{n}} \rep_{\ydiagram{1}}^{(j)}(X_i)\rep_{\ydiagram{1}}^{(k)}(X^i)+\underset{\substack{j,k=n+1 \\ j\neq k}}{\sum^{n+m}} \rep_{\ydiagram{1}}^{(j)}(X_i)\rep_{\ydiagram{1}}^{(k)}(X^i)+\\
    &\sum_{j=1}^{n}\sum_{k=n+1}^{n+m} \rep_{\ydiagram{1}}^{(j)}(X_i)\rep_{\ydiagram{1}}^{(k)}(X^i)+\sum_{j=n+1}^{n+m} \sum_{k=1}^{n}\rep_{\ydiagram{1}}^{(j)}(X_i)\rep_{\ydiagram{1}}^{(k)}(X^i)\left.\vphantom{\underset{\substack{j,k=1 \\ j=k}}{\sum^{n+m}}} \right]\text{.}
\end{split}
\end{align}
The first of these can be replaced by a multiple of the identity, the second two by appropriate parts containing the `sub'-Casimir operators, and the last part will contain the needed combination of flips:
\begin{align}
\begin{split}
    \rep_{\ydiagram{1}}^{\times (n+m)}\left(C^{\ualg(d)}\right)=&(n+m)\chi(\ydiagram{1})\id^{\ot (n+m)}_{\hil} +\\
    &\rep_{\ydiagram{1}}^{\times n}\left(C^{\ualg(d)}\right) \ot \id^{\ot m}_{\hil} - n\chi(\ydiagram{1})\id^{\ot (n+m)}_{\hil} + \\
    &\id^{\ot n}_{\hil}\ot\rep_{\ydiagram{1}}^{\times m}\left(C^{\ualg(d)}\right) - m\chi(\ydiagram{1})\id^{\ot (n+m)}_{\hil}+ \\
    &nm\frac{\chi(\ydiagram{2})+\chi(\ydiagram{1,1})-4\chi(\ydiagram{1})}{2}\id^{\ot (n+m)}_{\hil} + \frac{\chi(\ydiagram{2})-\chi(\ydiagram{1,1})}{2}\sum_{j=1}^{n}\sum_{k=n+1}^{n+m} \flip_{jk}\text{,}
\end{split}
\end{align}
which can be rearranged to get:
\begin{align}
    \smf_{n,m}\coloneqq&\frac{1}{nm}\sum_{j=1}^{n}\sum_{k=n+1}^{n+m} \flip_{jk}=\\
    \begin{split}
        =&\frac{2}{nm(\chi(\ydiagram{2})-\chi(\ydiagram{1,1}))}\left[\rep_{\ydiagram{1}}^{\times (n+m)}(C^{\ualg(d)}) - \rep_{\ydiagram{1}}^{\times n}(C^{\ualg(d)}) \ot \id^{\ot m}_{\hil} - \id^{\ot n}_{\hil} \ot \rep_{\ydiagram{1}}^{\times m}(C^{\ualg(d)}) \right]\\
        & + \frac{4\chi(\ydiagram{1})-\chi(\ydiagram{2})-\chi(\ydiagram{1,1})}{\chi(\ydiagram{2})-\chi(\ydiagram{1,1})}\id^{\ot (n+m)}_{\hil}\text{,}
    \end{split}
\end{align}
where we have used the fact that $\chi(\ydiagram{2})-\chi(\ydiagram{1,1})\neq 0$. Note, however, that for any choice of multiplicative constant of the $\ualg(d)$ Casimir eigenvalues we have that the coefficient of the identity is zero. Thus:
\begin{align}
    \smf_{n,m}\coloneqq&\frac{1}{nm}\sum_{j=1}^{n}\sum_{k=n+1}^{n+m} \flip_{jk}=\\
    \begin{split}
        =&\frac{2}{nm(\chi(\ydiagram{2})-\chi(\ydiagram{1,1}))}\left[\rep_{\ydiagram{1}}^{\times (n+m)}(C^{\ualg(d)}) - \rep_{\ydiagram{1}}^{\times n}(C^{\ualg(d)}) \ot \id^{\ot m}_{\hil} - \id^{\ot n}_{\hil} \ot \rep_{\ydiagram{1}}^{\times m}(C^{\ualg(d)}) \right]\text{.}
    \end{split}
\end{align}

A similar line of reasoning leads to acquiring the form of $\smb_{n,m}$. Let $\repp_{\ydiagram{1}}:\So(d)\to\gl(\hil)$ be the Lie algebra representation corresponding to the defining representation. The starting point is the following for $d>4$:
\begin{align}\label{eq:cas_2_so}
    \repp_{\ydiagram{1}}\times\repp_{\ydiagram{1}}\left(C^{\So(d)}\right)&=\sum_{i=1}^{\frac{d(d-1)}{2}}\left[\repp_{\ydiagram{1}}(Y_i)\ot\id_{\hil} + \id_{\hil}\ot\repp_{\ydiagram{1}}(Y_i)\right]\left[\repp_{\ydiagram{1}}(Y^i)\ot\id_{\hil} + \id_{\hil}\ot\repp_{\ydiagram{1}}(Y^i)\right] \\
    &=2\xi(\ydiagram{1})\id_{\hil}\ot\id_{\hil} + \sum_{i=1}^{\frac{d(d-1)}{2}}\repp_{\ydiagram{1}}(Y_i)\ot\repp_{\ydiagram{1}}(Y^i) + \repp_{\ydiagram{1}}(Y^i)\ot \repp_{\ydiagram{1}}(Y_i)\\
    &=\xi(\ydiagram{1,1}) \PROJ_{\ydiagram{1,1}} + \xi(0) \PROJ_{0} + \xi(\ydiagram{2}) \PROJ_{\ydiagram{2}}\\
    &=\left[\frac{\xi(\ydiagram{2})+\xi(\ydiagram{1,1})}{2}\right]\id_{\hil}\ot\id_{\hil} + \left[\frac{\xi(\ydiagram{2})-\xi(\ydiagram{1,1})}{2}\right] \flip - \xi(\ydiagram{2})\bb\text{,}
\end{align}
where $\{Y_i\}_{i=1}^{d(d-1)/2}$ is a basis in $\So(d)$, $\{Y^i\}_{i=1}^{d(d-1)/2}$ is the corresponding dual basis and, similarly to the previous case, we have the appropriate Casimir eigenvalues appearing along with $\PROJ_\lambda$ which denotes the projection onto the $\lambda$ subspace which is invariant under all $\Repp_{\ydiagram{1}}\ot\Repp_{\ydiagram{1}}\cong \Repp_0 \OP \Repp_{\ydiagram{1,1}} \OP \Repp_{\ydiagram{2}}$ transformations.

The low-dimensional cases ($d=2,3,4$) require some extra care that is discussed in Appendix \ref{sec:app_cas_lowd}, but they all ultimately lead to the equation above.
which leads to:
\begin{align}
    \smb_{n,m}\coloneqq&\frac{1}{nm}\sum_{j=1}^{n}\sum_{k=n+1}^{n+m} \bb_{jk}=\\
    \begin{split}
        =&\frac{\xi(\ydiagram{2})-\xi(\ydiagram{1,1})}{2\xi(\ydiagram{2})}\smf_{n,m} + \frac{\xi(\ydiagram{2})+\xi(\ydiagram{1,1})-4\xi(\ydiagram{1})}{2\xi(\ydiagram{2})}\id^{\ot (n+m)}_{\hil}\\
        &-\frac{1}{nm\xi(\ydiagram{2})}\left[\repp_{\ydiagram{1}}^{\times (n+m)}(C^{\So(d)}) - \repp_{\ydiagram{1}}^{\times n}(C^{\So(d)}) \ot \id^{\ot m}_{\hil} - \id^{\ot n}_{\hil} \ot \repp_{\ydiagram{1}}^{\times m}(C^{\So(d)})\right]\text{,}
    \end{split}
\end{align}
where we have used the fact that $\xi(\ydiagram{2})\neq 0$. Again, note that  for any choice of multiplicative constant of the $\So(d)$ Casimir eigenvalues we have that the coefficient of the identity is zero. Thus:
\begin{align}
    \smb_{n,m}\coloneqq&\frac{1}{nm}\sum_{j=1}^{n}\sum_{k=n+1}^{n+m} \bb_{jk}=\\
    \begin{split}\label{eq:bnm}
        =&\frac{\xi(\ydiagram{2})-\xi(\ydiagram{1,1})}{2\xi(\ydiagram{2})}\smf_{n,m}\!-\!\frac{1}{nm\xi(\ydiagram{2})}\left[\repp_{\ydiagram{1}}^{\times (n+m)}(C^{\So(d)}) - \repp_{\ydiagram{1}}^{\times n}(C^{\So(d)}) \ot \id^{\ot m}_{\hil} - \id^{\ot n}_{\hil} \ot \repp_{\ydiagram{1}}^{\times m}(C^{\So(d)})\right]\text{,}
    \end{split}
\end{align}
By using the concrete Casimir eigenvalues from Equations \eqref{eq:chi} and \eqref{eq:xi}, we obtain the following easy-to-use formulas:
\begin{align}
    \smf_{n,m}&=\frac{1}{2nm}\left[\rep_{\ydiagram{1}}^{\times (n+m)}(C^{\ualg(d)}) - \rep_{\ydiagram{1}}^{\times n}(C^{\ualg(d)}) \ot \id^{\ot m}_{\hil} - \id^{\ot n}_{\hil} \ot \rep_{\ydiagram{1}}^{\times m}(C^{\ualg(d)}) \right]\text{,}\\
    \smb_{n,m}&=\frac{1}{d}\smf_{n,m}-\frac{1}{4dnm}\left[\repp_{\ydiagram{1}}^{\times (n+m)}(C^{\So(d)}) - \repp_{\ydiagram{1}}^{\times n}(C^{\So(d)}) \ot \id^{\ot m}_{\hil} - \id^{\ot n}_{\hil} \ot \repp_{\ydiagram{1}}^{\times m}(C^{\So(d)})\right]\text{.}
\end{align}

\subsection{The low-dimensional cases}\label{sec:app_cas_lowd}

As mentioned before, the low-dimensional cases ($d<2(n+m)$) need some careful considerations. There are three non-trivial possibilities:
\begin{enumerate}
    \item A given orthogonal irrep that appeared at higher dimensions does not appear in lower dimensions. In this case, the projector for that irrep simply does not appear and does not take part in defining the parameter space of the $(n,m)$-extendible Brauer states. This only happens when $d<n+m$.
    \item A given orthogonal irrep that appeared at higher dimensions still appears, but its labelling Young diagram needs to be modified, as it is a non-standard label. In this case, due to Theorem \ref{th:mod_cas}, no change is needed to the calculations, other than a cosmetic change of relabelling these projections.
    \item For the even-dimensional cases, the representations labelled by Young diagrams with non-zero last rows split up into two when restricted to $\So(d)$ (this also happens when $d=2(n+m)$ but not later on). This only affects the Casimirs and, since the quadratic Casimir of $\So(d)$ does not differentiate between the two irreps, no actual change is needed to the calculations.
\end{enumerate}

These possibilities are highlighted at the treatment of each low-dimensional case for a concrete $(n,m)$-extendibility number.

Now, let us discuss the low-dimensional cases of Equation \eqref{eq:cas_2_so} that require extra care. These are dimensions $d=2,3,4$. As it turns out, the form of the equation and thus $\smb_{n,m}$ does not change (other than cosmetically), but a careful treatment is presented for completeness.

If $d=2$, the diagrams labelling the $\Ort(2)$ irreps can only be at most one row tall and all non-1-dimensional irreps are isomorphic to their associate irrep. The following is true, see Appendix \ref{sec:app_rep}:
\begin{equation}
    \Repp_{\ydiagram{1}}\ot\Repp_{\ydiagram{1}} \cong \Repp_{0^*} \OP \Repp_{0} \OP \Repp_{\ydiagram{2}}\text{,}
\end{equation}
furthermore, on the level of Lie-algebra this gives:
\begin{equation}
    \repp_{\ydiagram{1}}\times \repp_{\ydiagram{1}} \cong \repp_{0} \OP \repp_{0} \OP \repp_{\ydiagram{2}} \OP \repp_{\ydiagram[\cdot]{2}}\text{,}
\end{equation}
where the label $\ydiagram[\cdot]{2}$ denotes the [-2] label.

This means the following modification to Equation \eqref{eq:cas_2_so}, where for $d=2$ there is only one generator, $Y$ and its dual $\tilde{Y}$:
\begin{align}
    \repp_{\ydiagram{1}}\times\repp_{\ydiagram{1}}\left(C^{\So(2)}\right)&=\left[\repp_{\ydiagram{1}}(Y)\ot\id_{\hil} + \id_{\hil}\ot\repp_{\ydiagram{1}}(Y)\right]\left[\repp_{\ydiagram{1}}(\tilde{Y})\ot\id_{\hil} + \id_{\hil}\ot\repp_{\ydiagram{1}}(\tilde{Y})\right] \\
    &=2\xi(\ydiagram{1})\id_{\hil}\ot\id_{\hil} + \repp_{\ydiagram{1}}(Y)\ot\repp_{\ydiagram{1}}(\tilde{Y}) + \repp_{\ydiagram{1}}(\tilde{Y})\ot \repp_{\ydiagram{1}}(Y)\\
    &=\xi(0) \PROJ_{0} + \xi(0) \PROJ_{0} + \xi(\ydiagram{2}) \PROJ_{\ydiagram{2}} + \xi(\ydiagram[\cdot]{2})\PROJ_{\ydiagram[\cdot]{2}}\\
    &=\xi(\ydiagram{2})(\PROJ_{\ydiagram{2}}+\PROJ_{\ydiagram[\cdot]{2}})=\xi(\ydiagram{2})\left(\frac{\id_\hil\ot\id_\hil+\flip}{2}-\bb\right)\text{,}
\end{align}
where we have used the fact that $\PROJ_{\ydiagram{2}}+\PROJ_{\ydiagram[\cdot]{2}}$ is the projection onto the invariant subspace of $\Repp_{\ydiagram{2}}$. This splitting of the original diagram $\ydiagram{2}$ into two is possibility number 3 on the previous list and warrants no change in the calculations, as $\xi(\ydiagram{2})=\xi(\ydiagram[\cdot]{2})$ because of the quadratic nature of the Casimir eigenvalue with respect to the Young diagram's last line. Furthermore, the modification of diagram $\ydiagram{1,1}$ into $0^*$ is possibility number 2 of the previous list, and thus needs no extra care. That is, according to Theorem \ref{th:mod_cas}, $\xi(0)=\xi(\ydiagram{1,1})=0$ in the case of $d=2$, which can easily be checked. Therefore, no actual (non-cosmetic) change is needed to Equation \eqref{eq:cas_2_so}, and to the form of $\smb_{n,m}$ in the $d=2$ case.

If $d=3$, the diagrams labelling the $\Ort(3)$ irreps can at most be one row tall, which changes things thusly:
\begin{equation}
    \Repp_{\ydiagram{1}}\ot\Repp_{\ydiagram{1}} \cong \Repp_{\ydiagram{1}^*} \OP \Repp_{0} \OP \Repp_{\ydiagram{2}}\text{,}
\end{equation}
where, when restricting to Lie algebra representations, the one labelled by $\ydiagram{1}^*$ would simply change to $\ydiagram{1}$. Here, again, this is possibility number 2 on the list and warrants no (non-cosmetic) change in the form of Equation \eqref{eq:cas_2_so} or $\smb_{n,m}$.

If $d=4$, the only change is the splitting of the irrep labelled by $\ydiagram{1,1}$ into the plus and minus version: [1,1] and [1,-1]. This is possibility number 3 on the above list and again does not require a (non-cosmetic) change.

Thus, it is safe to use the above calculated form of $\smb_{n,m}$ in any $d>1$ dimensions.

\setcounter{equation}{0}

\section{Dimensions of $\sym_N$, $\U(d)$ and $\Ort(d)$ irreps}\label{sec:app_dim}

In this appendix, the dimension formulas of $\sym_N$, $\U(d)$ and $\Ort(d)$ irreducible representations are presented. These are needed in the concrete calculations of $(n,m)$-extendible Brauer states as attested by Equation \eqref{eq:commutant}.

\subsection{Dimension formula for $\sym_N$ ireps}

The formula for the dimensions of the irreps of the symmetric group $\sym_N$ is known as the hook length formula. See 4.12.~in \cite{fulton2004representation} or Corollary 9.1.6 in \cite{goodman2009symmetry}:

\begin{theorem}[Hook length formula]
    Let $\lambda$ be a Young diagram of $N$ boxes labelling an irrep of $\sym_N$. Then:
    \begin{equation}
        \dim \left(\Repppp_{\lambda}\right) = \frac{N!}{\prod_{(i,j)\in\lambda} h_{\lambda}(i,j)}\text{,}
    \end{equation}
    where $(i,j)$ denote, respectively, the $i$-th row and $j$-th column of $\lambda$ and $h_{\lambda}(i,j)$ denotes the hook length of the $(i,j)$ box of $\lambda$.
\end{theorem}

\subsection{Dimension formula for $\U(d)$ irreps}

The formula for calculating the dimension of a $\U(d)$ tensor irrep labelled by a Young diagram $\lambda$ is well-known, see Section 3 in \cite{king1971dimension}:

\begin{theorem}
    Let $\Rep_{\lambda}$ be a tensor irrep of $\U(d)$ labelled by the Young diagram $\lambda$. The dimension of this irrep is given by the following formula:
    \begin{equation}
        \dim(\Rep_{\lambda})=\prod_{(i,j)\in\lambda} \frac{d-i+j}{h_{\lambda}(i,j)}\text{,}
    \end{equation}
    where $(i,j)$ denote, respectively, the $i$-th row and $j$-th column of $\lambda$ and $h_{\lambda}(i,j)$ denotes the hook length of the $(i,j)$ box of $\lambda$.
\end{theorem}

\subsection{Dimension formula for $\Ort(d)$ irreps}

The formula to calculate the dimension of an $\Ort(d)$ irrep labelled by a Young diagram $\lambda$ is a bit more cumbersome to work with. See Equation (4.12) in \cite{king1971dimension}\footnote{Careful, King has made a typo! Using his notation, the first nominator should be: $(n+\sigma_i+\sigma_{i+j-1}-2i-j+1)$. See Equation (4.11) in \cite{king1971dimension} to confirm this.}:

\begin{theorem}
    Let $\Repp_{\lambda}$ be a tensor irrep of $\Ort(d)$ labelled by the Young diagram $\lambda$ with $p$ number of rows. The dimension of this irrep is given by the following formula:
    \begin{equation}
    \dim(\Repp_{\lambda})=\prod_{(i,j)\in\lambda} \frac{d+\lambda_{i}+\lambda_{i+j-1}-2i-j+1}{h_{\lambda}(i,j)} \prod_{(i,l\geq 1)\in(\tau-\lambda)} \frac{d-2i-l+1+\lambda_{i+\lambda_i+l-1}}{d-2i-l+1}\text{,}
    \end{equation}
    where $(i,j)$ denote respectively the $i$-th row and $j$-th column of $\lambda$, $h_{\lambda}(i,j)$ denotes the hook length of the $(i,j)$ box of $\lambda$, $\lambda_i$ denotes the number of boxes in the $i$-th row of $\lambda$, and $(\tau-\lambda)$ denotes the formal difference diagram (that is: $[\tau_1-\lambda_1,\tau_2-\lambda_2,\ldots]$) where $\tau$ is the triangular diagram having $p$ boxes in the first row: $\tau=[p,p-1,p-2,\ldots,1]$.
\end{theorem}

Note that the formal difference diagram $(\tau-\lambda)$ might have negative or zero boxes in a given row. From the point of view of the formula, these are unimportant as $l>0$ is a condition. Note also that the first column of the formal difference diagram $(\tau-\lambda)$ is considered to be the first one with boxes in it.\footnote{That is: the first column is not the disappearing original first column of either $\tau$ or $\lambda$.} However, the first row is the same as the first row of $\tau$ or $\lambda$, even if it has zero, or negative boxes in it.

\setcounter{equation}{0}

\section{Quantitative results}\label{sec:app_res}

This appendix contains the quantitative results for the cases of $(1,2)$-, $(1,3)$- and $(2,2)$-extendible Brauer states. The derivation of the $(1,2)$-extendibility case is presented in detail, whereas for the other two cases only the results are presented.

\subsection{The derivation and results for (1,2)-extendibility} \label{sec:app_12concrete}

Given the requirements imposed on an extending state $\hat{\rho}_{(1,2)}\in\states{\hil\ot\hil^{\ot 2}}$, the irrep decomposition of the representation it has to be invariant to is the following for $d>5$:
\begin{equation}
\begin{split}
    \left(\hat{\Repppp}^{\mathrm{L}}_1\hat{\Repp}_1 \right)\ot \left( \hat{\Repppp}^{\mathrm{R}}_2 \hat{\Repp}_2\right) \cong \; &\Repppp^{\mathrm{L}}_{\ydiagram{1}} \ot \Repppp^{\mathrm{R}}_{\ydiagram{1,1}} \ot \left(\Repp_{\ydiagram{1}} \OP \Repp_{\ydiagram{1,1,1}} \OP \Repp_{\ydiagram{2,1}} \right) \OP \\
    & \Repppp^{\mathrm{L}}_{\ydiagram{1}} \ot \Repppp^{\mathrm{R}}_{\ydiagram{2}}\ot\left({\color{red}\Repp_{\ydiagram{1}}} \OP {\color{red}\Repp_{\ydiagram{1}}} \OP \Repp_{\ydiagram{2,1}}\OP \Repp_{\ydiagram{3}}\right) \text{,}
\end{split}
\end{equation}
where the only irrep with multiplicity greater than one is highlighted in red. The case of smaller dimensions will be treated at the end of this section.

According to Equation \eqref{eq:commutant}, the form of an extending state $\hat{\rho}_{(1,2)}$ is the following:
\begin{equation}
\begin{split}\label{eq:12form}
    \hat{\rho}_{(1,2)}=\; &\mu_{\ydiagram{1},\ydiagram{1,1},\ydiagram{1}}\frac{\PROJ_{\ydiagram{1},\ydiagram{1,1},\ydiagram{1}}}{\dim(\ydiagram{1},\ydiagram{1,1},\ydiagram{1})}+\mu_{\ydiagram{1},\ydiagram{1,1},\ydiagram{1,1,1}}\frac{\PROJ_{\ydiagram{1},\ydiagram{1,1},\ydiagram{1,1,1}}}{\dim(\ydiagram{1},\ydiagram{1,1},\ydiagram{1,1,1})}+\mu_{\ydiagram{1},\ydiagram{1,1},\ydiagram{2,1}}\frac{\PROJ_{\ydiagram{1},\ydiagram{1,1},\ydiagram{2,1}}}{\dim(\ydiagram{1},\ydiagram{1,1},\ydiagram{2,1})}+\\
    &\mu_{\ydiagram{1},\ydiagram{2},\ydiagram{1}}\frac{\color{red}\mathbf{M}_{\ydiagram{1},\ydiagram{2},\ydiagram{1}}}{\dim(\ydiagram{1},\ydiagram{2},\ydiagram{1})} +\mu_{\ydiagram{1},\ydiagram{2},\ydiagram{2,1}}\frac{\PROJ_{\ydiagram{1},\ydiagram{2},\ydiagram{2,1}}}{\dim(\ydiagram{1},\ydiagram{2},\ydiagram{2,1})}+\mu_{\ydiagram{1},\ydiagram{2},\ydiagram{3}}\frac{\PROJ_{\ydiagram{1},\ydiagram{2},\ydiagram{3}}}{\dim(\ydiagram{1},\ydiagram{2},\ydiagram{3})}\text{,}
\end{split}
\end{equation}
where the three Young diagrams in each index together correspond to the $\Lambda$ indices in Equation \eqref{eq:commutant}, that is, the indices labelling the irrep of $\sym_1\times\sym_2\times\Ort(d)$ in this order. Furthermore, $\mu_\Lambda\geq0$ and they satisfy $\sum_{\Lambda}\mu_\Lambda=1$. The projections $\PROJ_\Lambda$ denote the projections onto the $\kil_\Lambda\ot\lil_\Lambda$ subspaces whenever $\dim(\kil_\Lambda)=1$. The red summand is related to the multiplicity-two irrep $\Repppp^{\mathrm{L}}_{\ydiagram{1}}\ot\Repppp^{\mathrm{R}}_{\ydiagram{2}}\ot{\color{red}\Repp_{\ydiagram{1}}}$:
\begin{equation}
    {\color{red}\mathbf{M}_{\ydiagram{1},\ydiagram{2},\ydiagram{1}}}\cong \rho_{\ydiagram{1},\ydiagram{2},\ydiagram{1}}\ot \id_{\lil_{\ydiagram{1},\ydiagram{2},\ydiagram{1}}}\text{,}
\end{equation}
where $\rho_{\ydiagram{1},\ydiagram{2},\ydiagram{1}}$ is any quantum state on the two dimensional multiplicity space $\kil_{\ydiagram{1},\ydiagram{2},\ydiagram{1}}$.

The dimensions of the $\sym_1$, $\sym_2$ and $\Ort(d)$ irreps multiply together and can be calculated using Appendix \ref{sec:app_dim}:
\begin{gather}
    \dim(\Repppp_{\ydiagram{1}})=1\text{,} \qquad \dim(\Repppp_{\ydiagram{2}})=1\text{,} \qquad \dim(\Repppp_{\ydiagram{1,1}})=1\text{,}\\
    \dim(\Repp_{\ydiagram{1}})=d\text{,} \qquad \dim(\Repp_{\ydiagram{1,1,1}})=\frac{d(d-1)(d-2)}{6}\text{,}\\ \dim(\Repp_{\ydiagram{2,1}})=\frac{(d+2)d(d-2)}{3}\text{,} \qquad \dim(\Repp_{\ydiagram{3}})=\frac{(d+4)d(d-1)}{6}\text{.}
\end{gather}
In this case only the dimensions of $\Ort(d)$ irreps are relevant.

Thus, Equation \eqref{eq:12form} tells us that the parameter space will contain 5 points related to the projections $\PROJ_\Lambda$ and an ellipse related to the qubit state $\rho_{\ydiagram{1},\ydiagram{2},\ydiagram{1}}$. Given the properties of the $\mu_\Lambda$ coefficients, the convex hull of these points and the ellipse will describe the parameter space of the $(1,2)$-extendible Brauer states.

The case of the projectors related to multiplicity-one irreps is trivial since in this case every represented Casimir operator can be diagonalised simultaneously. To calculate their contribution to the parameter space, we use Equations \eqref{eq:smf} and \eqref{eq:smb}. We collect the relevant Casimir eigenvalues and results in Table \ref{tab:projs}. Note that the trivial $\rep_{\ydiagram{1}}(C^{\ualg(d)})\ot\id_\hil^{\ot 2}$ and $\repp_{\ydiagram{1}}(C^{\So(d)})\ot\id_{\hil}^{\ot 2}$ operators have been omitted from the table for brevity: their contribution is always $\chi(\ydiagram{1})$ and $\xi(\ydiagram{1})$, respectively.

\begin{center}
\begin{table}[h!]
\def\arraystretch{1.5} 
\begin{tabular}{|c||c|c|c|c||c|c|}
    \hline
    $\tr/\dim$ & $\rep_{\ydiagram{1}}^{\times 3}\left(C^{\ualg(d)}\right)$ & $\id_\hil \ot \rep_{\ydiagram{1}}^{\times 2}\left(C^{\ualg(d)}\right)$ & $\repp_{\ydiagram{1}}^{\times 3}\left(C^{\So(d)}\right)$ & $\id_\hil\ot\repp_{\ydiagram{1}}^{\times 2}\left(C^{\So(d)}\right)$ & $\smf_{1,2}$ & $\smb_{1,2}$\\
    \hline \hline
    $\PROJ_{\ydiagram{1},\ydiagram{1,1},\ydiagram{1}}$ & $\chi(\ydiagram{2,1})$ & $\chi\left(\ydiagram{1,1}\right)$ & $\xi(\ydiagram{1})$ & $\xi(\ydiagram{1,1})$ & $\frac{1}{2}$ & $\frac{d-1}{2d}$\\
    \hline
    $\PROJ_{\ydiagram{1},\ydiagram{1,1},\ydiagram{1,1,1}}$ & $\chi(\ydiagram{1,1,1})$ & $\chi\left(\ydiagram{1,1}\right)$ & $\xi(\ydiagram{1,1,1})$ & $\xi(\ydiagram{1,1})$ & -1 & 0\\
    \hline
    $\PROJ_{\ydiagram{1},\ydiagram{1,1},\ydiagram{2,1}}$ & $\chi\left(\ydiagram{2,1}\right )$ & $\chi\left(\ydiagram{1,1}\right)$ & $\xi(\ydiagram{2,1})$ & $\xi(\ydiagram{1,1})$ & $\frac{1}{2}$ & 0\\
    \hline
    $\PROJ_{\ydiagram{1},\ydiagram{2},\ydiagram{2,1}}$ & $\chi\left(\ydiagram{2,1}\right)$ & $\chi\left(\ydiagram{2}\right)$ & $\xi(\ydiagram{2,1})$ & $\xi(\ydiagram{2})$ & $-\frac{1}{2}$ & 0\\
    \hline
    $\PROJ_{\ydiagram{1},\ydiagram{2},\ydiagram{3}}$ & $\chi\left(\ydiagram{3}\right )$ & $\chi\left(\ydiagram{2}\right)$ & $\xi(\ydiagram{3})$ & $\xi(\ydiagram{2})$ & 1 & 0\\
    \hline
\end{tabular}
\caption{\small The contributions of the projectors related to multiplicity-one irreps in the $(1,2)$-extendibility case. The table contains the values of the traces of the products of the operators from the first line and first column, divided by the relevant dimension. The trivial $\rep_{\ydiagram{1}}(C^{\ualg(d)})\ot\id_\hil^{\ot 2}$ and $\repp_{\ydiagram{1}}(C^{\So(d)})\ot\id_{\hil}^{\ot 2}$ operators have been omitted for brevity.} \label{tab:projs}
\end{table}
\end{center}

The case of the multiplicity-two irrep $\Repppp^{\mathrm{L}}_{\ydiagram{1}}\ot\Repppp^{\mathrm{R}}_{\ydiagram{2}}\ot{\color{red}\Repp_{\ydiagram{1}}}$ is more involved, as not every Casimir operator can be diagonalised simultaneously in the different representations. What we can guarantee is that the Casimir operators represented as $\rep_{\ydiagram{1}}(C^{\ualg(d)}) \ot \id^{\ot 2}_\hil$, $\id_\hil \ot \rep_{\ydiagram{1}}^{\times 2}(C^{\ualg(d)})$, and $\repp_{\ydiagram{1}}^{\times 3}(C^{\So(d)})$ will be diagonal. Furthermore, given that $n=1$, the Casimir operator represented as $\repp_{\ydiagram{1}}(C^{\So(d)})\ot\id_\hil^{\ot 2}$ is also diagonal. For all other cases, we introduce Figure \ref{fig:casmult2} below, which mirrors Figure \ref{fig:casimirs} and illustrates the two different diagonalisations for the multiplicity-two irrep $\Repppp^{\mathrm{L}}_{\ydiagram{1}}\ot\Repppp^{\mathrm{R}}_{\ydiagram{2}}\ot{\color{red}\Repp_{\ydiagram{1}}}$. Note that the problematic part is the middle section, pertaining to the Casimirs in the representations $\rep_{\ydiagram{1}}^{\times 3}$ on the left and the pair $\repp_{\ydiagram{1}}\ot\id_\hil^{\ot 2}$ and $\id_{\hil}\ot\repp_{\ydiagram{1}}^{\times 2}$ on the right: In general these three matrices are not simultaneously diagonal, and thus appear diagonal in bases $a$ and $b$, respectively.

\begin{figure}[H]
    \centering
	\begin{tikzpicture}[scale=1]

    \newcommand{\Widtha}{3.75}
    \newcommand{\Widthb}{5}
    \newcommand{\Heighta}{-2.5}
    \newcommand{\Heightb}{1}
    \newcommand{\smalla}{0.25}
    \newcommand{\smallb}{0.75}

    \draw (\Widtha+\Widthb/2,\smallb) -- (\Widtha+\Widthb/2,2*\Heighta-\smallb);

    \node at (0,0) {$\begin{pmatrix} \chi(\ydiagram{1}) & 0 \\ 0 & \chi(\ydiagram{1}) \end{pmatrix}$};
    \node at (\Widtha,0) {$\begin{pmatrix} \chi(\ydiagram{2}) & 0 \\ 0 & \chi(\ydiagram{2}) \end{pmatrix}$};

    \node at (\Widtha+\Widthb,0) {$\begin{pmatrix} \chi(\ydiagram{1}) & 0 \\ 0 & \chi(\ydiagram{1}) \end{pmatrix}$};
    \node at (2*\Widtha+\Widthb,0) {$\begin{pmatrix} \chi(\ydiagram{2}) & 0 \\ 0 & \chi(\ydiagram{2}) \end{pmatrix}$};

    \draw[double distance=3pt, arrows = {-Latex[length=0pt 2.5 0]}] (\smalla,-\smallb)--(\Widtha/2-\smalla,\Heighta+\smallb);
    \draw[double distance=3pt, arrows = {-Latex[length=0pt 2.5 0]}] (\Widtha-\smalla,-\smallb)--(\Widtha/2+\smalla,\Heighta+\smallb);
    \node at (\Widtha/2,\Heighta+2*\smallb) {$\bigotimes$};

    \draw[double distance=3pt, arrows = {-Latex[length=0pt 2.5 0]}] (\Widtha+\Widthb,-\smallb)--(\Widtha+\Widthb,\Heighta+\smallb);
    \draw[double distance=3pt, arrows = {-Latex[length=0pt 2.5 0]}] (\Widtha+\Widtha+\Widthb,-\smallb)--(\Widtha+\Widtha+\Widthb,\Heighta+\smallb);
    \node at (\Widtha/2+\Widthb + \Widtha, \Heighta/2+\smalla/3) {\small $\ualg(d) \to \So(d)$};

    \node at (\Widtha/2,\Heighta) {$\begin{pmatrix} \chi(\ydiagram{2,1}) & 0 \\ 0 & \chi(\ydiagram{3}) \end{pmatrix}_a$};

    \node at (\Widtha+\Widthb,\Heighta) {$\begin{pmatrix} \xi(\ydiagram{1}) & 0 \\ 0 & \xi(\ydiagram{1}) \end{pmatrix}_b$};
    \node at (2*\Widtha+\Widthb,\Heighta) {$\begin{pmatrix} \xi(0) & 0 \\ 0 & \xi(\ydiagram{2}) \end{pmatrix}_b$};

    \draw[double distance=3pt, arrows = {-Latex[length=0pt 2.5 0]}] (\Widtha/2,\Heighta-\smallb)--(\Widtha/2,2*\Heighta+\smallb);
    \node at (\Widtha/2+5*\smalla, \Heighta+\Heighta/2+\smalla/3) {\small $\ualg(d) \to \So(d)$};

    \draw[double distance=3pt, arrows = {-Latex[length=0pt 2.5 0]}] (\Widtha+\Widthb+\smalla,\Heighta-\smallb)--(\Widtha+\Widthb+\Widtha/2-\smalla,2*\Heighta+\smallb);
    \draw[double distance=3pt, arrows = {-Latex[length=0pt 2.5 0]}] (\Widtha+\Widthb+\Widtha-\smalla,\Heighta-\smallb)--(\Widtha+\Widthb+\Widtha/2+\smalla,2*\Heighta+\smallb);
    \node at (\Widtha+\Widthb+\Widtha/2,2*\Heighta+2*\smallb) {$\bigotimes$};

    \node at (\Widtha/2,2*\Heighta) {$\begin{pmatrix} \xi(\ydiagram{1}) & 0 \\ 0 & \xi(\ydiagram{1}) \end{pmatrix}$};

    \node at (\Widtha+\Widthb+\Widtha/2,2*\Heighta) {$\begin{pmatrix} \xi(\ydiagram{1}) & 0 \\ 0 & \xi(\ydiagram{1}) \end{pmatrix}$};

    \node at (0,2.5*\Heighta) {};
    
	\end{tikzpicture}
    \caption{\small Illustration of the two types of diagonalisation processes of the Casimir operators in different representations, for the multiplicity-two irrep $ \Repppp^{\mathrm{L}}_{\ydiagram{1}} \ot \Repppp^{\mathrm{R}}_{\ydiagram{2}} \ot {\color{red}\Repp_{\ydiagram{1}}}$ in the $(1,2)$-extendibility case. The first line is the same for each case and it describes the diagonal form of the Casimir operator in, respectively, the $\rep_{\ydiagram{1}}\ot\id_\hil^{\ot 2}$ and the $\id_\hil\ot\rep_{\ydiagram{1}}^{\times 2}$ representations. The second line is the crucial line containing the Casimir in representations $\rep_{\ydiagram{1}}^{\times 3}$ on the left and in representations $\repp_{\ydiagram{1}}\ot\id_\hil^{\ot 2}$ and $\id_{\hil}\ot\repp_{\ydiagram{1}}^{\times 2}$ on the right. The left matrix is in basis $a$ and is generally not diagonal in basis $b$ in which the right two matrices are diagonal. The last line is, again, the same for both cases and describes the diagonal form of the Casimir operator in the $\repp_{\ydiagram{1}}^{\times 3}$ representation. The $\ot \id_{\lil_{\ydiagram{1},\ydiagram{2},\ydiagram{1}}}$ factors are omitted everywhere for brevity.}
    \label{fig:casmult2}
\end{figure}

\vspace{-0.5cm}

Let us fix the basis $a$. We can write the general form of the commuting summand related to the multiplicity-two irrep $\Repppp^{\mathrm{L}}_{\ydiagram{1}}\ot\Repppp^{\mathrm{R}}_{\ydiagram{2}}\ot{\color{red}\Repp_{\ydiagram{1}}}$ as:
\begin{equation}
    \frac{\color{red}\mathbf{M}_{\ydiagram{1},\ydiagram{2},\ydiagram{1}}}{\dim(\ydiagram{1},\ydiagram{2},\ydiagram{1})}\cong \begin{pmatrix} \sin^2(\delta) & y\sin(2\delta)e^{i\epsilon}\\ y\sin(2\delta)e^{-i\epsilon} & \cos^2(\delta) \end{pmatrix}_a\ot \frac{\id_{\lil_{\ydiagram{1},\ydiagram{2},\ydiagram{1}}}}{\dim(\ydiagram{1},\ydiagram{2},\ydiagram{1})}\text{,}
\end{equation}
where $y\in[0,1/2]$, $\delta\in[0,\pi/2)$ and $\epsilon\in[0,2\pi)$. The matrix on the right-hand side is nothing else but a general parametrisation of qubit states. These parameters are free in the same sense as the $\mu_\Lambda$ parameters.

As stated in Equation \eqref{eq:uni_intertwiner}, the intertwiner between bases $a$ and $b$ is unitary. That is, we have a generic $2\times 2$ unitary matrix
\begin{equation}
    V\coloneqq e^{i\phi}\begin{pmatrix}
    e^{i\alpha}\cos(\theta) & e^{i\beta}\sin(\theta) \\ -e^{-i\beta}\sin(\theta) & e^{-i\alpha}\cos(\theta)
\end{pmatrix}\text{,}
\end{equation}
where $\phi,\alpha,\beta,\theta$ are unknown fixed variables that, in general, depend on the specific representations in question and also on the dimension.

Thus, to get the form of $\id_\hil\ot\repp_{\ydiagram{1}}^{\times 2}(C^{\So(d)})$ in the $a$ basis, we have:
\begin{equation}
    V^\dagger \begin{pmatrix} \xi(0) & 0 \\ 0 & \xi(\ydiagram{2}) \end{pmatrix}_b V = \begin{pmatrix}
        \cos^2(\theta)\xi(0)+\sin^2(\theta)\xi(\ydiagram{2}) & e^{i(\beta-\alpha)}\cos(\theta)\sin(\theta)(\xi(0)-\xi(\ydiagram{2})) \\ e^{-i(\beta-\alpha)}\cos(\theta)\sin(\theta)(\xi(0)-\xi(\ydiagram{2})) & \sin^2(\theta)\xi(0)+\cos^2(\theta)\xi(\ydiagram{2}) \end{pmatrix}_a\text{,}
\end{equation}
where, again, the $\ot \id_{\lil_{\ydiagram{1},\ydiagram{2},\ydiagram{1}}}$ factor has been omitted on both sides for brevity. Note that the only remaining relevant variables are $\beta-\alpha$ and $\theta$. Through a detailed calculation we find that whatever $\beta-\alpha$ is, it can be merged into the $\epsilon$ parameter of the qubit state. After substituting the Casimir eigenvalues from Appendix \ref{sec:app_cas}, this leads to the following,
\begin{align}
    \frac{\tr({\color{red}\mathbf{M_{\ydiagram{1},\ydiagram{2},\ydiagram{1}}}}\smf_{1,2})}{\dim(\ydiagram{1},\ydiagram{2},\ydiagram{1})}&=\frac{2-3\sin^2(\delta)}{2}\text{,}\\
    \frac{\tr({\color{red}\mathbf{M_{\ydiagram{1},\ydiagram{2},\ydiagram{1}}}}\smb_{1,2})}{\dim(\ydiagram{1},\ydiagram{2},\ydiagram{1})}&=\frac{d+1-2dy\sin(2\theta)\sin(2\delta)\cos(\epsilon)+\cos(2\delta)(3+d\cos(2\theta))}{4d}\text{.}
\end{align}

The only unknown remaining is the $\theta$ parameter of the unitary intertwiner $V$. To calculate it, we use Proposition 9.3.14 in \cite{goodman2009symmetry} which states that given the representation $\hat{\Repppp}_N\hat{\Rep}_{N}$ of the group $\sym_N\times \U(d)$, the canonical projection onto the isotypic component labelled by $\lambda$ is the sum of the appropriate Young symmetrisers:
\begin{equation}\label{eq:yng_sym1}
    \hat{\PROJ}_{\ydiagram{2,1}}=\frac{1}{3}(2\id_\hil^{\ot 3} - \flip_{12}\flip_{13} - \flip_{13}\flip_{12})\text{,}
\end{equation}
where $\hat{\PROJ}_{\ydiagram{2,1}}$ is the projection onto the $\sym_3\times \U(d)$-isotypic subspace $\hat{\kil}_{\ydiagram{2,1}}\ot\hat{\lil}_{\ydiagram{2,1}}$, and $\flip_{ij}$ is the operator that swaps the $i$-th and $j$-th Hilbert space while leaving the third intact.

On the other hand, from the branching rules described in Appendix \ref{sec:app_rep}, we also have:
\begin{equation}\label{eq:yng_sym2}
    \hat{\PROJ}_{\ydiagram{2,1}}=\PROJ_{\ydiagram{1},\ydiagram{1,1},\ydiagram{1}}+\PROJ_{\ydiagram{1},\ydiagram{1,1},\ydiagram{2,1}}+\PROJ_{\ydiagram{1},\ydiagram{2},\ydiagram{1}}^{(\ydiagram{2,1})}+\PROJ_{\ydiagram{1},\ydiagram{2},\ydiagram{2,1}}\text{,}
\end{equation}
where $\PROJ_{\Lambda}$ is the projection onto the multiplicity-one $\sym_1\times\sym_2\times\Ort(d)$-isotypic component labelled by $\Lambda$, and we have defined the restriction of $\PROJ_{\ydiagram{1},\ydiagram{2},\ydiagram{1}}$ to the $\sym_3\times \U(d)$-isotypic subspace $\hat{\kil}_{\ydiagram{2,1}}\ot\hat{\lil}_{\ydiagram{2,1}}$:
\begin{equation}
    \PROJ_{\ydiagram{1},\ydiagram{2},\ydiagram{1}}^{(\ydiagram{2,1})} \coloneqq \hat{\PROJ}_{\ydiagram{2,1}} \PROJ_{\ydiagram{1},\ydiagram{2},\ydiagram{1}}\hat{\PROJ}_{\ydiagram{2,1}}\text{.}
\end{equation}
By construction we have that:
\begin{equation}
    \PROJ_{\ydiagram{1},\ydiagram{2},\ydiagram{1}}\cong \id_{\kil_{\ydiagram{1},\ydiagram{2},\ydiagram{1}}}\ot \id_{\lil_{\ydiagram{1},\ydiagram{2},\ydiagram{1}}}\text{,}
\end{equation}
and, furthermore, we have:
\begin{equation}
    \PROJ_{\ydiagram{1},\ydiagram{2},\ydiagram{1}}^{(\ydiagram{2,1})}\cong \begin{pmatrix} 1 & 0\\ 0 & 0 \end{pmatrix}_a\ot \id_{\lil_{\ydiagram{1},\ydiagram{2},\ydiagram{1}}}\text{,}
\end{equation}
that is, when $\sin^{2}(\delta)=1$.

Using Equations \eqref{eq:yng_sym1} and \eqref{eq:yng_sym2}, the Table \ref{tab:projs}, and the concrete form of $\bb$, we can calculate
\begin{equation}
    \frac{\tr\left(\PROJ_{\ydiagram{1},\ydiagram{2},\ydiagram{1}}^{(\ydiagram{2,1})}\bb_{12}\right)}{\dim(\ydiagram{1},\ydiagram{2},\ydiagram{1})}=\frac{d-1}{6d}\text{.}
\end{equation}

Leading to the result that determines $\theta$:
\begin{equation}
    \cos(2\theta)=\frac{d-4}{3d}\text{,} \qquad \sin(2\theta)=\pm\frac{2}{3d}\sqrt{2(d+2)(d-1)}\text{.}
\end{equation}

Again, the $\pm1$ phase can be merged into the $\epsilon$ parameter, which, after trigonometric identities, leads to:
\begin{align}
    \frac{\tr({\color{red}\mathbf{M_{\ydiagram{1},\ydiagram{2},\ydiagram{1}}}}\smf_{1,2})}{\dim(\ydiagram{1},\ydiagram{2},\ydiagram{1})}&=\frac{1}{4}+\frac{3}{4}\cos(2\delta)\text{,}\\
    \frac{\tr({\color{red}\mathbf{M_{\ydiagram{1},\ydiagram{2},\ydiagram{1}}}}\smb_{1,2})}{\dim(\ydiagram{1},\ydiagram{2},\ydiagram{1})}&=\frac{d+1}{4d}+\frac{d+5}{12d}\cos(2\delta)+\frac{\sqrt{2(d+2)(d-1)}}{3d}y\cos(\epsilon)\sin(2\delta)\text{.}
\end{align}

To maximize or minimize the contribution to the $\tr(\rho\bb)$-parameter, we choose $y=1/2$ and $\cos(\epsilon)=\pm1$. We can merge the $\cos(\epsilon)$ contribution into the contribution of $\sin(2\delta)$ by extending the $\delta$ values to $[0,\pi)$. Given that $2\delta$ appears everywhere, we rescale such that the new $\delta$ parameter now takes values in $[0,2\pi)$. 

Altogether we get the following parametric equation of an ellipse, centred at $(1/4,(d+1)/4d)$:
\begin{align}
    \frac{\tr({\color{red}\mathbf{M_{\ydiagram{1},\ydiagram{2},\ydiagram{1}}}}\smf_{1,2})}{\dim(\ydiagram{1},\ydiagram{2},\ydiagram{1})}&=\frac{1}{4}+\frac{3}{4}\cos(\delta)\text{,}\\
    \frac{\tr({\color{red}\mathbf{M_{\ydiagram{1},\ydiagram{2},\ydiagram{1}}}}\smb_{1,2})}{\dim(\ydiagram{1},\ydiagram{2},\ydiagram{1})}&=\frac{d+1}{4d}+\frac{d+5}{12d}\cos(\delta)+\frac{\sqrt{2(d+2)(d-1)}}{6d}\sin(\delta)\text{,}
\end{align}
where $\delta\in[0,2\pi)$.

The cases of dimensions $d\leq 5$, can be discussed using the ideas presented in Appendix \ref{sec:app_cas_lowd}. The only non-cosmetic modification happens at $d=2$, in which the projectors $\PROJ_{\ydiagram{1},\ydiagram{1,1},\ydiagram{1,1,1}}$, $\PROJ_{\ydiagram{1},\ydiagram{1,1},\ydiagram{2,1}}$ and $\PROJ_{\ydiagram{1},\ydiagram{2},\ydiagram{2,1}}$ do not appear.

Figure \ref{fig:12detailed_d2_d6} illustrates the convex hull of the contributions of the projectors and the ellipse for $d=2,6$.

\vspace{-0.2cm}

\begin{figure}[H]
    \centering
    \includegraphics[width=0.65\linewidth]{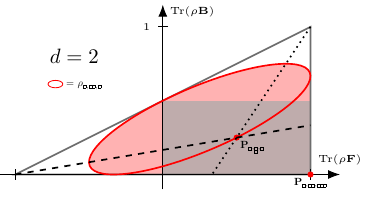}
\end{figure}

\vspace{-1cm}

\begin{figure}[H]
    \centering
    \includegraphics[width=0.65\linewidth]{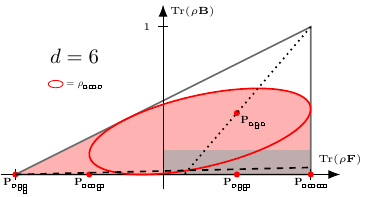}
    \caption{\small The set of Brauer states at $d=2,6$ in the $\tr(\rho\flip) - \tr(\rho\bb)$ parametrisation appearing as the grey triangle. The subset of $(1,2)$-extendible Brauer states appears in light red in the background. The dark red points and ellipse denote the contributions of the appropriate projectors and qubit state, respectively. The set of Werner states is illustrated as the dashed line, while the set of isotropic states as the dotted line. The set of separable Brauer states constitute the grey rectangle.}
    \label{fig:12detailed_d2_d6}
\end{figure}

\subsection{The results for (1,3)-extendibility}\label{sec:app_res2}

Given the requirements imposed on an extending state $\hat{\rho}_{(1,3)}\in\states{\hil\ot\hil^{\ot 3}}$ the irrep decomposition of the representation it has to be invariant to is the following for $d>7$:
\begin{equation}
\begin{split}
    \left(\hat{\Repppp}^{\mathrm{L}}_1\hat{\Repp}_1 \right)\ot \left( \hat{\Repppp}^{\mathrm{R}}_3 \hat{\Repp}_3\right) \cong \; & \Repppp^{\mathrm{L}}_{\ydiagram{1}} \ot \Repppp^{\mathrm{R}}_{\ydiagram{1,1,1}}\ot \left(\Repp_{\ydiagram{1,1}} \OP \Repp_{\ydiagram{1,1,1,1}} \OP \Repp_{\ydiagram{2,1,1}} \right) \OP \\
    & \Repppp^{\mathrm{L}}_{\ydiagram{1}} \ot \Repppp^{\mathrm{R}}_{\ydiagram{2,1}}\ot \left(\Repp_{0} \OP {\color{orange}\Repp_{\ydiagram{1,1}}}\OP {\color{orange}\Repp_{\ydiagram{1,1}}}\OP {\color{brass}\Repp_{\ydiagram{2}}} \OP {\color{brass}\Repp_{\ydiagram{2}}} \OP \Repp_{\ydiagram{2,1,1}} \OP \Repp_{\ydiagram{2,2}} \OP \Repp_{\ydiagram{3,1}} \right) \OP \\
    & \Repppp^{\mathrm{L}}_{\ydiagram{1}} \ot \Repppp^{\mathrm{R}}_{\ydiagram{3}}\ot \left(\Repp_{0} \OP \Repp_{\ydiagram{1,1}} \OP {\color{brown}\Repp_{\ydiagram{2}}}\OP {\color{brown}\Repp_{\ydiagram{2}}} \OP \Repp_{\ydiagram{3,1}} \OP \Repp_{\ydiagram{4}} \right)\text{,}
\end{split}
\end{equation}
Note that the coloured representations have multiplicity greater than one.

The contributions of the projectors related to multiplicity-one irreps is gathered in Table \ref{tab:projs13}.
\begin{center}
\begin{table}[h!]
\def\arraystretch{1.5} 
\begin{tabular}{|c||c|c|c|c||c|c|||c|}
\hline
$\tr/\dim$ & $\rep_{\ydiagram{1}}^{\times 4}\left(C^{\ualg(d)}\right)$ & $\id_\hil \ot \rep_{\ydiagram{1}}^{\times 3}\left(C^{\ualg(d)}\right)$ & $\repp_{\ydiagram{1}}^{\times 4}\left(C^{\So(d)}\right)$ & $\id_\hil \ot \repp_{\ydiagram{1}}^{\times 3}\left(C^{\So(d)}\right)$ & $\smf_{1,3}$ & $\smb_{1,3}$ & Appearance\\
\hline \hline
$\PROJ_{\ydiagram{1},\ydiagram{1,1,1},\ydiagram{1,1}}$ & $\chi(\ydiagram{2,1,1})$ & $\chi(\ydiagram{1,1,1})$ & $\xi(\ydiagram{1,1})$ & $\xi(\ydiagram{1,1,1})$ & $\frac{1}{3}$ & $\frac{d-2}{3d}$ & $d\geq 3$ \\
\hline
$\PROJ_{\ydiagram{1},\ydiagram{1,1,1},\ydiagram{1,1,1,1}}$ & $\chi\left(\ydiagram{1,1,1,1}\right)$ & $\chi(\ydiagram{1,1,1})$ & $\xi\left(\ydiagram{1,1,1,1}\right)$ & $\xi(\ydiagram{1,1,1})$ & $-1$ & 0 & $d\geq 4$ \\
\hline
$\PROJ_{\ydiagram{1},\ydiagram{1,1,1},\ydiagram{2,1,1}}$ & $\chi(\ydiagram{2,1,1})$ & $\chi(\ydiagram{1,1,1})$ & $\xi(\ydiagram{2,1,1})$ & $\xi(\ydiagram{1,1,1})$ & $\frac{1}{3}$ & 0 & $d\geq 4$ \\
\hline
$\PROJ_{\ydiagram{1},\ydiagram{2,1},0}$ & $\chi(\ydiagram{2,2})$ & $\chi(\ydiagram{2,1})$ & $\xi(0)$ & $\xi(\ydiagram{1})$ & 0 & $\frac{d-1}{3d}$ & $d\geq 2$ \\
\hline
$\PROJ_{\ydiagram{1},\ydiagram{2,1},\ydiagram{2,1,1}}$ & $\chi(\ydiagram{2,1,1})$ & $\chi(\ydiagram{2,1})$ & $\xi(\ydiagram{2,1,1})$ & $\xi(\ydiagram{2,1})$ & $-\frac{2}{3}$ & 0 & $d\geq 4$ \\
\hline
$\PROJ_{\ydiagram{1},\ydiagram{2,1},\ydiagram{2,2}}$ & $\chi(\ydiagram{2,2})$ & $\chi(\ydiagram{2,1})$ & $\xi(\ydiagram{2,2})$ & $\xi(\ydiagram{2,1})$ & 0 & 0 & $d\geq 4$ \\
\hline
$\PROJ_{\ydiagram{1},\ydiagram{2,1},\ydiagram{3,1}}$ & $\chi(\ydiagram{3,1})$ & $\chi(\ydiagram{2,1})$ & $\xi(\ydiagram{3,1})$ & $\xi(\ydiagram{2,1})$ & $\frac{2}{3}$ & 0 & $d\geq 3$ \\
\hline
$\PROJ_{\ydiagram{1},\ydiagram{3},0}$ & $\chi(\ydiagram{4})$ & $\chi(\ydiagram{3})$ & $\xi(0)$ & $\xi(\ydiagram{1})$ & $1$ & $\frac{d+2}{3d}$ & $d\geq 2$ \\
\hline
$\PROJ_{\ydiagram{1},\ydiagram{3},\ydiagram{1,1}}$ & $\chi(\ydiagram{3,1})$ & $\chi(\ydiagram{3})$ & $\xi(\ydiagram{1,1})$ & $\xi(\ydiagram{1})$ & $-\frac{1}{3}$ & 0 & $d\geq 2$ \\
\hline
$\PROJ_{\ydiagram{1},\ydiagram{3},\ydiagram{3,1}}$ & $\chi(\ydiagram{3,1})$ & $\chi(\ydiagram{3})$ & $\xi(\ydiagram{3,1})$ & $\xi(\ydiagram{3})$ & $-\frac{1}{3}$ & 0 & $d\geq 3$ \\
\hline
$\PROJ_{\ydiagram{1},\ydiagram{3},\ydiagram{4}}$ & $\chi(\ydiagram{4})$ & $\chi(\ydiagram{3})$ & $\xi(\ydiagram{4})$ & $\xi(\ydiagram{3})$ & $1$ & 0 & $d\geq 2$ \\
\hline
\end{tabular}
\caption{\small The contributions of the projectors related to multiplicity-one irreps in the $(1,3)$-extendibility case. The table contains the values of the traces of the products of the operators from the first line and first column, divided by the relevant dimension. The trivial $\rep_{\ydiagram{1}}(C^{\ualg(d)})\ot\id_\hil^{\ot 3}$ and $\repp_{\ydiagram{1}}(C^{\So(d)})\ot\id_{\hil}^{\ot 3}$ operators have been omitted for brevity. The last column describes when the given projector appears with its contribution although possibly with different labels.} \label{tab:projs13}
\end{table}
\end{center}

\vspace{-0.3cm}

The ellipse contribution coming from the multiplicity-two irrep $\Repppp^{\mathrm{L}}_{\ydiagram{1}} \ot \Repppp^{\mathrm{R}}_{\ydiagram{3}}\ot {\color{brown}\Repp_{\ydiagram{2}}}$, appearing when $d\geq 2$:
\begin{equation}
\left(\frac{1}{3}+\frac{2}{3}\cos(\delta)\; , \; \frac{d+2}{6d}+\frac{\sqrt{d(d+4)}}{6d}\sin(\delta) + \frac{1}{3d}\cos(\delta)\right)\text{,}
\end{equation}
where $\delta\in[0,2\pi)$.

The ellipse contribution coming from the multiplicity-two irrep $\Repppp^{\mathrm{L}}_{\ydiagram{1}} \ot \Repppp^{\mathrm{R}}_{\ydiagram{2,1}}\ot {\color{orange}\Repp_{\ydiagram{1,1}}}$, appearing when $d\geq 3$:
\begin{equation}
\left(\frac{2}{3}\cos(\delta) \; , \; \frac{d+1}{6d}+\frac{\sqrt{3(d+2)(d-2)}}{12d}\sin(\delta) + \frac{d+4}{12d}\cos(\delta) \right)\text{,}
\end{equation}
where $\delta\in[0,2\pi)$.

The ellipse contribution coming from the multiplicity-two irrep $\Repppp^{\mathrm{L}}_{\ydiagram{1}} \ot \Repppp^{\mathrm{R}}_{\ydiagram{2,1}}\ot {\color{brass}\Repp_{\ydiagram{2}}}$, appearing when $d\geq 3$:
\begin{equation}
\left(\frac{1}{3}+\frac{1}{3}\cos(\delta) \; , \; \frac{d-1}{6d}+\frac{\sqrt{d(d-2)}}{6d}\sin(\delta) + \frac{1}{6d}\cos(\delta) \right)\text{,}
\end{equation}
where $\delta\in[0,2\pi)$.

In the $d=2$ case two extra projectors appear in connection with the multiplicity-two irreps that appear when $d\geq3$. These are collected in Table \ref{tab:projs13_spec}.

\vspace{-0.2cm}

\begin{table}[h!]
\centering
\def\arraystretch{1.5} 
\begin{tabular}{|c||c|c|c|c||c|c|||c|}
\hline
$\tr/\dim$ & $\rep_{\ydiagram{1}}^{\times 4}\left(C^{\ualg(d)}\right)$ & $\id_\hil \ot \rep_{\ydiagram{1}}^{\times 3}\left(C^{\ualg(d)}\right)$ & $\repp_{\ydiagram{1}}^{\times 4}\left(C^{\So(d)}\right)$ & $\id_\hil \ot \repp_{\ydiagram{1}}^{\times 3}\left(C^{\So(d)}\right)$ & $\smf_{1,3}$ & $\smb_{1,3}$ & Appearance\\
\hline \hline
$\PROJ_{\ydiagram{1},\ydiagram{2,1},0^*}$ & $\chi(\ydiagram{3,1})$ & $\chi(\ydiagram{2,1})$ & $\xi(0)$ & $\xi(\ydiagram{1})$ & $\frac{2}{3}$ & $\frac{1}{2}$ & $d=2$ \\
\hline
$\PROJ_{\ydiagram{1},\ydiagram{2,1},\ydiagram{2}}$ & $\chi\left(\ydiagram{3,1}\right)$ & $\chi(\ydiagram{2,1})$ & $\xi\left(\ydiagram{2}\right)$ & $\xi(\ydiagram{1})$ & $\frac{2}{3}$ & $\frac{1}{6}$ & $d=2$ \\
\hline
\end{tabular}
\caption{\small The contributions of the special projectors related to multiplicity-one irreps in the $(1,3)$-extendibility case that appear when $d=2$. The table contains the values of the traces of the product of the operators from the first line and first column, divided by the relevant dimension. The trivial $\rep_{\ydiagram{1}}(C^{\ualg(d)})\ot\id_\hil^{\ot 3}$ and $\repp_{\ydiagram{1}}(C^{\So(d)})\ot\id_{\hil}^{\ot 3}$ operators have been omitted for brevity. The last column describes when the given projector appears with its contribution.} \label{tab:projs13_spec}
\end{table}

Figure \ref{fig:13detailed_d2_d8} illustrates the convex hull of the contributions of the projectors and the ellipses for $d=2,8$.

\vspace{-0.3cm}

\begin{figure}[H]
    \centering
    \includegraphics[width=0.7\linewidth]{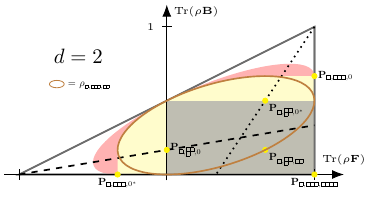}
\end{figure}

\vspace{-1cm}

\begin{figure}[H]
    \centering
    \includegraphics[width=0.7\linewidth]{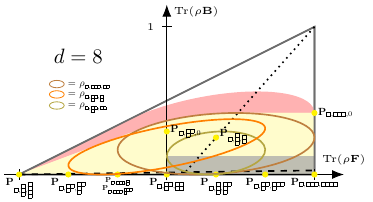}
    \caption{\small The set of Brauer states at $d=2,8$ in the $\tr(\rho\flip) - \tr(\rho\bb)$ parametrisation appearing as the grey triangle. The subset of $(1,2)$-extendible Brauer states appears in light red in the background. The subset of $(1,3)$-extendible Brauer states appears in light yellow. The dark yellow points and differently coloured ellipses denote the contributions of the appropriate projectors and qubit states, respectively. The set of Werner states is illustrated as the dashed line, while the set of isotropic states as the dotted line. The set of separable Brauer states constitute the grey rectangle.}
    \label{fig:13detailed_d2_d8}
\end{figure}

\subsection{The results for (2,2)-extendibility}\label{sec:app_res3}

In this section we present the $(2,2)$-extendibility case in more detail given its special symmetry. From Schur-Weyl duality (see Appendix \ref{sec:app_sw}) we have the following irrep decomposition:
\begin{equation}
    \hat{\Repppp}_4 \hat{\Rep}_4 \cong \left[\Repppp_{\ydiagram{1,1,1,1}}\ot\Rep_{\ydiagram{1,1,1,1}}\right] \OP \left[\Repppp_{\ydiagram{2,1,1}}\ot\Rep_{\ydiagram{2,1,1}} \right] \OP \left[\Repppp_{\ydiagram{2,2}}\ot\Rep_{\ydiagram{2,2}}\right] \OP \left[\Repppp_{\ydiagram{3,1}}\ot\Rep_{\ydiagram{3,1}} \right]\OP \left[\Repppp_{\ydiagram{4}}\ot\Rep_{\ydiagram{4}} \right]\text{.}
\end{equation}

The generators of the permutation symmetries are:
\begin{equation}
    \flip_{12}\text{,} \quad \flip_{34}\text{,} \quad \flip_{13}\flip_{24}\text{,}
\end{equation}
where the last generator is the \textit{faction flip}.
These three operators generate an 8-element subgroup in $\sym_4$, namely the dihedral group $\D_4$. It is well known that $\D_4$ is defined by the following presentation:
\begin{equation}
    \left\langle x,a \; \middle| \; a^4=x^2=e \text{,} \; xax^{-1}=a^{-1}\right\rangle\text{,}
\end{equation}
where $x$ and $a$ are the generators and $e$ denotes the identity. In our case, we can make the following identifications:
\begin{equation}
    x=\flip_{13}\flip_{24}\text{,} \qquad a=\flip_{12}(\flip_{13}\flip_{24})\text{.}
\end{equation}

We also have that $\D_4$ contains $\sym_2\times\sym_2$ as a subgroup with $\flip_{12}$ and $\flip_{34}$ as generators.

The $\D_4$ irreps are labelled as follows: the number denotes the dimension of the representation. There are 4 variants of 1-dimensional representations. The trivial is the simplest one. The ones that carry a star represent the irreps where the elements $\flip_{13}\flip_{24}$ and $\flip_{14}\flip_{23}$ are represented as $-1$. The ones that carry a bar represent the irreps where the elements $\flip_{12}$ and $\flip_{34}$ are represented as $-1$.

The character tables of these groups are described in Table \ref{tab:1}.

\def\arraystretch{1.5} 

\ytableausetup{mathmode, boxsize=\Yt}
\begin{center}
\begin{table}[ht]
\begin{tabular}{|c||c|c|c|c|c|c|c|}
    \hline
    $\sym_4$ & $()$ & \multicolumn{2}{c|}{$(12)$} & \multicolumn{2}{c|}{$(12)(34)$} & $(1,2,3,4)$ & $(1,2,3)$\\
    \hline \hline
    $\ydiagram{4}$ & 1 & \multicolumn{2}{c|}{1} & \multicolumn{2}{c|}{1} & 1 & 1 \\
    \hline
    $\ydiagram{1,1,1,1}$ & 1 & \multicolumn{2}{c|}{-1} & \multicolumn{2}{c|}{1} & -1 & 1 \\
    \hline
    $\ydiagram{2,2}$ & 2 & \multicolumn{2}{c|}{0} & \multicolumn{2}{c|}{2} & 0 & -1 \\
    \hline
    $\ydiagram{3,1}$ & 3 & \multicolumn{2}{c|}{1} & \multicolumn{2}{c|}{-1} & -1 & 0 \\
    \hline
    $\ydiagram{2,1,1}$ & 3 & \multicolumn{2}{c|}{-1} & \multicolumn{2}{c|}{-1} & 1 & 0 \\
    \hline
    \multicolumn{8}{c}{ }\\[-0.5em]
    \hline
    $\D_4$ & $\vphantom{\Big|}\{\id^{\ot 4}_{\hil}\}$ & \multicolumn{2}{c|}{$\{\flip_{12},\flip_{34}\}$} & $\{\flip_{12}\flip_{34}\}$ & $\{\flip_{13}\flip_{24},\flip_{14}\flip_{23}\}$ & $\{\flip_{12}(\flip_{13}\flip_{24}),(\flip_{13}\flip_{24})\flip_{12}\}$ & --- \\
    \hline \hline
    1 & 1 & \multicolumn{2}{c|}{1} & 1 & 1 & 1 & --- \\
    \hline
    $1^{\! *}$ & 1 & \multicolumn{2}{c|}{1} & 1 & -1 & -1 & ---\\
    \hline
    $\bar{1}$ & 1 & \multicolumn{2}{c|}{-1} & 1 & 1 & -1 & ---\\
    \hline
    $\bar{1}^{\! *}$ & 1 & \multicolumn{2}{c|}{-1} & 1 & -1 & 1 & --- \\ 
    \hline
    2 & 2 & \multicolumn{2}{c|}{0} & -2 & 0 & 0 & ---\\
    \hline
    \multicolumn{7}{c}{ }\\[-0.5em]
    \hline
    $\sym_2\times\sym_2$ & $\vphantom{\Big|}\{\id^{\ot 4}_{\hil}\}$ & $\{\flip_{12}\}$ & $\{\flip_{34}\}$ & $\{\flip_{12}\flip_{34}\}$ & --- & --- & ---\\
    \hline
    \hline
    $\ydiagram{2}\ot \ydiagram{2}$ & 1 & 1 & 1 & 1 & --- & --- & --- \\
    \hline
    $\ydiagram{1,1} \ot \ydiagram{2}$ & 1 & -1 & 1 & -1 & --- & --- & --- \\
    \hline
    $\ydiagram{2}\ot \ydiagram{1,1}$ & 1 & 1 & -1 & -1 & --- & --- & --- \\
    \hline
    $\ydiagram{1,1} \ot \ydiagram{1,1}$ & 1 & -1 & -1 & 1 & --- & --- & --- \\
    \hline
\end{tabular}
\caption{\small The character table of $\sym_4$ using the cyclic notation and the character tables of $\D_4$ and $\sym_2\times\sym_2$ using the relevant elements in $\sym_4$ to label the columns. The rows are labelled by the relevant irrep labels. The table is drawn up in a way to reflect which conjugacy class each element belongs to.}
\label{tab:1}
\end{table}
\end{center}

\vspace{-0.5cm}

From the character table, the following restrictions can be deduced from the $\sym_4$ irreps to $\D_4$:
\begin{equation}
\begin{gathered}
    \Repppp_{\ydiagram{4}}\Bigg\rvert_{\D_4}\cong \Repppp_{1} \text{,} \qquad 
    \Repppp_{\ydiagram{1,1,1,1}}\Bigg\rvert_{\D_4}\cong \Repppp_{\bar{1}}\text{,} \qquad
    \Repppp_{\ydiagram{2,2}}\Bigg\rvert_{\D_4}\cong \Repppp_{1}\OP \Repppp_{\bar{1}}\text{,} \\
    \Repppp_{\ydiagram{3,1}}\Bigg\rvert_{\D_4}\cong \Repppp_{1^{\! *}} \OP \Repppp_{2}\text{,} \qquad 
     \Repppp_{\ydiagram{2,1,1}}\Bigg\rvert_{\D_4}\cong\Repppp_{\bar{1}^{\! *}}\OP \Repppp_{2}\text{.}
\end{gathered}
\end{equation}
Furthermore, the restriction of the $\D_4$ irreps to $\sym_2\times\sym_2$ is trivial:
\begin{equation}
\begin{gathered}
    \Repppp_{1}\Big\rvert_{\sym_2\times\sym_2}\cong\Repppp^{\mathrm{L}}_{\ydiagram{2}}\ot\Repppp^{\mathrm{R}}_{\ydiagram{2}}\text{,} \qquad
    \Repppp_{1^{\! *}}\Big\rvert_{\sym_2\times\sym_2}\cong\Repppp^{\mathrm{L}}_{\ydiagram{2}}\ot\Repppp^{\mathrm{R}}_{\ydiagram{2}}\text{,} \qquad
    \Repppp_{\bar{1}}\Big\rvert_{\sym_2\times\sym_2}\cong \Repppp^{\mathrm{L}}_{\ydiagram{1,1}}\ot\Repppp^{\mathrm{R}}_{\ydiagram{1,1}}\text{,} \\
    \Repppp_{\bar{1}^{\! *}}\Big\rvert_{\sym_2\times\sym_2}\cong\Repppp^{\mathrm{L}}_{\ydiagram{1,1}}\ot\Repppp^{\mathrm{R}}_{\ydiagram{1,1}}\text{,} \qquad \Repppp_{2}\Big\rvert_{\sym_2\times\sym_2}\cong(\Repppp^{\mathrm{L}}_{\ydiagram{1,1}}\ot\Repppp^{\mathrm{R}}_{\ydiagram{2}}) \OP (\Repppp^{\mathrm{L}}_{\ydiagram{2}}\ot\Repppp^{\mathrm{R}}_{\ydiagram{1,1}}) \text{.} 
\end{gathered}
\end{equation}
Using the restriction rules found in Appendix \ref{sec:app_rep}, we find that the irrep decomposition an extending state $\hat{\rho}_{(2,2)}\in\states{\hil^{\ot 2}\ot \hil^{\ot 2}}$ has to be invariant to is the following for $d>7$:
\begin{equation}
\begin{split}\label{eq:rep22}
    \left(\hat{\Repppp}_4 \hat{\Rep}_4\right)\Big\rvert_{\D_4\times \Ort(d)} \cong \; & \Repppp_{1}\ot\left[{\color{blue}\Repp_{0}}\OP {\color{green}\Repp_{\ydiagram{2}}} \OP \Repp_{\ydiagram{2,2}} \OP {\color{blue}\Repp_{0}} \OP {\color{green}\Repp_{\ydiagram{2}}} \OP \Repp_{\ydiagram{4}}\right] \OP \\
    &\Repppp_{1^{\! *}}\ot \left[\Repp_{\ydiagram{1,1}} \OP \Repp_{\ydiagram{2}} \OP \Repp_{\ydiagram{3,1}}\right] \OP \\
    &\Repppp_{\bar{1}} \ot \left[\Repp_{\ydiagram{1,1,1,1}} \OP \Repp_{0} \OP \Repp_{\ydiagram{2}} \OP \Repp_{\ydiagram{2,2}}\right] \OP \\
    & \Repppp_{\bar{1}^{\! *}}\ot\left[\Repp_{\ydiagram{1,1}} \OP \Repp_{\ydiagram{2,1,1}} \right] \OP \\
    &\Repppp_{2}\ot\left[{\color{blue-green}\Repp_{\ydiagram{1,1}}} \OP \Repp_{\ydiagram{2,1,1}} \OP {\color{blue-green}\Repp_{\ydiagram{1,1}}} \OP \Repp_{\ydiagram{2}} \OP \Repp_{\ydiagram{3,1}} \right]\text{.}
\end{split}
\end{equation}
Note that the coloured representations have multiplicity greater than one.

The contributions of the projectors related to multiplicity-one irreps is gathered in Table \ref{tab:projs14}.

\begin{center}
\begin{table}[h!]
\def\arraystretch{1.5} 
\begin{tabular}{|c||c|c|c|c||c|c|||c|}
\hline
$\tr/\dim$ & $\rep_{\ydiagram{1}}^{\times 4}\left(C^{\ualg(d)}\right)$ & $\rep_{\ydiagram{1}}^{\times 2,\mathrm{L/R}}\left(C^{\ualg(d)}\right)$ & $\repp_{\ydiagram{1}}^{\times 4}\left(C^{\So(d)}\right)$ & $\repp_{\ydiagram{1}}^{\times 2,\mathrm{L/R}}\left(C^{\So(d)}\right)$ & $\smf_{2,2}$ & $\smb_{2,2}$ & Appearance\\
\hline \hline
$\PROJ_{1,\ydiagram{2,2}}$ & $\chi\left(\ydiagram{2,2}\right)$ & $\chi\left(\ydiagram{2}\right)$ & $\xi(\ydiagram{2,2})$ & $\xi(\ydiagram{2})$ & $-\frac{1}{2}$ & $0$ & $d\geq 4$ \\
\hline
$\PROJ_{1,\ydiagram{4}}$ & $\chi(\ydiagram{4})$ & $\chi\left(\ydiagram{2}\right)$ & $\xi(\ydiagram{4})$ & $\xi(\ydiagram{2})$ & 1 & 0 & $d\geq 2$ \\
\hline
$\PROJ_{1^{\! *},\ydiagram{1,1}}$ & $\chi(\ydiagram{3,1})$ & $\chi\left(\ydiagram{2}\right)$ & $\xi(\ydiagram{1,1})$ & $\xi(\ydiagram{2})$ & 0 & $\frac{d+2}{4d}$ & $d\geq 2$ \\
\hline
$\PROJ_{1^{\! *},\ydiagram{2}}$ & $\chi(\ydiagram{3,1})$ & $\chi\left(\ydiagram{2}\right)$ & $\xi(\ydiagram{2})$ & $\frac{\xi(0)+\xi(\ydiagram{2})}{2}$ & 0 & 0 & $d\geq 2$ \\
\hline
$\PROJ_{1^{\! *},\ydiagram{3,1}}$ & $\chi(\ydiagram{3,1})$ & $\chi\left(\ydiagram{2}\right)$ & $\xi(\ydiagram{3,1})$ & $\xi(\ydiagram{2})$ & 0 & 0 & $d\geq 3$ \\
\hline
$\PROJ_{\bar{1},\ydiagram{1,1,1,1}}$ & $\chi\left(\ydiagram{1,1,1,1}\right)$ & $\chi\left(\ydiagram{1,1}\right)$ & $\xi\left(\ydiagram{1,1,1,1}\right)$ & $\xi(\ydiagram{1,1})$ & $-1$ & 0 & $d\geq 4$ \\
\hline
$\PROJ_{\bar{1},0}$ & $\chi(\ydiagram{2,2})$ & $\chi\left(\ydiagram{1,1}\right)$ & $\xi\left(0\right)$ & $\xi(\ydiagram{1,1})$ & $\frac{1}{2}$ & $\frac{d-1}{2d}$ & $d\geq 2$ \\
\hline
$\PROJ_{\bar{1},\ydiagram{2}}$ & $\chi(\ydiagram{2,2})$ & $\chi\left(\ydiagram{1,1}\right)$ & $\xi\left(\ydiagram{2}\right)$ & $\xi(\ydiagram{1,1})$ & $\frac{1}{2}$ & $\frac{d-2}{4d}$ & $d\geq 3$ \\
\hline
$\PROJ_{\bar{1},\ydiagram{2,2}}$ & $\chi(\ydiagram{2,2})$ & $\chi\left(\ydiagram{1,1}\right)$ & $\xi\left(\ydiagram{2,2}\right)$ & $\xi(\ydiagram{1,1})$ & $\frac{1}{2}$ & 0 & $d\geq 4$ \\
\hline
$\PROJ_{\bar{1}^{\! *},\ydiagram{1,1}}$ & $\chi(\ydiagram{2,1,1})$ & $\chi\left(\ydiagram{1,1}\right)$ & $\xi(\ydiagram{1,1})$ & $\xi(\ydiagram{1,1})$ & 0 & $\frac{d-2}{4d}$ & $d\geq 3$ \\
\hline
$\PROJ_{\bar{1}^{\! *},\ydiagram{2,1,1}}$ & $\chi(\ydiagram{2,1,1})$ & $\chi\left(\ydiagram{1,1}\right)$ & $\xi(\ydiagram{2,1,1})$ & $\xi(\ydiagram{1,1})$ & 0 & 0 & $d\geq 4$ \\
\hline
$\PROJ_{2,\ydiagram{2,1,1}}$ & $\chi(\ydiagram{2,1,1})$ & $\frac{\chi\left(\ydiagram{1,1}\right) + \chi\left(\ydiagram{2}\right)}{2}$ & $\xi\left(\ydiagram{2,1,1}\right)$ & $\frac{\xi(\ydiagram{1,1})+ \xi(\ydiagram{2})}{2}$ & $-\frac{1}{2}$ & 0 & $d\geq 4$ \\
\hline
$\PROJ_{2,\ydiagram{2}}$ & $\chi(\ydiagram{3,1})$ & $\frac{\chi\left(\ydiagram{1,1}\right) + \chi\left(\ydiagram{2}\right)}{2}$ & $\xi\left(\ydiagram{2}\right)$ & $\frac{\xi(\ydiagram{1,1})+ \xi(\ydiagram{2})}{2}$ & $\frac{1}{2}$ & $\frac{1}{4}$ & $d\geq 2$ \\
\hline
$\PROJ_{2,\ydiagram{3,1}}$ & $\chi(\ydiagram{3,1})$ & $\frac{\chi\left(\ydiagram{1,1}\right) + \chi\left(\ydiagram{2}\right)}{2}$ & $\xi\left(\ydiagram{3,1}\right)$ & $\frac{\xi(\ydiagram{1,1})+ \xi(\ydiagram{2})}{2}$ & $\frac{1}{2}$ & 0 & $d\geq 3$ \\
\hline
\end{tabular}
\caption{\small The contributions of the projectors related to multiplicity-one irreps in the $(2,2)$-extendibility case. The table contains the values of the traces of the products of the operators from the first line and first column, divided by the relevant dimension. The notation $\rep_{\ydiagram{1}}^{\times 2, \mathrm{L/R}}$ denotes the representations $\rep_{\ydiagram{1}}^{\times 2}\ot\id_{\hil}^{\ot 2}$ and $\id_{\hil}^{\ot 2}\ot\rep_{\ydiagram{1}}^{\times 2}$, respectively, and similarly for the case of $\repp_{\ydiagram{1}}^{\times 2,\mathrm{L/R}}$. The last column describes when the given projector appears with its contribution although possibly with different labels.} \label{tab:projs14}
\end{table}
\end{center}

\vspace{-0.4cm}

The ellipse contribution coming from the multiplicity-two irrep $\Repppp_{1}\ot {\color{blue}\Repp_{0}}$, appearing when $d\geq 2$:
\begin{equation}
\left(\frac{1}{4}+\frac{3}{4}\cos(\delta) \; , \; \frac{d+1}{4d}+\frac{d+5}{12d}\cos(\delta)+\frac{\sqrt{2(d+2)(d-1)}}{6d}\sin(\delta)\right)\text{,}
\end{equation}
where $\delta\in[0,2\pi)$. Note that this is the same ellipse that appears in the (1,2)-extendibility case.

The ellipse contribution coming from the multiplicity-two irrep $\Repppp_{1}\ot {\color{green}\Repp_{\ydiagram{2}}}$, appearing when $d\geq 3$:
\begin{equation}
\left(\frac{1}{4}+\frac{3}{4}\cos(\delta) \; , \; \frac{d+2}{8d} + \frac{d+10}{24d}\cos(\delta)-\frac{\sqrt{2(d+4)(d-2)}}{12d}\sin(\delta) \right)\text{,}
\end{equation}
where $\delta\in[0,2\pi)$.

The ellipse contribution coming from the multiplicity-two irrep $\Repppp_{2}\ot {\color{blue-green}\Repp_{\ydiagram{1,1}}}$, appearing when $d\geq 3$:
\begin{equation}
\left(\frac{1}{2}\cos(\delta) \; , \; \frac{1}{8}+\frac{1}{4d}\cos(\delta)+\frac{\sqrt{d^2-4}}{8d}\sin(\delta) \right)\text{,}
\end{equation}
where $\delta\in[0,2\pi)$.

In the $d=2$ case two extra projectors appear in connection with the multiplicity-two irreps that appear when $d\geq3$. These are collected in Table \ref{tab:projs22_spec}.

\begin{center}
\begin{table}[h!]
\def\arraystretch{1.5} 
\begin{tabular}{|c||c|c|c|c||c|c|||c|}
\hline
$\tr/\dim$ & $\rep_{\ydiagram{1}}^{\times 4}\left(C^{\ualg(d)}\right)$ & $\rep_{\ydiagram{1}}^{\times 2,\mathrm{L/R}}\left(C^{\ualg(d)}\right)$ & $\repp_{\ydiagram{1}}^{\times 4}\left(C^{\So(d)}\right)$ & $\repp_{\ydiagram{1}}^{\times 2,\mathrm{L/R}}\left(C^{\So(d)}\right)$ & $\smf_{2,2}$ & $\smb_{2,2}$ & Appearance\\
\hline \hline
$\PROJ_{1,\ydiagram{2}}$ & $\chi\left(\ydiagram{4}\right)$ & $\chi\left(\ydiagram{2}\right)$ & $\xi(\ydiagram{2})$ & 
$\frac{\xi(0)+\xi(\ydiagram{2})}{2}$ & $1$ & $\frac{1}{2}$ & $d=2$ \\
\hline
$\PROJ_{2,0^*}$ & $\chi(\ydiagram{3,1})$ & $\frac{\chi\left(\ydiagram{1,1}\right)+\chi\left(\ydiagram{2}\right)}{2}$ & $\xi(0)$ & $\frac{\xi(0)+\xi(0)}{2}$ & $\frac{1}{2}$ & $\frac{1}{4}$ & $d=2$ \\
\hline
\end{tabular}
\caption{\small The contributions of the special projectors related to multiplicity-one irreps in the $(2,2)$-extendibility case that appear when $d=2$. The table contains the values of the traces of the product of the operators from the first line and first column, divided by the relevant dimension. The notation $\rep_{\ydiagram{1}}^{\times 2, \mathrm{L/R}}$ denotes the representations $\rep_{\ydiagram{1}}^{\times 2}\ot\id_{\hil}^{\ot 2}$ and $\id_{\hil}^{\ot 2}\ot\rep_{\ydiagram{1}}^{\times 2}$, respectively, and similarly for the case of $\repp_{\ydiagram{1}}^{\times 2,\mathrm{L/R}}$. The last column describes when the given projector appears.} \label{tab:projs22_spec}
\end{table}
\end{center}

\vspace{-0.5cm}

Figure \ref{fig:22detailed_d2_d8} illustrates the convex hull of the contributions of the projectors and the ellipses for $d=2,8$.

\begin{figure}[H]
    \centering
    \includegraphics[width=0.7\linewidth]{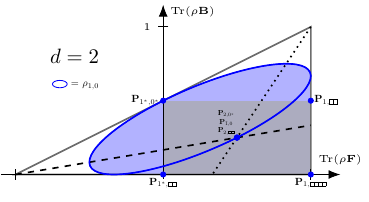}
\end{figure}

\vspace{-1.5cm}

\begin{figure}[H]
    \centering
    \includegraphics[width=0.7\linewidth]{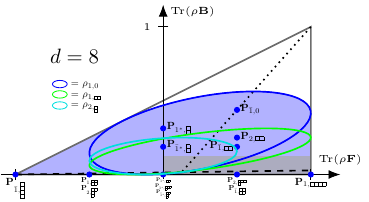}
    \caption{\small The set of Brauer states at $d=2,8$ in the $\tr(\rho\flip) - \tr(\rho\bb)$ parametrisation appearing as the grey triangle. The subset of $(2,2)$-extendible Brauer states appears in light blue in the background and coincides with the subset of $(1,2)$-extendible Brauer states. The dark blue points and differently coloured ellipses denote the contributions of the appropriate projectors and qubit states, respectively. The set of Werner states is illustrated as the dashed line, while the set of isotropic states as the dotted line. The set of separable Brauer states constitute the grey rectangle.}
  \label{fig:22detailed_d2_d8}
\end{figure}

\setcounter{equation}{0}

\section{Derivation of the estimate of the set of $(1,m)$-extendible Brauer states} \label{sec:app_estimate}

In this appendix, we present the derivation of the estimate of the set of $(1,m)$-extendible Brauer states. This is done by finding multiplicity-one irreps in the irrep decomposition of the representation that the $(1,m)$-extending states have to be invariant to. Specifically, we are looking to find multiplicity-one irreps that have extremal contribution to the parameter space in the $\tr(\rho\flip)$ and $\tr(\rho\bb)$ parameters. The convex hull of such points will give our estimate of the set of $(1,m)$-extendible Brauer states.

In the following, we shall denote by $\nu$ any $m$-box and at most $d$-row Young diagram appearing as a label for a $\U(d)$ irrep $\Rep_\nu$ on $\hil^{\ot m}$. Similarly, $\lambda$ denotes any $m+1$-box and at most $d$-row Young diagram appearing as a label in the irrep decomposition of $\Rep_{\ydiagram{1}}\ot\Rep_{\nu}$. On the other hand, $\nu'$ will denote a Young diagram appearing as a label for an $\Ort(d)$ irrep $\Repp_{\nu'}$ such that it appears in the irrep decomposition of the restriction $\Rep_{\nu}\big|_{\Ort(d)}$. Finally, $\lambda'$ will denote a Young diagram that is a label for an $\Ort(d)$ irrep $\Repp_{\lambda'}$ such that it appears in the irrep decomposition of the restriction $\Rep_{\lambda}\big|_{\Ort(d)}$ and also appears in the irrep decomposition of $\Repp_{\ydiagram{1}}\ot\Repp_{\nu'}$. Note that in the following the Young diagrams labelling $\Ort(d)$ irreps might be non-standard but can always be transformed into standard labels. When it does not cause confusion, we sometimes identify the labels with the representations themselves.

Thus, by using the concrete Casimir eigenvalues described in Equations \eqref{eq:chi} and \eqref{eq:xi} and the identity from Equation \eqref{eq:chixi} we have the following for normed projections onto multiplicity-one invariant spaces:
\begin{align}
    \frac{\tr(\PROJ_{\ydiagram{1},\nu,\lambda'}\smf_{1,m})}{\dim(\ydiagram{1},\nu,\lambda')}&=\frac{1}{2m}\big(\chi(\lambda)-\chi(\nu)-d\big)\text{,}\\
    \frac{\tr(\PROJ_{\ydiagram{1},\nu,\lambda'}\smb_{1,m})}{\dim(\ydiagram{1},\nu,\lambda')}&=\frac{1}{2dm}\bigg(\chi(\lambda)-\chi(\nu)-\big(\chi(\lambda')-\chi(\nu') + |\nu'|-|\lambda'|+1\big)\bigg)\text{.}
\end{align}

Trivially, any Young diagram $\nu$ can be written as a collection of equal-length rows that we call blocks:
\begin{equation}
    \nu=[\nu_1,\nu_2,\ldots,\nu_d]=[\tilde{\nu}_1,\tilde{\nu}_1,\ldots,\tilde{\nu}_1,\tilde{\nu}_2,\tilde{\nu}_2,\ldots,\tilde{\nu}_2,\dots]=[l_1\times \tilde{\nu}_1, l_b\times \tilde{\nu}_2,\ldots]\text{,}
\end{equation}
where $\nu_i$ describes the number of boxes in the $i$-th row of $\nu$ with $i\in[d]$ and we have introduced the notation $\tilde{\nu}_\alpha\in\{\nu_i\}_{i=1}^{d}$ where $\alpha\in[p]$, where $p$ is the number of blocks in $\nu$. For these we have that $\tilde{\nu}_1>\tilde{\nu}_2>\ldots$ and $l_1$ is the number of rows containing exactly $\tilde{\nu}_1$ boxes in them, $l_2$ is the number of rows with exactly $\tilde{\nu}_2$ boxes in them and so on.

Pieri's formula described in Appendix \ref{sec:app_rep} tells us that the Young diagram $\lambda$ can only appear if an extra box is appended to $\nu$ at the end of either row $1$, or row $l_1+1$, or row $l_1+l_2+1$, and so on. This is needed for $\lambda$ to be a valid Young diagram. Thus, based on Equation \eqref{eq:chi}, we find that if $\lambda$ came about by appending an extra box to the beginning of the $\delta$-th block of $\nu$, then we have:
\begin{equation}
    \chi(\lambda)-\chi(\nu)=2\tilde{\nu}_\delta+d-2\sum_{\alpha=1}^{\delta-1}l_\alpha\text{.}
\end{equation}

Similarly, the Pieri-type rule for $\Ort(d)$, described in Appendix \ref{sec:app_rep}, states that $\lambda'$ can only appear if an extra box is either appended to $\nu$ in a similar fashion as described before, or if a box is taken away from $\nu$ from the end of the non-zero blocks. In either case, we shall denote by $\delta'$ the number of the block in question. Thus, based on Equation \eqref{eq:chi} we find the following:
\begin{equation}
    \chi(\lambda')-\chi(\nu') + |\nu'| - |\lambda'| + 1=\begin{cases}2\tilde{\nu'}_{\delta'}+d-2\sum_{\alpha=1}^{\delta'-1}l'_\alpha\text{,} &\text{if $|\lambda'|=|\nu'|+1$,}\\ -2\tilde{\nu'}_{\delta'}-d+2+2\sum_{\alpha=1}^{\delta'}l'_\alpha \text{,} & \text{if $|\lambda'|=|\nu'|-1$.}\end{cases}
\end{equation}

Altogether we have:
\begin{align}
    \frac{\tr(\PROJ_{\ydiagram{1},\nu,\lambda'}\smf_{1,m})}{\dim(\ydiagram{1},\nu,\lambda')}&=\frac{\tilde{\nu}_\delta-\sum_{\alpha=1}^{\delta-1}l_\alpha}{m}\text{,}\\
    \frac{\tr(\PROJ_{\ydiagram{1},\nu,\lambda'}\smb_{1,m})}{\dim(\ydiagram{1},\nu,\lambda')}&=\begin{cases}\frac{1}{dm}\left(\tilde{\nu}_\delta - \tilde{\nu'}_{\delta'} + \sum_{\beta=1}^{\delta'-1}l'_\beta - \sum_{\alpha=1}^{\delta-1} l_\alpha\right)\text{,} & \text{if $|\lambda'|=|\nu'|+1$,} \\ \frac{1}{dm}\left(\tilde{\nu}_\delta + \tilde{\nu'}_{\delta'} + d -1 - \sum_{\beta=1}^{\delta'}l'_\beta - \sum_{\alpha=1}^{\delta-1} l_\alpha\right)\text{,} & \text{if $|\lambda'|=|\nu'|-1$.}\end{cases}
\end{align}

The $\tr(\rho\flip)$ parameter is between $-1$ and $1$. In the following, we try to find multiplicity-one irreps such that for given $m$ and $d$ they are extremal in the $\tr(\rho\flip)$ parameter while also trying to extremise the $\tr(\rho\bb)$ parameter. Thus, we are trying to find a 4-vertex polygon to describe the estimate of the set of $(1,m)$-extendible Brauer states. Note that the $\tr(\rho\flip)$ parameter will be minimised if the extra box of $\lambda$ is placed in the last line of $\nu$, which should preferably be empty. On the other hand, the $\tr(\rho\flip)$ parameter will be maximised if the extra box of $\lambda$ is placed in the first line of $\nu$, which should preferably be maximal.

We start with the latter case by noting that the positive extreme $\tr(\rho\flip)=1$ can always be achieved by diagrams $\nu=[m]$ and $\lambda=[m+1]$. Thus, we look for appropriate diagrams $\nu'$ and $\lambda'$ such that the $\tr(\rho\bb)$ parameter is extremised. A very natural ansatz is the following:
\begin{equation}
    \nu=[m]\text{,} \qquad \lambda=[m+1]\text{,} \qquad \nu'=[m]\text{,} \qquad \lambda'=[m+1]\text{.}
\end{equation}
It is trivial to check that this is indeed describing a multiplicity-one irrep. Furthermore, no modifications are needed to the diagrams $\nu'$ and $\lambda'$ as they are already standard labels in any dimension $d$. Thus, the above diagrams contribute the point
\begin{equation}
    (1,0)\text{.}
\end{equation}

For the upper extreme of the $\tr(\rho\bb)$ parameter,  we can choose the following Young diagrams if $m$ is odd:
\begin{equation}
    \nu=[m]\text{,} \qquad \lambda=[m+1]\text{,} \qquad \nu'=[1]\text{,} \qquad \lambda'=[0]\text{.}
\end{equation}
Again, it is trivial to check that the restriction of $\Rep_\nu$ to $\Ort(d)$ contains $\nu'=[1]$ and that the restriction of $\Rep_\lambda$ contains $\lambda'=[0]$ and that this necessarily has multiplicity-one. And, again, no modifications are needed. Thus, the above diagrams contribute the point
\begin{equation}
    \left(1,\frac{m+d-1}{d(m+1)}\right)\text{.}
\end{equation}
However, if $m$ is even then no such multiplicity-one irrep can be found while maintaining $\tr(\rho\flip)=1$, since we find that, for example, the choice of $\lambda'=[1]$ can be linked to two irreps coming from the restriction of $\Rep_\nu$: $\nu'=[2]$ and $\nu'=[0]$. Therefore, the best estimate for a point with $\tr(\rho\flip)=1$ parameter is to use the point we get from the odd $m+1$ parameter.

Now we go on to minimise the $\tr(\rho\flip)$ parameter. If $m<d-1$, then we can choose the Young diagrams
\begin{equation}
    \nu=[1^{m}]\text{,} \qquad \lambda=[1^{m+1}]\text{,} \qquad \nu'=[1^m]\text{,} \qquad \lambda'=[1^{m+1}]\text{,}
\end{equation}
where we have introduced the notation $[1^m]\coloneqq[1,\ldots,1]$ denoting the $m$ tall column of single boxes. Note that these are special Young diagrams in that the restriction of $\U(d)$ irreps labelled by such Young diagrams lead to a single $\Ort(d)$ irrep labelled by the same Young diagram, with the possible need of modification. Thus, them being multiplicity-one is guaranteed. Furthermore, the modification rule is trivial if $m>k$ we simply have $[1^{m}]\to[1^{d-m}]^*$. However, this does not change the fact that this ansatz contributes the point
\begin{equation}
    \left(-1,0\right)\text{.}
\end{equation}

If $m\geq d-1$ and $m-d$ is odd, then a natural ansatz to minimize the $\tr(\rho\flip)$ parameter is
\begin{equation}
    \nu=[m-d+2,1^{d-2}]\text{,} \qquad \lambda=[m-d+2,1^{d-1}]\text{,} \qquad \nu'=[1^{d-1}]\text{,} \qquad \lambda'=[1^{d}]\text{.}
\end{equation}
This case is much less straightforward than the one before, as one needs a careful analysis of the diagrams appearing in the restrictions and their multiplicities. Furthermore, the modification rules also need to be taken into account. However, the reasoning is the same as before: if $m-d$ is odd, then $m-d+2$ is odd as well, leading to a single, multiplicity one $\Ort(d)$ irrep, labelled by $\nu'=[1^{d-1}]$ in the restriction of $\nu=[m-d+2,1^{d-2}]$. Then, from the Pieri-type rule, we find a single $\lambda'=[1^{d}]$. This contributes the point
\begin{equation}
    \left(\frac{1-d}{m},0\right)\text{.}
\end{equation}
However, if $m-d$ is even, then there is no such multiplicity-one irrep. Thus, again, our best estimate for a point with $\tr(\rho\bb)=0$ parameter is to use the point we get from the odd $m+1$ parameter case.

The final step in completing the estimate is to find a multiplicity-one irrep with non-zero $\tr(\rho\bb)$ parameter but a minimal $\tr(\rho\flip)$ parameter. Here we focus on a set-up such that $\lambda'=[0]$ appears with multiplicity-one. Therefore, if $m<2d-2$ and $m$ is odd, we choose the following Young diagrams
\begin{equation}
    \nu=[2^{(m-1)/2},1]\text{,} \qquad \lambda=[2^{(m+1)/2}]\text{,} \qquad \nu'=[1]\text{,}\qquad \lambda'=[0]\text{.}
\end{equation}
From the restriction rules, it is apparent that $\nu'=[1]$ appears with multiplicity one in the restriction of $\nu$. Moreover, $\lambda'=[0]$ also appears in the restriction of $\lambda$ with multiplicity one. Thus, we have the point
\begin{equation}
    \left(\frac{3-m}{2m},\frac{2d-m+1}{2dm}\right)\text{.}
\end{equation}

\vspace{-0.15cm}

Note that the $\tr(\rho\flip)$ parameter is not absolutely minimal, but this is needed for the non-zero $\tr(\rho\bb)$ parameter to exist.
If $m$ is even, then, again, we run into multiplicities with this set-up. Thus, the best estimate for a point with such a $\tr(\rho\bb)$ parameter is to use the point we get from the odd $m+1$ parameter.

If $m\geq 2d-2$ and $m$ is odd, then we use the trick we used before and choose the Young diagrams
\begin{equation}
    \nu=[m+3-2d,2^{d-2},1]\text{,} \qquad \lambda=[m+3-2d,2^{d-1}]\text{,} \qquad \nu'=[1]\text{,}\qquad \lambda'=[0]\text{,}
\end{equation}
contributing the point
\begin{equation}
    \left(\frac{2-d}{m},\frac{1}{dm}\right)\text{.}
\end{equation}
Again, if $m$ is even, the best estimate for such a point is to use the point we get from the odd $m+1$ parameter.

Altogether, we get the following points:
\begin{align}
    \text{For $m\geq 1$:}\qquad &(1,0)\text{,}\\
    \text{For $m\geq 1$:}\qquad &\begin{cases} \text{if $m$ is odd:} & \left(1,\frac{m+d-1}{dm}\right)\text{,}\\ \text{if $m$ is even:} & \left(1,\frac{m+d}{d(m+1)}\right)\text{,}\end{cases}\\
    \text{For $m<d-1$:} \qquad &(-1,0)\text{,}\\
    \text{For $m\geq d-1$:}\qquad &\begin{cases} \text{if $m-d$ is odd:} & \left(\frac{1-d}{m},0\right)\text{,}\\ \text{if $m-d$ is even:} & \left(\frac{1-d}{m+1},0\right)\text{,}\end{cases}\\
    \text{For $m<2d-2$:}\qquad &\begin{cases}\text{if $m$ is odd:} & \left(\frac{3-m}{2m},\frac{2d-m+1}{2dm}\right)\text{,} \\ \text{if $m$ is even:} & \left(\frac{2-m}{2(m+1)},\frac{2d-m}{2d(m+1)}\right)\text{,} \end{cases}\\
    \text{For $m\geq 2d-2$:}\qquad &\begin{cases}\text{if $m$ is odd:} & \left(\frac{2-d}{m},\frac{1}{dm}\right)\text{,}\\ \text{if $m$ is even} & \left(\frac{2-d}{m+1},\frac{1}{d(m+1)}\right)\text{.} \end{cases}
\end{align}

As noted in Section \ref{sec:estimate}, this estimate could certainly be improved. For example, in the $m=2$ case we do not obtain the multiplicity-one irrep related to the Young diagrams:
\begin{equation}
    \nu=[1,1]\text{,} \qquad \lambda=[2,1]\text{,}\qquad \nu'=[1,1]\text{,}\qquad \lambda'=[1]\text{.}
\end{equation}

\vspace{-0.15cm}

Nevertheless, we maintain that the estimate is a good first-order approximation as illustrated by the Figures in \ref{fig:1m_estimate_compared_d2_d10} where we compare the $m=3$ estimate to our exact solution.

\begin{figure}[H]
    \centering
    \begin{subfigure}{.49\textwidth}
      \centering
      \includegraphics[width=\linewidth]{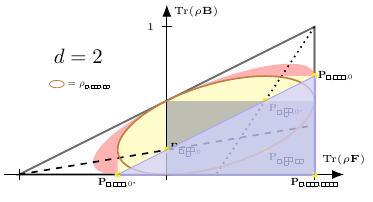}
    \end{subfigure}
    \begin{subfigure}{.49\textwidth}
      \centering
      \includegraphics[width=\linewidth]{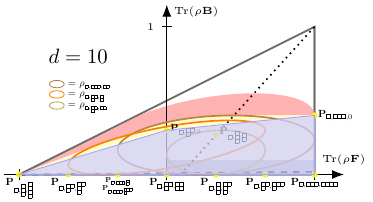}
    \end{subfigure}
    \caption{\small The set of Brauer states at $d=2$ and $d=10$ in the $\tr(\rho\flip) - \tr(\rho\bb)$ parametrisation appearing as the grey triangle. The subset of $(1,2)$-extendible Brauer states appears in red in the background, and the subset of $(1,3)$-extendible Brauer states appears in yellow in the foreground. The estimate of the set of $(1,3)$-extendible Brauer states appears as the blue polygon. The set of Werner states is illustrated as the dashed line, while the set of isotropic states as the dotted line. The set of separable Brauer states constitute the grey rectangle.}
  \label{fig:1m_estimate_compared_d2_d10}
\end{figure}

\setcounter{equation}{0}

\section{Proof of Theorem \ref{th:minmax}}\label{sec:app_proof}

In this appendix, we present the proof of Theorem \ref{th:minmax} that we state again for completeness:
\minmax*

For clarity, in this appendix, we mainly denote $\U(d)$ irrep labels by the Greek letter $\lambda$ and $\Ort(d)$ irrep labels by $\mu$.

\subsection{Lemmas}

In the following, we state and prove a number of lemmas related to $\U(d)$ and $\Ort(d)$ irreps which are needed to prove Theorem \ref{th:minmax}. We often refer to the dominance order of Young diagrams introduced in Reference~\cite{brylawski1973lattice}. Furthermore, we define the sum of Young diagrams as the sum of the corresponding rows of the two Young diagrams.

We start by stating lemmas for $\U(d)$ tensor irreps. There exists a so-called `dual' or `conjugate' Pieri rule, which is easy to prove and is known in the literature, see for example Chapter VI/6.~in Reference \cite{macdonald1998symmetric}.

\begin{lemma}[Dual Pieri rule] \label{lm:dual_pieri}
    Let $\Rep_\lambda$ and $\Rep_{[1^{N}]}$ be tensor irreps of $\U(d)$ labelled by Young diagrams $\lambda$ and $[1^{N}]$ of depth smaller or equal to $d$. Their tensor product decomposes as
    \begin{equation}
    \Rep_{\lambda}\ot\Rep_{[1^{N}]}\cong \bigoplus_{\lambda'} \Rep_{\lambda'}\text{,}
    \end{equation}
    where the summation goes through all Young diagrams $\lambda'$ one gets from $\lambda$ by appending $N$ number of boxes to it, no two in the same row, while keeping it a valid Young diagram.
\end{lemma}

The following is a trivial consequence of the dual Pieri rule:

\begin{corollary}\label{cor:largest_u_irrep_dominance}
    The largest Young diagram $\lambda'$ with respect to the dominance order of Young diagrams appearing as an irrep label in the tensor product $\Rep_{\lambda}\ot\Rep_{[1^N]}$ is $\lambda'=\lambda+[1^N]$.
\end{corollary}

\begin{proof}
    The statement is as follows
    \begin{equation}
        \lambda'=\lambda+[1^N]=[\lambda_1+1,\; \ldots,\; \lambda_N+1,\; \lambda_{N+1},\; \ldots,\; \lambda_d].
    \end{equation}
    
    From Lemma \ref{lm:dual_pieri} we also have that any irrep label $\tilde{\lambda}$ in the tensor product $\Rep_{\lambda}\ot\Rep_{[1^N]}$ comes from appending $N$ number of boxes to $\lambda$, no two in the same row while keeping it a valid Young diagram. The requirement for domination  of such an irrep label $\tilde{\lambda}$ by $\lambda'$ ($\tilde{\lambda}\trianglelefteq \lambda'$) is that for any integer $k\in[d]$ we must have:
    \begin{equation}
        \tilde{\lambda}_1+\ldots+\tilde{\lambda}_k \leq \lambda'_1+\ldots + \lambda'_k \text{.}
    \end{equation}
    
    However, since $\tilde{\lambda}\neq \lambda'$ there must be at least one row in the first $N$ rows of $\tilde{\lambda}$ that is missing an extra box at the end as it has been moved lower in the sequence of rows. Thus, trivially, $\lambda'$ will dominate $\tilde{\lambda}$.
\end{proof}

The following is a similar trivial lemma:

\begin{lemma}\label{lm:larger_u_gives_larger}
    Let $\lambda_1$ and $\lambda_2$ be two $n$-box Young diagrams such that $\lambda_1\triangleleft \lambda_2$ in dominance order. Then we have that $\lambda_1+[1^N]\triangleleft \lambda_2+[1^N]$ for any $N\in[d]$.
\end{lemma}

\begin{proof}
    The proof is trivial since if $\lambda_1\triangleleft \lambda_2$ then we have for all $k\in[d]$:
    \begin{equation}
        \sum_{i=1}^{k} (\lambda_1)_i \leq \sum_{i=1}^{k} (\lambda_2)_i\text{.}
    \end{equation}
    By adding $1$ to both sides in the first $N$ rows we do not change this relationship.
\end{proof}

Using these two lemmas and the one corollary, we can move on to prove an important lemma:

\begin{lemma}\label{lm:largest_u_irrep_big_tensor}
    Let $\lambda$ be an at most $d$-row Young diagram labelling a $\U(d)$ tensor irrep, and let $\{a_i\}_{i=1}^{d}$ be a set of non-negative integers. Let us take the following tensor product of representations:
    \begin{equation}
        \Rep_{\lambda}\ot\left(\bigotimes_{i=1}^{d} \Rep_{[1^{i}]}^{\ot a_{i}}\right)\text{.}
    \end{equation}
    
    The label $\lambda'$ of the $\U(d)$ irrep in this tensor product which is largest in the dominance order of Young diagrams is:
    \begin{equation}
        \lambda'=\lambda+\sum_{i=1}^{d} a_i[1^{i}]\text{.}
    \end{equation}
\end{lemma}

\begin{proof}
    We use proof by contradiction. Let us assume that $\tilde{\lambda}\neq\lambda'$ is an irrep label in the above defined tensor product such that $\lambda'\trianglelefteq \tilde{\lambda}$ and it similarly dominates every other label as well. Without loss of generality, let us assume that $a_{d}\geq 1$ and let us detach the last tensor product in the above equation:
    \begin{equation}
        \left[\Rep_{\lambda}\ot\left(\bigotimes_{i=1}^{d-1} \Rep_{[1^{i}]}^{\ot a_{i}}\right) \ot \Rep_{[1^{d}]}^{\ot (a_{d}-1)}\right] \ot \Rep_{[1^{d}]}\text{.}
    \end{equation}
    
    From Corollary \ref{cor:largest_u_irrep_dominance} we know that if any $\Rep_{\gamma}$ is an irrep in the irrep decomposition of the tensor product in brackets, then the $\U(d)$ irrep in $\Rep_{\gamma}\ot\Rep_{[1^{d}]}$ which has the largest Young diagram label in dominance order is $\gamma+[1^{d}]$. Since $\tilde{\lambda}$ is the largest in the dominance order, whatever Young diagram $\tilde{\gamma}$ it is related to, we must have that $\tilde{\lambda}=\tilde{\gamma}+[1^d]$, otherwise $\tilde{\lambda}$ would not be the largest.
    
    Thus, we have that if $\tilde{\lambda}\neq\lambda'$, then we must have that $\tilde{\lambda}-[1^{d}]\neq \lambda'-[1^{d}]$. Now, assuming that the tensor product has enough factors, we step back one more time and examine the label $\tilde{\delta}$ such that $\Rep_{\tilde{\delta}}\ot \Rep_{[1^{N}]}$ is the tensor product in which $\Rep_{\tilde{\gamma}}$ is an irrep, for some $N\in[d]$. Corollary \ref{cor:largest_u_irrep_dominance} states that if $\tilde{\gamma}\neq \tilde{\delta}+[1^{N}]$, then $\tilde{\delta}+[1^{N}]\triangleright\tilde{\gamma}$, and thus from Lemma \ref{lm:larger_u_gives_larger} we find that $\tilde{\delta}+[1^{N}]+[1^{d}]\triangleright \tilde{\lambda}$ which is a contradiction. This argument can be repeated any number of times leading to:
    \begin{equation}
        \tilde{\lambda}-\sum_{i=1}^{d} a_{i}[1^{i}] \neq \lambda'-\sum_{i=1}^{d} a_i[1^{i}]=\lambda\text{,}
    \end{equation}
    which is a contradiction.
\end{proof}

Now we move on to lemmas related to $\Ort(d)$ irreps. The following is a trivial result that considers two Young diagrams, one of which contains the other:

\begin{lemma}\label{lm:largercas}
    Let $\mu$ and $\mu'$ be at most $k$-row Young diagrams labelling $\Ort(d)$ irreps, where $d=2k$ or $d=2k+1$, such that $\mu_i \leq \mu'_i$ for all $i\in[k]$. Then for the quadratic Casimir eigenvalues for the corresponding $\So(d)$ irreps we have:
    \begin{equation}
        \xi(\mu)\leq \xi(\mu')\text{.}
    \end{equation}
\end{lemma}

\begin{proof}
    Let $\mu$ have $q\leq k$ rows, and $\mu'$ have $q\leq q'\leq k$ rows. Furthermore, let $z_i\coloneqq\mu'_i-\mu_i\geq 0$. For the Casimir eigenvalues we have:
    \begin{equation}
        \xi(\mu)=\sum_{i=1}^{q} 2\mu_i(\mu_i + d - 2i)\text{,}
    \end{equation}
    
    \begin{align}
        \xi(\mu')&=\sum_{i=1}^{q'} 2\mu'_i(\mu'_i+d-2i)= \sum_{i=1}^{q'} 2(\mu_i+z_i)(\mu_i+z_i + d - 2i)\\
        &=\sum_{i=1}^{q} 2\mu_i(\mu_i+d-2i) + \sum_{i=1}^{q'} 2z_i(2\mu_i+z_i+d-2i)\\
        &\geq \xi(\mu)\text{,}
    \end{align}
    where we have used that since $d\geq 2q'$ we have that $d-2i\geq 0$ for all $i$.
\end{proof}

We have the following, trivial corollary:

\begin{corollary}
    If $\lambda$ is an at most $k$-row Young diagram labelling a $\U(d)$ irrep, then the $\Ort(d)$ irrep in the restriction $\Rep_{\lambda}\big|_{\Ort(d)}$ related to the $\So(d)$ irrep(s) with the largest quadratic Casimir eigenvalue is $\Repp_{\lambda}$.
\end{corollary}

\begin{proof}
    Given that $\lambda$ has at most $k$ rows, no modification rules are ever needed when applying the branching rules described in Appendix \ref{sec:app_branch}. The branching rules require that any irrep label in the restriction of $\Rep_{\lambda}$ be contained within $\lambda$. Thus, given that $\mu=\lambda$ is possible by choosing $\kappa=0$, we have that necessarily every other irrep in the irrep decomposition is related to $\So(d)$ irreps with smaller quadratic Casimir eigenvalues.
\end{proof}

We can generalise a previous result of ours to the $\Ort(d)$ case. Note that the lemma is true even if $\mu$ or $\mu'$ does not label an $\Ort(d)$ irrep.

\begin{lemma} \label{lm:ort_dominance}
    Let $\mu$ and $\mu'$ be at most $d$-row, $n$-box Young diagrams labelling $\Ort(d)$ irreps. These may be non-standard labels that need modification. If $\mu\triangleleft\mu'$ with respect to the dominance order of Young diagrams, then we have that $\xi(\mu)<\xi(\mu')$ for the quadratic Casimir eigenvalues of the related $\So(d)$ irreps.
\end{lemma}

\begin{proof}
    The proof is based on a result in \cite{jakab2022extendibility}, which states the same thing for $\U(d)$ irrep labels. We simply use the fact that if $\mu$ and $\mu'$ are valid, but possibly non-standard, labels of $\Ort(d)$ irreps, with $n$ boxes then, according to Equation \eqref{eq:chixi} we have:
    \begin{equation}
        \xi(\mu)=2(\chi(\mu)-n) \text{,} \qquad \xi(\mu')=2(\chi(\mu')-n)\text{,}
    \end{equation}
    where $\chi(\mu)$ is the $\ualg(d)$ quadratic Casimir eigenvalue related to the $\U(d)$ irrep labelled by $\mu$. From this and from the result in \cite{jakab2022extendibility}, the statement follows trivially.
\end{proof}

A dual Pieri-type rule can also be formulated for $\Ort(d)$, see Theorem 5.4.~in Ref.~\cite{sundaram1990orthogonal} or Theorem 4.1.~in Ref.~\cite{okada2016pieri} for a more modern treatment.

\begin{lemma}[Dual Pieri-type rule] \label{lm:dual_pieri_ort}
    Let $\Repp_\mu$ and $\Repp_{[1^{N}]}$ be tensor irreps of $\Ort(d)$ ($d=2k$ or $d=2k+1$) labelled by Young diagrams $\mu$ and $[1^{N}]$ of depth smaller or equal to $k$. Their tensor product decomposes as
    \begin{equation}\label{eq:dual_pieri_ort}
        \Repp_{\mu}\ot\Repp_{[1^{N}]}\cong \bigoplus_{\mu'} \Repp_{\mu'}\text{,}
    \end{equation}
    where the sum is over the multiset of $\mu'$ obtained from $\mu$ by removing $0 \leq i \leq N$ number of boxes each from different rows (while keeping it a Young diagram) and then adding $N-i$ number of boxes each to different rows (while keeping it a Young diagram). Equation \eqref{eq:dual_pieri_ort} extends naturally to the cases with associate irreps.
\end{lemma}

Putting some of the previous lemmas together, we obtain the following lemma:

\begin{lemma}\label{lm:largest_o_irrep_casimir}
    Let $\Repp_{\mu}$ and $\Repp_{[1^{N}]}$ be tensor irreps of $\Ort(d)$ ($d=2k$ or $d=2k+1$) labelled by Young diagrams $\mu$ and $[1^{N}]$ of depth smaller or equal to $k$. The $\Ort(d)$ irrep in the tensor product $\Repp_\mu\ot\Repp_{[1^N]}$ related to the $\So(d)$ irrep(s) with largest quadratic Casimir eigenvalue is labelled by $\mu+[1^N]$.
\end{lemma}

\begin{proof}
    From the dual Pieri-type rule (Lemma \ref{lm:dual_pieri_ort}) we have that for any fixed $i$ such that $0\leq i \leq N$ we may remove $i$ number of boxes from $\mu$, each from different rows. This results in the $\left\{\nu\in\Yng(|\mu|-i,k) \;\middle | \; c^{\mu}_{\nu,[1^{i}]}\neq 0\right\}$ set of Young diagram which are all contained in $\mu$. After the removal we add $N-i$ number of boxes to such a Young diagram $\nu^{(i)}$ each to different rows, which can be described as the labels of the resulting irreps in the $\U(d)$ irrep tensor product $\Rep_{\nu^{(i)}}\ot\Rep_{[1^{N-i}]}$. However, from Corollary \ref{cor:largest_u_irrep_dominance} we have that for any such $\nu^{(i)}$, the Young diagram which is greatest in the dominance order in $\Rep_{\nu^{(i)}}\ot\Rep_{[1^{N-i}]}$ is, $\nu_i+[1^{N-i}]$. Finally, given Lemma \ref{lm:ort_dominance} this means that for any $\nu_i$, the Young diagram related to the $\So(d)$ irrep(s) with largest quadratic Casimir eigenvalue is $\nu_i+[1^{N-i}]$. This is also necessarily contained in $\mu+[1^N]$, which proves, using Lemma \ref{lm:largercas}, that $\mu+[1^N]$ is related to the $\So(d)$ irrep(s) with largest quadratic Casimir eigenvalue.
\end{proof}

Finally, we have our main lemma:
\begin{lemma}\label{lm:largest_o_irrep_big_tensor}
    Let $\mu=$ be an at most $k$-row Young diagram labelling an $\Ort(d)$ irrep where $d=2k$ or $d=2k+1$. Let $\{a_i\}_{i=1}^{k}$ be a set of non-negative integers. Let us take the following tensor product of representations:
    \begin{equation}
        \Repp_{\mu}\ot\left(\bigotimes_{i=1}^{k} \Repp_{[1^{i}]}^{\ot a_{i}}\right)\text{.}
    \end{equation}
    
    The label $\mu'$ of the $\Ort(d)$ irrep in this tensor product related to the $\So(d)$ irrep(s) with the largest quadratic Casimir eigenvalue is:
    \begin{equation}
        \mu'=\mu+\sum_{i=1}^{k} a_i[1^{i}]\text{.}
    \end{equation}
\end{lemma}

\begin{proof}
    The proof is identical to that of Lemma \ref{lm:largest_u_irrep_big_tensor} except that we need to invoke Lemma \ref{lm:largest_o_irrep_casimir}.
\end{proof}

Some additional lemmas are needed for proving the minimal case. We mainly show that the column-type Young diagrams are related to smallest quadratic Casimir eigenvalues and that they cannot be eliminated with the modification rules.

\begin{lemma}\label{lm:append_to_column}
    Let $[1^{N}]$ be an $N$-tall column Young diagram as a (possibly non-standard) label for an $\Ort(d)$ irrep, such that $N\leq d$. Let $\mu'=[1^{N}]+\mu$ be a Young diagram such that $\mu$ is a Young diagram with at most $N$ rows. Let $\mu'$ also be a valid (possibly non-standard) label for an $\Ort(d)$ irrep. Then $\xi([1^{N}])\leq \xi(\mu')$.
\end{lemma}

\begin{proof}
    We can easily show the following:
    \begin{equation}
        \xi(\mu')=\xi([1^{N}]) + 4|\mu| + \xi(\mu),
    \end{equation}
    where $|\mu|$ is the number of boxes in the Young diagram $\mu$. We need to show that $4|\mu|+\xi(\mu)\geq 0$. Even if $\mu$ is itself not a valid label for an $\Ort(d)$ irrep we can still use the identity from Equation \eqref{eq:chixi}:
    \begin{equation}
        \xi(\mu)=2(\chi(\mu)-|\mu|)\text{.}
    \end{equation}
    
    Fortunately, $\mu$ is necessarily a valid label for a $\U(d)$ irrep, thus $\chi(\mu)\geq 0$. Therefore, we find:
    \begin{equation}
        \xi(\mu')= \xi([1^{N}]) + 2|\mu| + 2\chi(\mu)\geq \xi([1^{N}])\text{.}
    \end{equation}
\end{proof}

\begin{lemma}\label{lm:moredepth}
    Let $[1^{N}]$ be an $N$-tall column Young diagram as a (possibly non-standard) label for an $\Ort(d)$ irrep, such that $N\leq d$ and let $z$ be a non-negative integer such that $N+z\leq d$. Let $\mu'=[1^{N+z}]+\mu$ be a Young diagram such that $\mu$ is a Young diagram with $z$ number of boxes. Let $\mu'$ also be a valid (possibly non-standard) label for an $\Ort(d)$ irrep. Then $\xi([1^{N}])\leq \xi(\mu')$.
\end{lemma}

\begin{proof}
    We can easily show the following:
    \begin{equation}
        \xi(\mu')=\xi([1^{N}])+\xi(\mu)+2z(2+d-z-2N)\text{.}
    \end{equation}
    
    Lemma \ref{lm:ort_dominance} states that $\xi(\mu)\geq \xi([1^{z}])$, since the Young diagram $[1^{z}]$ is lowest in dominance order between the Young diagrams of $z$ boxes. Note that even if the Young diagram $\mu$ is not a label for an $\Ort(d)$ irrep the relationship is true. Thus we have:
    \begin{equation}
        \xi(\mu')\geq\xi([1^{N}])+\xi([1^{z}])+2z(2+d-z-2N)=\xi([1^{N}])+2z(2+2d-2z-2N)\geq \xi([1^{N}])\text{,}
    \end{equation}
    where we have used that $N+z\leq d$.
\end{proof}

\begin{lemma}\label{lm:nocancel}
    Let $\mu$ be a Young diagram which is a non-standard label for an $\Ort(d)$ irrep, where $d=2k$ or $d=2k+1$, such that it has at least two columns ($\mu_1=\ell>1$) and its first column has length $k<N\leq d$. Then, after modifications, the resulting diagram also has $\ell$ number of columns.
\end{lemma}

\begin{proof}
    The modification rules described in Appendix \ref{sec:app_mod} state that an $h\coloneqq 2N-d$ long continuous boundary hook has to be removed from $\mu$. This procedure might need to be done repeatedly, but it is enough to examine the first iteration. To be able to lose the last ($\ell$-th) column, we need to remove $N+\ell-1$ boxes from $\mu$. Thus, we need:
    \begin{equation}
        2N-d\overset{!}{=}N+\ell-1 \quad \Longrightarrow \quad N-d\overset{!}{=}\ell-1\geq 1\text{,}
    \end{equation}
    which is impossible since $N\leq d$. If any subsequent modification is needed this argument can be repeated for some $\mu'$.
\end{proof}

\subsection{Proof}

Now we move on to the proof of Theorem \ref{th:minmax}.

\begin{proof}
First we prove the result for the Young diagram related to the $\So(d)$ irrep(s) with minimal quadratic Casimir eigenvalue:

Given a Young diagram $\lambda$ labelling a $\U(d)$ irrep, the $\Ort(d)$ irreps in its restriction are given by the branching rules described in Appendix \ref{sec:app_branch}:
\begin{equation}
        \Rep_{\lambda}\Big\rvert_{\Ort(d)} \cong \bigoplus_{\mu} \left(\sum_{\kappa} c_{2\kappa,\mu}^{\lambda}\right) \Repp_{\mu}\text{,}
\end{equation}
where $\mu$ might be a non-standard label for an $\Ort(d)$ irrep. Given this formula, a natural ansatz to consider is when $\kappa$ is maximal ($\kappa=[x_1,\ldots,x_d]$) and thus $\mu$ has minimal number of boxes. In our case, this means that we consider Young diagrams with $Y\coloneqq\sum_{i=1}^{d}y_i$ number of boxes. The $Y$-tall column Young diagram $[1^{Y}]$ always appears as an irrep when $\kappa$ is maximal and, according to Lemma \ref{lm:ort_dominance}, it has the smallest quadratic Casimir eigenvalue among $Y$-box Young diagrams. Note that while sometimes modification rules might cancel out some $\Ort(d)$ irreps in the restriction of a $\U(d)$ irrep, we can show that it is impossible to cancel out one-column Young diagrams. Cancelling out a one-column Young diagram would require a non-standard, at least two column Young diagram to appear so that the modification rules described in Appendix \ref{sec:app_mod} allow for a modification with negative sign. However, Lemma \ref{lm:nocancel} shows that any multi-column Young diagram will remain multi-column even after modification rules, thus a one-column Young diagram will never be cancelled out.

Therefore, our ansatz surely appears in the restriction. It remains to prove that it has minimal quadratic Casimir eigenvalue. If $\mathrm{depth}(\lambda)\leq k$ then, by Lemmas \ref{lm:largercas} and \ref{lm:ort_dominance} our ansatz necessarily has minimal quadratic Casimir eigenvalue. However, if $\mathrm{depth}(\lambda)>k$ then it might happen that by taking a non-maximal $\kappa$, a Young diagram $\mu$ might be achieved that, after modification, leads to a smaller quadratic Casimir eigenvalue.

However, if $\mu$ is such that its first column remains $Y$ tall, then by Lemma \ref{lm:append_to_column} we have that the corresponding irrep has larger quadratic Casimir eigenvalue than the one corresponding to the ansatz $[1^{Y}]$. Thus, the only remaining possibility is to modify $\kappa$ such that the resulting $\mu$ has depth larger than $Y$. This can only happen if $\mathrm{depth}(\lambda)>Y$, otherwise we are done. If $\mathrm{depth}(\lambda)>Y$, then there must be rows with even number of boxes in them. The Littlewood-Richardson formula states that the only way $c_{2\kappa,\mu}^{\lambda}>0$ for a Young diagram $\mu$ having depth larger than $Y$ is to choose $\kappa$ in such a way that the rows with even number of boxes are not totally filled up. However, given that $c_{2\kappa,\mu}^{\lambda}$ contains $2\kappa$, we necessarily have extra boxes appearing in $\mu$, which, by the Littlewood-Richardson formula must appear at least in the second column of $\mu$. However, Lemma \ref{lm:moredepth} states that the simplest such Young diagrams correspond to larger quadratic Casimir eigenvalues than $[1^{Y}]$, and more complex Young diagrams can all be reduced to these simplest cases. 

Now we move on to prove the result for the Young diagram related to the $\So(d)$ irrep(s) with maximal quadratic Casimir eigenvalue:

We trivially have that the $\U(d)$ irrep $\Rep_{\lambda}$ can be understood as being in the irrep decomposition of the tensor product of the irreps related to its columns. Let
\begin{equation}
\lambda=a_{d}[1^{d}] + a_{d-1}[1^{d-1}] + \ldots + a_{1}[1]\text{,} 
\end{equation}
where $[1^N]$ denotes the $N$-tall one-box column and the addition is meant elementwise, describing the concatenation of the respective columns to the right side of the Young diagram. We have also introduced the non-negative integers $a_i$ that describe the number of exactly $i$-tall columns in the Young diagram $\lambda$.

Note that we have:
\begin{equation}
\lambda_{i}=\sum_{j=i}^{d} a_{j} \qquad \text{and} \qquad a_j=\begin{cases} \lambda_j-\lambda_{j+1}, & \text{if $j> d$,}\\ \lambda_d, & \text{if $j=d$.} \end{cases}
\end{equation}

Illustratively, using our previous example, we have:
\begin{equation}
[5,4,4,2,1]=\ydiagram{5,4,4,2,1}=\; \ydiagram{1,1,1,1,1}\; +\; \ydiagram{1,1,1,1}\; + \; 2\; \ydiagram{1,1,1}\; +\; \ydiagram{1}=[1^{5}] + [1^{4}] + 2[1^{3}] + [1]\text{.}
\end{equation}

Our statement was that we can understand $\Rep_{\lambda}$ as being in the irrep decomposition of the following tensor product:
\begin{equation}\label{eq:column_tensor}
\Rep_\lambda \in \Rep_{[1^{d}]}^{\ot a_{d}} \ot \Rep_{[1^{d-1}]}^{\ot a_{d-1}} \ot \cdots \ot \Rep_{[1]}^{\ot a_{1}}\text{,}
\end{equation}
where we have used a slight abuse of notation.

For the $\Ort(d)$ irreps in the restriction we also have:
\begin{equation}
\Rep_{\lambda}\bigg|_{\Ort(d)}\in \left(\Rep_{[1^{d}]}\bigg|_{\Ort(d)}\right)^{\ot a_{d}} \ot \left(\Rep_{[1^{d-1}]}\bigg|_{\Ort(d)}\right)^{\ot a_{d-1}} \ot \cdots \ot \left(\Rep_{[1]}\bigg|_{\Ort(d)}\right)^{\ot a_{1}}\text{.}
\end{equation}

From the branching rules described in Appendix \ref{sec:app_branch}, we have that for any $\U(d)$ irrep labelled by an $i$-tall single-box column Young diagram, we have a singular $\Ort(d)$ irrep in its restriction:
\begin{equation}
\Rep_{[1^{i}]}\bigg|_{\Ort(d)}\cong \Repp_{[1^{i}]}\cong \begin{cases}\Repp_{[1^{i}]}\text{,} & \text{if $i\leq k$,}\\ \Repp_{[1^{d-i}]^*}\text{,} & \text{if $i>k$.} \end{cases}
\end{equation}

Thus, we trivially have:
\begin{equation}
\Rep_{\lambda}\bigg|_{\Ort(d)}\in \Repp_{[1^{d}]}^{\ot a_d} \ot \Repp_{[1^{d-1}]}^{\ot a_{d-1}} \ot \cdots \ot \Repp_{[1]}^{\ot a_1}\cong \Repp_{0^*}^{\ot a_d}\ot \Repp_{[1]^*}^{\ot a_{d-1}}\ot \cdots \ot \Repp_{[1]}^{\ot a_1}\text{.}
\end{equation}

From Lemma \ref{lm:largest_o_irrep_big_tensor} we find that the $\Ort(d)$ irrep in the above tensor product with the largest quadratic Casimir eigenvalue has the following as its label:
\begin{equation}
\mu=\begin{cases}
a_k [1^k] + (a_{k-1}+a_{k+1})[1^{k-1}] + \ldots + (a_{1}+a_{d-1})[1]\text{,} & \text{if $d=2k$,}\\
(a_k + a_{k+1})[1^k] + (a_{k-1}+a_{k+2})[1^{k-1}] + \ldots + (a_{1}+a_{d-1})[1]\text{,} & \text{if $d=2k+1$.}
\end{cases}
\end{equation}

From this we can calculate that for the even case we have a telescopic sum:
\begin{align}
\mu_i=&a_k+\sum_{j=1}^{k-i} (a_{k-j}+a_{k+j})=\lambda_k-\lambda_{k+1}+\sum_{j=1}^{k-i} (\lambda_{k-j}-\lambda_{k-j+1}+\lambda_{k+j}-\lambda_{k+j+1})\\
=&\lambda_{k}-\lambda_{k+1} + (\lambda_{k-1} - \lambda_{k} + \lambda_{k+1} - \lambda_{k+2}) + (\lambda_{k-2}-\lambda_{k-1}+\lambda_{k+2}-\lambda_{k+3})\nonumber\\
&+ \ldots + \\
&+ (\lambda_{i+1}-\lambda_{i+2} + \lambda_{d-i-1} + \lambda_{d-i}) + (\lambda_{i}-\lambda_{i+1} + \lambda_{d-i} - \lambda_{d-i+1})\nonumber\\
=& \lambda_{i}-\lambda_{d-i+1}\text{,} 
\end{align}
and for the odd case we have a similar telescopic sum:
\begin{align}
\mu_i=&\sum_{j=0}^{k-i} (a_{k-j}+a_{k+1+j})=\sum_{j=0}^{k-i} (\lambda_{k-j}-\lambda_{k-j+1} + \lambda_{k+1+j} - \lambda_{k+1+j+1})\\
=& (\lambda_{k}-\lambda_{k+1}+\lambda_{k+1}-\lambda_{k+2}) + (\lambda_{k-1}-\lambda_{k}+\lambda_{k+2}-\lambda_{k+3})\nonumber\\
& + \ldots +\\
&+ (\lambda_{i+1}-\lambda_{i+2}+\lambda_{d-i-1}-\lambda_{d-i}) + (\lambda_{i}-\lambda_{i+1}+\lambda_{d-i}-\lambda_{d-i+1})\nonumber\\
=&\lambda_{i}-\lambda_{d-i+1}\text{.}
\end{align}

It only remains to prove that the $\Ort(d)$ irrep $\Repp_{\mu}$ is indeed an irrep in the restriction $\Rep_{\lambda}\big|_{\Ort(d)}$ since many other $\U(d)$ irreps appear in the tensor product described in Equation \eqref{eq:column_tensor}. Consider the following: according to Lemma \ref{lm:largest_u_irrep_big_tensor}, $\lambda$ is the largest Young diagram in dominance order in the tensor product of Equation \eqref{eq:column_tensor}.

Now, let us imagine that the above defined Young diagram $\mu$ labels an irrep in the restriction of some other $\U(d)$ irrep than $\Rep_{\lambda}$, say $\Rep_{\tilde{\lambda}}$. By the previous reasoning, $\tilde{\lambda}\triangleleft \lambda$. Now, let us go through the same reasoning for $\Rep_{\tilde{\lambda}}$: we can write it as the tensor product of the irreps labelled by its columns, and we find that the $\tilde{\mu}$ label, found in the irrep decomposition of the tensor product and related to the $\So(d)$ irrep with highest quadratic Casimir eigenvalue, is:
\begin{equation}
\tilde{\mu}=[\tilde{\lambda}_1-\tilde{\lambda}_d,\ldots]\text{.}
\end{equation}

We shall prove that $\xi(\tilde{\mu})< \xi(\mu)$. Given any Young diagram $\lambda$ of at most $d$ rows we can split it into two, at most $k$-row Young diagrams. We define the upper part as $\lambda_{\text{U}}\coloneqq[\lambda_1,\ldots,\lambda_k]$, and the lower part as $\lambda_{\text{L}}\coloneqq[\lambda_{k+1},\ldots,\lambda_{2k}]$ if $d=2k$ and $\lambda_{\text{L}}\coloneqq [\lambda_{k+2},\ldots, \lambda_{2k+1}]$ if $d=2k+1$. Our result says that $\mu=\lambda_{\text{U}}-(\lambda_{\text{L}})^{\mathrm{r}}$, and $\tilde{\mu}=\tilde{\lambda}_{\text{U}}-(\tilde{\lambda}_{\text{L}})^{\mathrm{r}}$, where the $r$ in the upper index means the reversal of the rows of the Young diagram.

Let us examine what $\tilde{\lambda}\triangleleft \lambda$ means for $\tilde{\mu}$ and $\mu$. Given that $\tilde{\lambda}\triangleleft \lambda$, we have that to get to $\tilde{\lambda}$ at least one box from $\lambda$ has been moved to some lower row. Let us focus on this simplest case as the general case can be inferred from repeating the argument. We have three general possibilities and two extra possibilities for the odd dimensional cases:
\begin{enumerate}
\item A box has been moved from a row in $\lambda_{\text{U}}$ to a lower row in $\lambda_{\text{U}}$. In this case we have that $\tilde{\lambda}_{\text{U}}\triangleleft \lambda_{\text{U}}$ and $\tilde{\lambda}_{\text{L}}=\lambda_{\text{L}}$ which leads to $\tilde{\mu}\triangleleft \mu$ given that $\mu=\lambda_{\text{U}}-(\lambda_{\text{L}})^{\mathrm{r}}$, and $\tilde{\mu}=\tilde{\lambda}_{\text{U}}-(\tilde{\lambda}_{\text{L}})^{\mathrm{r}}$.
\item A box has been moved from a row in $\lambda_{\text{L}}$ to a lower row in $\lambda_{\text{L}}$. In this case we have that $\tilde{\lambda}_{\text{L}}\triangleleft \lambda_{\text{L}}$ and $\tilde{\lambda}_{\text{U}}=\lambda_{\text{U}}$. Note that by broadening the definition of dominance order to non-standard Young diagrams we find the following relationships: $-\lambda_{\text{L}}\triangleleft -\tilde{\lambda}_{\text{L}}$ and $(\lambda_{\text{L}})^{\text{r}}\triangleleft (\tilde{\lambda}_{\text{L}})^{\text{r}}$ together leading to $-(\tilde{\lambda}_{\text{L}})^{\text{r}}\triangleleft -(\lambda_{\text{L}})^{\text{r}}$ which leads to $\tilde{\mu}\triangleleft \mu$.
\item A box is moved from $\lambda_{\text{U}}$ to $\lambda_{\text{L}}$. In this case we have that $\lambda_{\text{U}}$ contains $\tilde{\lambda}_{\text{U}}$ and $\tilde{\lambda}_{\text{L}}$ contains $\lambda_{\text{L}}$, leading to the fact that $\mu$ contains $\tilde{\mu}$.
\item In odd dimensions, a box can be moved from $\lambda_{\text{U}}$ to $\lambda_{k+1}$, the row which appears in neither halves. In this case $\lambda_{\text{U}}$ contains $\tilde{\lambda}_{\text{U}}$ while $\tilde{\lambda}_{\text{L}}=\lambda_{\text{L}}$, leading to the fact that $\mu$ contains $\tilde{\mu}$.
\item In odd dimensions, a box can be moved from $\lambda_{k+1}$, the row which appears in neither halves, to $\lambda_{\text{L}}$. In this case $\tilde{\lambda}_{\text{L}}$ contains $\lambda_{\text{L}}$ while $\tilde{\lambda}_{\text{U}}=\lambda_{\text{U}}$, leading to the fact that $\mu$ contains $\tilde{\mu}$. 
\end{enumerate}

Thus, in all 3 cases we have $\xi(\tilde{\mu})< \xi(\mu)$ by either invoking Lemma \ref{lm:largercas} or \ref{lm:ort_dominance}. Thus, we have that the original Young diagram $\lambda$ was largest in dominance order, and that the related Young diagram $\mu$ has the largest related quadratic Casimir eigenvalue which is not found in the decomposition of any $\U(d)$ irrep labelled by a Young diagram $\tilde{\lambda}$ which is lower in dominance order compared to $\lambda$. Therefore, $\Repp_{\mu}$ must be in the restriction of $\Rep_{\lambda}$.
\end{proof}

\setcounter{equation}{0}

\section{The limiting shape for de~Finetti extendibility}\label{sec:app_definetti_infty}

In this appendix, we present, for a given dimension $d$, the set of parameters for Brauer states that are de~Finetti extendible for all $n$.

\subsection{A first look at the limiting shape}

From the de~Finetti theorem, we know that these are states which are the convex combination of product states $\sigma\ot\sigma$. Note that these `de~Finetti states' form a convex set, with the product states on the boundary. Note also that, given a general de~Finetti state, it can be turned into a Brauer state by twirling it with the appropriate group elements without leaving the set of de~Finetti states. One can use the fact that the expectation value of the $\flip$ and $\bb$ operators does not change before or after the twirl, and the fact that taking the expectation value is a linear operation, thus preserving convex sets, to realise that one needs only to calculate the expectation values of product states to find the boundary of the set of parameters describing Brauer states that are de~Finetti extendible for all $n$:
\begin{align}
    f(\sigma)&=\tr(\flip \sigma \ot \sigma)\text{,}\\
    b(\sigma)&=\tr(\bb \sigma \ot \sigma)\text{.}
\end{align}

Given the pairs $f(\sigma),b(\sigma)$, one should be able to find the boundary in question. Let us have the following eigendecomposition for $\sigma$:
\begin{equation}
    \sigma=\sum_{i=1}^{d} p_i \proj{\psi_i}\text{,}
\end{equation}
where we have used Dirac notation to describe the orthogonal projectors onto $\psi_i$ and $\sum_{i=1}^dp_i=1$. A simple calculation leads to the following:
\begin{align}
    f(\sigma)&=\tr(\flip \sigma \ot \sigma)=
    \tr(\sigma^2)=\sum_{i=1}^{d} p_i^2\text{,}\\
    b(\sigma)&=\tr(\bb \sigma \ot \sigma)= \frac{1}{d} \tr(\sigma\sigma^T)= \frac{1}{d}\tr(\sigma\overline{\sigma})=\frac{1}{d}\sum_{i,j=1}^{d} p_ip_j |\inn{\psi_i}{\overline{\psi_j}}|^2\text{,}
\end{align}
where we have defined:
\begin{equation}
    \overline{\sigma}\coloneqq \sum_{i=1}^{d} p_i \proj{\overline{\psi_i}}\text{,}
\end{equation}
where the complex conjugation of the vector $\psi_i$ is with respect to the distinguished basis $\{e_i\}_{i=1}^{d}$. Note that the flip parameter is nothing else but the purity of $\sigma$, and is, in some sense, independent of the constituent vectors of $\sigma$, whereas the $\bb$ parameter very much depends on them.

A naive upper and lower bound can be put on both $f(\sigma)$ and $b(\sigma)$:
\begin{equation}
    f(\sigma)\in\left[\frac{1}{d},1\right]\text{,} \qquad \qquad b(\sigma)\in\left[0,\frac{1}{d}\right]\text{.}
\end{equation}
The lower bound of $f(\sigma)$ is given by the Cauchy-Schwarz inequality and is reached when $p_i=1/d$. The upper bound of $f(\sigma)$ is trivial ($p_i^2\leq p_i$) and is reached when $p_i=\delta_{i,a}$ for some $a\in[d]$. The lower bound for $b(\sigma)$ is trivial as we are summing non-negative quantities and is attained when $p_i=\delta_{i,a}$ for some $a\in[d]$ and $\psi_a \perp \overline{\psi_a}$. The upper bound is again determined by the Cauchy-Schwarz inequality and is attained when $p_i=\delta_{i,a}$ for some $a\in[d]$ and $\psi_a=\overline{\psi_a}$. This implies that as $d\to\infty$, the set of parameters of the $\infty$-de~Finetti-extendible states will be the lower right line segment of the parameter space $(0,0)-(1,0)$.

\subsection{Tighter upper and lower bound to $b(\sigma)$ as a function of $\{p_i\}_{i=1}^{d}$}

A tighter upper and lower bound can be given to $b(\sigma)$ as a function of $\{p_i\}_{i=1}^{d}$ and this is presented in this section. We have that:
\begin{equation}
    b(\sigma)=\frac{1}{d}\sum_{i,j=1}^{d} p_ip_j |\inn{\psi_i}{\overline{\psi_j}}|^2\text{.}
\end{equation}

Let us focus on the sum and realise that in general $\left\{\overline{\psi_i}\right\}_{i=1}^{d}$ could be any orthonormal basis in the Hilbert space. Thus, in general, we have:
\begin{equation}
    \sum_{i,j=1}^{d} p_i p_j |\inn{\psi_i}{\overline{\psi_j}}|^2=\sum_{i,j=1}^{d} p_i p_j |\inn{v_i}{w_j}|^2=\sum_{i,j=1}^{d} p_i p_j A_{ij}\text{,}
\end{equation}
where $\{v_i\}_{i=1}^{d}$ and $\{w_i\}_{i=1}^{d}$ are orthonormal bases and the matrix $A_{ij}$ has been introduced. Notice that $A_{ij}$ is a doubly stochastic matrix, meaning that:
\begin{equation}
    \sum_{i=1}^{d}A_{ij}=\sum_{j=1}^{d} A_{ij}=1\text{.}
\end{equation}
The Birkhoff-von Neumann theorem states that any doubly stochastic matrix can be written as the convex sum of permutation matrices. Let $h\in \sym_d$ be a permutation. Thus:
\begin{equation}
    A=\sum_{h\in\sym_d} r_{h} P^{h}\text{,}
\end{equation}
where $r_{h}$ are non-negative numbers summing to 1, and $P^{h}$ is the permutation operator corresponding to the permutation $h\in\sym_d$, leading to:
\begin{equation}
    A_{ij}=\sum_{h\in\sym_d} r_{h} P^{h}_{ij}=\sum_{h\in\sym_d} r_{h} \delta_{j,h(i)}\text{.}
\end{equation}
Thus, we have
\begin{equation}
    \sum_{i,j=1}^{d} p_i p_j A_{ij} = \sum_{i,j=1}^{d} p_i p_j \sum_{h\in\sym_d} r_{h} \delta_{j,h(i)} = \sum_{h\in\sym_d} r_{h} \sum_{i=1}^{d} p_i p_{h(i)}\text{.}
\end{equation}
The convex sum of numbers can be upper and lower bounded by their maximal and minimal value. Thus, we have:
\begin{equation}
    \min_{h\in\sym_d} \sum_{i=1}^{d} p_i p_{h(i)} \leq\sum_{h\in\sym_d} r_{h} \sum_{i=1}^{d} p_i p_{h(i)} \leq \max_{h\in\sym_d} \sum_{i=1}^{d} p_i p_{h(i)}\text{.}
\end{equation}
Now we can use the rearrangement inequalities: without loss of generality, we can suppose that the numbers $\{p_{i}\}_{i=1}^{d}$ are in a non-decreasing order, that is, $p_1\leq p_2\leq \ldots \leq p_d$\footnote{If not, they can always be rearranged.}. This leads to the following upper and lower bounds:
\begin{equation}
    \frac{1}{d}\sum_{i=1}^{d} p_i p_{d+1-i} \leq b(\sigma) \leq \frac{1}{d}\sum_{i=1}^{d} p_i^2=\frac{1}{d}f\text{,}
\end{equation}
where on the right-hand side the sum is exactly $f/d$.

Note that these upper and lower bounds are reached. For the upper bound, it suffices to choose a basis that is invariant to complex conjugation, that is, $\overline{\psi_i}=\psi_i$. For the lower bound, one has to find a basis that has the following property $\overline{\psi_i}=\psi_{d+1-i}$. This is always possible.
An example if $d$ is even:
\begin{equation}
    \psi_{j}\coloneqq
    \begin{cases}
    \frac{1}{\sqrt{2}}(e_j+ie_{d+1-j})\text{,} &\text{if $j\leq \frac{d}{2}$,} \\ \\ 
    \overline{\psi_{d+1-j}}\text{,} &\text{if $j>\frac{d}{2}$.}
    \end{cases}
\end{equation}
An example if $d$ is odd:
\begin{equation}
    \psi_{j}\coloneqq
    \begin{cases}\
    \frac{1}{\sqrt{2}}(e_j+ie_{d+1-j})\text{,} &\text{if $j\leq \frac{d-1}{2}$,}\\ \\
    e_j\text{,} &\text{if $j=\frac{d-1}{2}+1$,}\\ \\
    \overline{\psi_{d+1-j}}\text{,} &\text{if $j>\frac{d-1}{2}+1$.}
    \end{cases}
\end{equation}
Note that the index $j$ has been used to differentiate between it and the imaginary unit $i$. It is easy to check that these form orthonormal bases.

Using these tight bounds, we are ready to express the maximal and minimal $b$ as functions of $f$ in the next two sections.

\subsection{The maximal $b$ parameter as a function of $f$}

In light of the previous section, it is trivial to find the maximal $b$ parameter as a function of $f$:
\begin{equation}
    b_{\max}(f)=\frac{f}{d}\text{.}
\end{equation}

\subsection{The minimal $b$ parameter as a function of $f$}

Unlike the maximal case, finding the minimal $b$ parameter as a function of $f$ is much more involved. We know that given a non-decreasing set of probabilities $\{p_i\}_{i=1}^{d}$, we have:
\begin{equation}
    b_{\min}(\{p_i\}_{i=1}^{d})=\frac{1}{d}\sum_{i=1}^{d} p_i p_{d+1-i}\text{,}
\end{equation}
while we also have:
\begin{equation}
    f(\{p_i\}_{i=1}^{d})=\sum_{i=1}^{d} p_i^2\text{,}
\end{equation}
with $f\in[1/d,1]$.

There will be two cases depending on whether $d$ is even or odd.

\subsubsection{When $d$ is even:}

Note that $b_{\min}(\{p_i\}_{i=1}^{d})\equiv 0$ is achievable when $p_{d/2}=0$, and consequently, all probabilities with lower indices are also zero. This is the solution when $f\in[2/d,1]$. The solution for the other part of the range is a bit more involved. Let us introduce the following variables:
\begin{equation}
    u_i\coloneqq p_i + p_{d+1-i}\text{,} \qquad i\in[1,\ldots,d/2]\text{.}
\end{equation}
Note that the sum of $u_i$ is 1. Examine the following:
\begin{equation}
    \sum_{i=1}^{d/2} u_i^2 = f(\{p_i\}_{i=1}^{d})+db(\{p_i\}_{i=1}^{d})\text{.}
\end{equation}
Now, let us use the quadratic-arithmetic mean inequality on the left-hand side:
\begin{align}
    \frac{2}{d} \sum_{i=1}^{d/2} u_i &\leq \sqrt{\frac{2}{d}\sum_{i=1}^{d/2} u_i^2}\text{,}\\
    \frac{4}{d^2} &\leq \frac{2}{d}\sum_{i=1}^{d/2} u_i^2 = \frac{2}{d}\left[ f(\{p_i\}_{i=1}^{d}) + d b(\{p_i\}_{i=1}^{d})\right]\text{,}\\
    \frac{2}{d^2} - \frac{f(\{p_i\}_{i=1}^{d})}{d} &\leq b(\{p_i\}_{i=1}^{d})\text{,}
\end{align}
where, in the second step, we have squared the inequality and used that the sum of $u_i$ is 1. Note that, naturally, when $f>2/d$, the left-hand side is zero and the lower bound is no longer tight. For $f\leq 2/d$, equality can be achieved when $u_i=u_j$ for all $i,j$. From the requirement that probabilities sum to 1 we have:
\begin{equation}
    u_i=u_j=\frac{2}{d}\text{.}
\end{equation}
Leading to:
\begin{equation}
    f(\{p_i\}_{i=1}^{d})=\sum_{i=1}^{\frac{d}{2}} \left(2p_i^2-\frac{4}{d}p_i+\frac{4}{d^2}\right)\text{,}
\end{equation}
which is minimal when $p_i=1/d$ for all $i$ and maximal when $p_i=0$ for all $i$, leading to the range $f\in[1/d,2/d]$.
Thus, the solution is:
\begin{equation}
    b_{\min}^{d\text{ even}}(f)=\begin{cases} \frac{2-df}{d^2}\text{,} & \text{if $f\in\left[\frac{1}{d},\frac{2}{d}\right]$,} \\
    0\text{,} & \text{if $f\in\left[\frac{2}{d},1\right]$.} \end{cases}
\end{equation}

\subsubsection{When $d$ is odd:}

Note again that $b_{\min}(\{p_i\}_{i=1}^{d})\equiv 0$ is achievable when $p_{(d+1)/2}=0$, and consequently, all probabilities with lower indices are also zero. This is the solution when $f\in[2/(d-1),1]$.

A similar reasoning as before brings us to the rest of the solution. Let us introduce the following variables:
\begin{align}
    u_i&\coloneqq p_i + p_{d+1-i}\text{,} \qquad i\in[1,\ldots, (d-1)/2]\text{,}\\
    v&\coloneqq p_{(d+1)/2}\text{.}
\end{align}
Note the difference with the even case being that one of the probabilities is unpaired. We have:
\begin{equation}
    2v^2 + \sum_{i=1}^{(d-1)/2} u_i^2= f(\{p_i\}_{i=1}^{d})+db(\{p_i\}_{i=1}^{d})\text{.}
\end{equation}
Using the previous reasoning of the means leads us to:
\begin{equation}\label{eq:lower_odd}
    \frac{2-4v+2dv^2}{d(d-1)}-\frac{f(\{p_i\}_{i=1}^{d})}{d} \leq b(\{p_i\}_{i=1}^{d})\text{.}
\end{equation}
Equality is achieved when $u_i=u_j$ for all $i,j$. From the requirement that probabilities sum to 1 we have:
\begin{equation}
    u_i=u_j=\frac{2(1-v)}{d-1}\text{.}
\end{equation}
Note, however, that the lower bound is not yet manifest and some minimisation is left to be done. Let us imagine that $v$ and $f$ are independent of each other.

First, let us assume that $f(\{p_i\}_{i=1}^{d})$ is fixed. Then the left-hand side of Equation \eqref{eq:lower_odd} is minimal if $v=1/d$ is possible. This is possible if $f\in[1/d,(2d-1)/d^2]$ which is obtained by a similar reasoning as before: the lower limit is attainable when all probabilities are $1/d$ and the upper limit is attainable when all probabilities with a lower index than $(d+1)/2$ are zero.

Second, let us assume that $0\leq v \leq 1/d$ is fixed. Then the left-hand side of Equation \eqref{eq:lower_odd} is minimal when $f(\{p_i\}_{i=1}^{d})$ is maximal, that is, when all probabilities with a smaller index than $(d+1)/2$ are zero, leading to a range of $f$ in $[(2d-1)/d^2, 2/(d-1)]$, and the following expression for $v$ and $b$:
\begin{equation}
    v=\frac{2-\sqrt{(d+1)(df+f-2)}}{d+1}\text{,}
\end{equation}

\begin{equation}
    \frac{v^2}{d}=\frac{\left(2-\sqrt{(d+1)(df+f-2)}\right)^2}{d(d+1)^2}=b(\{p_i\}_{i=1}^{d})\text{.}
\end{equation}

Third, when $f>2/(d-1)$ the optimal solution is already known to be $b_{\min}(\{p_i\}_{i=1}^{d})\equiv 0$.

Given how these ranges in $v$ and $f$ fit together, one concludes that the solution is simply their concatenation on the different ranges:
\begin{equation}
    b_{\min}^{d\text{ odd}}(f)=\begin{cases} \frac{2-df}{d^2}\text{,} & \text{if $f\in\left[\frac{1}{d},\frac{2d-1}{d^2}\right]$,} \\
     \frac{\left(2-\sqrt{(d^2-1)f-2(d-1)}\right)^2}{d(d+1)^2}\text{,} & \text{if $f\in\left[\frac{2d-1}{d^2},\frac{2}{d-1}\right]$,} \vspace{0.2cm} \\
    0\text{,} & \text{if $f\in\left[\frac{2}{d-1},1\right]$.} \end{cases}
\end{equation}

\end{appendices}

\end{document}